\documentclass{article}
\usepackage{amsmath}
\usepackage{amsthm}
\usepackage{amssymb}
\usepackage{amsfonts}
\begin{document}
\newtheorem{thm}{Theorem}[section]
\newtheorem{lem}[thm]{Lemma}
\newtheorem{prop}[thm]{Proposition}
\newtheorem{cor}[thm]{Corollary}
\theoremstyle{definition}
\newtheorem{assum}[thm]{Assumption}
\newtheorem{notation}[thm]{Notation}
\newtheorem{defn}[thm]{Definition}
\newtheorem{clm}[thm]{Claim}
\newtheorem{ex}[thm]{Example}
\theoremstyle{remark}
\newtheorem{rem}[thm]{Remark}
\newcommand{\unit}{\mathbb I}
\newcommand{\ali}[1]{{\mathfrak A}_{[ #1 ,\infty)}}
\newcommand{\alm}[1]{{\mathfrak A}_{(-\infty, #1 ]}}
\newcommand{\nn}[1]{\lV #1 \rV}
\newcommand{\br}{{\mathbb R}}
\newcommand{\dm}{{\rm dom}\mu}
\newcommand{\Ad}{\mathop{\mathrm{Ad}}\nolimits}
\newcommand{\Proj}{\mathop{\mathrm{Proj}}\nolimits}
\newcommand{\RRe}{\mathop{\mathrm{Re}}\nolimits}
\newcommand{\RIm}{\mathop{\mathrm{Im}}\nolimits}
\newcommand{\Wo}{\mathop{\mathrm{Wo}}\nolimits}
\newcommand{\Prim}{\mathop{\mathrm{Prim}_1}\nolimits}
\newcommand{\Primz}{\mathop{\mathrm{Prim}}\nolimits}
\newcommand{\Class}{\mathop{\mathrm{Class}}\nolimits}
\newcommand{\ClassA}{\mathop{\mathrm{ClassA}}\nolimits}
\def\qed{{\unskip\nobreak\hfil\penalty50
\hskip2em\hbox{}\nobreak\hfil$\square$
\parfillskip=0pt \finalhyphendemerits=0\par}\medskip}
\def\proof{\trivlist \item[\hskip \labelsep{\bf Proof.\ }]}
\def\endproof{\null\hfill\qed\endtrivlist\noindent}
\def\proofof[#1]{\trivlist \item[\hskip \labelsep{\bf Proof of #1.\ }]}
\def\endproofof{\null\hfill\qed\endtrivlist\noindent}
\newcommand{\caA}{{\mathcal A}}
\newcommand{\caB}{{\mathcal B}}
\newcommand{\caC}{{\mathcal C}}
\newcommand{\caD}{{\mathcal D}}
\newcommand{\caE}{{\mathcal E}}
\newcommand{\caF}{{\mathcal F}}
\newcommand{\caG}{{\mathcal G}}
\newcommand{\caH}{{\mathcal H}}
\newcommand{\caI}{{\mathcal I}}
\newcommand{\caJ}{{\mathcal J}}
\newcommand{\caK}{{\mathcal K}}
\newcommand{\caL}{{\mathcal L}}
\newcommand{\caM}{{\mathcal M}}
\newcommand{\caN}{{\mathcal N}}
\newcommand{\caO}{{\mathcal O}}
\newcommand{\caP}{{\mathcal P}}
\newcommand{\caQ}{{\mathcal Q}}
\newcommand{\caR}{{\mathcal R}}
\newcommand{\caS}{{\mathcal S}}
\newcommand{\caT}{{\mathcal T}}
\newcommand{\caU}{{\mathcal U}}
\newcommand{\caV}{{\mathcal V}}
\newcommand{\caW}{{\mathcal W}}
\newcommand{\caX}{{\mathcal X}}
\newcommand{\caY}{{\mathcal Y}}
\newcommand{\caZ}{{\mathcal Z}}
\newcommand{\bbA}{{\mathbb A}}
\newcommand{\bbB}{{\mathbb B}}
\newcommand{\bbC}{{\mathbb C}}
\newcommand{\bbD}{{\mathbb D}}
\newcommand{\bbE}{{\mathbb E}}
\newcommand{\bbF}{{\mathbb F}}
\newcommand{\bbG}{{\mathbb G}}
\newcommand{\bbH}{{\mathbb H}}
\newcommand{\bbI}{{\mathbb I}}
\newcommand{\bbJ}{{\mathbb J}}
\newcommand{\bbK}{{\mathbb K}}
\newcommand{\bbL}{{\mathbb L}}
\newcommand{\bbM}{{\mathbb M}}
\newcommand{\bbN}{{\mathbb N}}
\newcommand{\bbO}{{\mathbb O}}
\newcommand{\bbP}{{\mathbb P}}
\newcommand{\bbQ}{{\mathbb Q}}
\newcommand{\bbR}{{\mathbb R}}
\newcommand{\bbS}{{\mathbb S}}
\newcommand{\bbT}{{\mathbb T}}
\newcommand{\bbU}{{\mathbb U}}
\newcommand{\bbV}{{\mathbb V}}
\newcommand{\bbW}{{\mathbb W}}
\newcommand{\bbX}{{\mathbb X}}
\newcommand{\bbY}{{\mathbb Y}}
\newcommand{\bbZ}{{\mathbb Z}}
\newcommand{\str}{^*}
\newcommand{\lv}{\left \vert}
\newcommand{\rv}{\right \vert}
\newcommand{\lV}{\left \Vert}
\newcommand{\rV}{\right \Vert}
\newcommand{\la}{\left \langle}
\newcommand{\ra}{\right \rangle}
\newcommand{\ltm}{\left \{}
\newcommand{\rtm}{\right \}}
\newcommand{\lcm}{\left [}
\newcommand{\rcm}{\right ]}
\newcommand{\ket}[1]{\lv #1 \ra}
\newcommand{\bra}[1]{\la #1 \rv}
\newcommand{\lmk}{\left (}
\newcommand{\rmk}{\right )}
\newcommand{\al}{{\mathcal A}}
\newcommand{\md}{M_d({\mathbb C})}
\newcommand{\Tr}{\mathop{\mathrm{Tr}}\nolimits}
\newcommand{\Ran}{\mathop{\mathrm{Ran}}\nolimits}
\newcommand{\Ker}{\mathop{\mathrm{Ker}}\nolimits}
\newcommand{\spn}{\mathop{\mathrm{span}}\nolimits}
\newcommand{\Mat}{\mathop{\mathrm{M}}\nolimits}
\newcommand{\UT}{\mathop{\mathrm{UT}}\nolimits}
\newcommand{\GL}{\mathop{\mathrm{GL}}\nolimits}
\newcommand{\spa}{\mathop{\mathrm{span}}\nolimits}
\newcommand{\supp}{\mathop{\mathrm{supp}}\nolimits}
\newcommand{\rank}{\mathop{\mathrm{rank}}\nolimits}
\newcommand{\id}{\mathop{\mathrm{id}}\nolimits}
\newcommand{\idd}{\mathop{\mathrm{id}}\nolimits}
\newcommand{\ran}{\mathop{\mathrm{Ran}}\nolimits}
\newcommand{\dr}{ \mathop{\mathrm{d}_{{\mathbb R}^k}}\nolimits} 
\newcommand{\dc}{ \mathop{\mathrm{d}_{\cc}}\nolimits} \newcommand{\drr}{ \mathop{\mathrm{d}_{\rr}}\nolimits} 
\newcommand{\hee}{\widehat{e_{11}^{(n_0)}}}
\newcommand{\zin}{\mathbb{Z}}
\newcommand{\rr}{\mathbb{R}}
\newcommand{\cc}{\mathbb{C}}
\newcommand{\ww}{\mathbb{W}}
\newcommand{\nan}{\mathbb{N}}\newcommand{\bb}{\mathbb{B}}
\newcommand{\aaa}{\mathbb{A}}\newcommand{\ee}{\mathbb{E}}
\newcommand{\pp}{\mathbb{P}}
\newcommand{\wks}{\mathop{\mathrm{wk^*-}}\nolimits}
\newcommand{\he}{\hat {\mathbb E}}
\newcommand{\ikn}{{\caI}_{k,n}}
\newcommand{\mk}{{\Mat_k}}
\newcommand{\mnz}{\Mat_{n_0}}
\newcommand{\mn}{\Mat_{n}}
\newcommand{\mkk}{\Mat_{k_R+k_L+1}}
\newcommand{\mnzk}{\mnz\otimes \mkk}
\newcommand{\hbb}{H^{k,\bb}_{m,p,q}}
\newcommand{\gb}[1]{\Gamma_{#1,\bb}}
\newcommand{\cgv}[1]{\caG_{#1,\vv}}
\newcommand{\gv}[1]{\Gamma_{#1,\vv}}
\newcommand{\ga}[1]{\Gamma_{#1,\bbA}}
\newcommand{\gvtr}[1]{\Gamma_{#1,\vv(t)}^{(R)}}
\newcommand{\gbtr}[1]{\Gamma_{#1,\bb(t)}^{(R)}}
\newcommand{\cgb}[1]{\caG_{#1,\bb}}
\newcommand{\cgbt}[1]{\caG_{#1,\bb(t)}}
\newcommand{\gvp}[1]{G_{#1,\vv}}
\newcommand{\gbp}[1]{G_{#1,\bb}}
\newcommand{\gbpt}[1]{G_{#1,\bb(t)}}
\newcommand{\Pbm}[1]{\Phi_{#1,\bb}}
\newcommand{\Pvm}[1]{\Phi_{#1,\bb}}
\newcommand{\mb}{m_{\bb}}
\newcommand{\hpuk}{\hat P^{(n_{0,\bbV},k_{R,\bbK},k_{L\bbV})}_R}
\newcommand{\hpdk}{\hat P^{(n_0,k_{R\bbV},k_{L\bbV})}_L}
\newcommand{\E}[1]{\widehat{\mathbb{E}}^{(#1)}}
\newcommand{\lal}{{\boldsymbol\lambda}}
\newcommand{\vpl}{{\boldsymbol\varpi}}
\newcommand{\oo}{{\boldsymbol\omega}}
\newcommand{\uu}{{\boldsymbol u}}
\newcommand{\uup}{{\boldsymbol \upsilon}}
\newcommand{\vv}{{\boldsymbol v}}
\newcommand{\bbm}{{\boldsymbol m}}
\newcommand{\kl}[1]{{\mathcal K}_{#1}}
\newcommand{\wb}[1]{\widehat{B_{\mu^{(#1)}}}}
\newcommand{\wa}[1]{\widehat{A_{\mu^{(#1)}}}}
\newcommand{\ws}[1]{\widehat{\psi_{\mu^{(#1)}}}}
\newcommand{\wsn}[1]{\widehat{\psi_{\nu^{(#1)}}}}
\newcommand{\wv}[1]{\widehat{v_{\mu^{(#1)}}}}
\newcommand{\wbn}[1]{\widehat{B_{\nu^{(#1)}}}}
\newcommand{\wan}[1]{\widehat{A_{\nu^{(#1)}}}}
\newcommand{\wo}[1]{\widehat{\omega_{\mu^{(#1)}}}}
\newcommand{\dist}{\dc}
\newcommand{\hpu}{\hat P^{(n_0,k_R,k_L)}_R}
\newcommand{\hpd}{\hat P^{(n_0,k_R,k_L)}_L}
\newcommand{\pu}{ P^{(k_R,k_L)}_R}
\newcommand{\pd}{ P^{(k_R,k_L)}_L}
\newcommand{\pzu}{ P^{(k_R,0)}_R}
\newcommand{\pzd}{ P^{(0,k_L)}_L}
\newcommand{\puuz}{P_{R,\bbA}}
\newcommand{\pddz}{P_{L,\bbA}}
\newcommand{\puu}{\tilde P_R}
\newcommand{\pdd}{\tilde P_L}
\newcommand{\qu}[1]{ Q^{(k_R,k_L)}_{R, #1}}
\newcommand{\qd}[1]{ Q^{(k_R,k_L)}_{L,#1}}
\newcommand{\qur}[1]{ Q^{(k_R,0)}_{R, #1}}
\newcommand{\qdl}[1]{ Q^{(0,k_L)}_{L,#1}}
\newcommand{\hqu}[1]{ \hat Q^{(n_0,k_R,k_L)}_{R, #1}}
\newcommand{\hqd}[1]{ \hat Q^{(n_0,k_R,k_L)}_{L,#1}}
\newcommand{\eij}[1] {E^{(k_R,k_L)}_{#1}}
\newcommand{\eijz}[1] {E^{(n_0-1,n_0-1)}_{#1}}
\newcommand{\heij}[1] {\hat E^{(k_R,k_L)}_{#1}}
\newcommand{\cn}{\mathop{\mathrm{CN}(n_0,k_R,k_L)}\nolimits}
\newcommand{\ghd}[1]{\mathop{\mathrm{GHL}(#1,n_0,k_R,k_L,\bbG)}\nolimits}
\newcommand{\ghu}[1]{\mathop{\mathrm{GHR}(#1,n_0,k_R,k_L,\bbD)}\nolimits}
\newcommand{\ghdb}[1]{\mathop{\mathrm{GHL}(#1,n_0,k_R,k_L,\bbG)}\nolimits}
\newcommand{\ghub}[1]{\mathop{\mathrm{GHR}(#1,n_0,k_R,k_L,\bbD)}\nolimits}
\newcommand{\hfu}[1]{{\mathfrak H}_{#1}^R}
\newcommand{\hfd}[1]{{\mathfrak H}_{#1}^L}
\newcommand{\hfui}[1]{{\mathfrak H}_{#1,1}^R}
\newcommand{\hfdi}[1]{{\mathfrak H}_{#1,1}^L}
\newcommand{\hfuz}[1]{{\mathfrak H}_{#1,0}^R}
\newcommand{\hfdz}[1]{{\mathfrak H}_{#1,0}^L}
\newcommand{\CN}{\overline{\hpd}\lmk\mnzk \rmk\overline{\hpu}}
\newcommand{\cnz}[1] {\chi_{#1}^{(n_0)}}
\newcommand{\eu}{\eta_{R}^{(k_R,k_L)}}
\newcommand{\ezu}{\eta_{R}^{(n_0-1,n_0-1)}}
\newcommand{\ed}{\eta_{L}^{(k_R,k_L)}}
\newcommand{\ezd}{\eta_{L}^{(n_0-1,n_0-1)}}
\newcommand{\fii}[1]{f_{#1}^{(k_R,k_L)}}
\newcommand{\fiz}[1]{f_{#1}^{(n_0-1,n_0-1)}}
\newcommand{\gaa}[2]{g_{#1,#2,\bbA}}
\newcommand{\zeij}[1] {e_{#1}^{(n_0)}}
\newcommand{\CL}{\Class(n,n_0,k_R,k_L)}
\newcommand{\CLn}{\Class_2(n,n_0,k_R,k_L)}
\newcommand{\braket}[2]{\left\langle#1,#2\right\rangle}
\newcommand{\abs}[1]{\left\vert#1\right\vert}
\newcommand{\puz}{ P^{(n_0-1,n_0-1)}_R}
\newcommand{\pdz}{ P^{(n_0-1,n_0-1)}_L}
\newcommand{\ir}{I_{R}^{(k_R,k_L)}}
\newcommand{\il}{I_{L}^{(k_R,k_L)}}
\newcommand{\eijr}[1] {E^{(k_R,0)}_{#1}}
\newcommand{\eijl}[1] {E^{(0,k_L)}_{#1}}
\newcommand{\prr}{ P^{(k_R,0)}_R}
\newcommand{\pll}{ P^{(0,k_L)}_L}
\newtheorem{nota}{Notation}[section]
\def\qed{{\unskip\nobreak\hfil\penalty50
\hskip2em\hbox{}\nobreak\hfil$\square$
\parfillskip=0pt \finalhyphendemerits=0\par}\medskip}
\def\proof{\trivlist \item[\hskip \labelsep{\bf Proof.\ }]}
\def\endproof{\null\hfill\qed\endtrivlist\noindent}
\def\proofof[#1]{\trivlist \item[\hskip \labelsep{\bf Proof of #1.\ }]}
\def\endproofof{\null\hfill\qed\endtrivlist\noindent}
\title{A class of asymmetric gapped Hamiltonians on quantum spin chains and its characterization III}
\author{
{\sc Yoshiko Ogata}\footnote{Supported in part by
the Grants-in-Aid for
Scientific Research, JSPS.}\\
{\small Graduate School of Mathematical Sciences}\\
{\small The University of Tokyo, Komaba, Tokyo, 153-8914, Japan}
}

\maketitle{}
\centerline{\sl }
\begin{abstract}
In this paper, we consider classification problem of asymmetric gapped Hamiltonians, which are given as the non-degenerate part of the Hamiltonians introduced in \cite{Ogata1}.
We consider the $C^1$-classification, which takes into account the effect of boundaries.
We show that the left and right degeneracies of edge ground states are the complete invariant.
As a corollary, we consider the bulk-classification problem.
We study 
Hamiltonians that1.are given by translation invariant finite range interactions, 2.are gapped in the bulk, 3.are
frustration-free, 4.have uniformly bounded ground state degeneracy on finite intervals, and
5.have a unique bulk ground state.
We show that for the bulk-classification, any such Hamiltonians are equivalent.

\end{abstract}

\section{Introduction}

In Part I \cite{Ogata1}, we introduced a class of gapped Hamiltonians on quantum spin chains,
which allows asymmetric edge ground states.
They are defined as MPS (matrix product state) Hamiltonians given by $n$-tuple of matrices in $\ClassA$, a class we introduced in \cite{Ogata1}.
It is an asymmetric generalization of the class of Hamiltonians given in 
\cite{Fannes:1992vq}. 
We investigated the properties of this new class in Part I \cite{Ogata1}.
In particular, we showed that Hamiltonians in this class satisfies five qualitative physical conditions, which are listed as [A1]-[A5] in \cite{Ogata2}. 
We call $\tilde H$, the set of Hamiltonians satisfying [A1]-[A5].
In Part II \cite{Ogata2}, conversely,
we showed that these five properties [A1]-[A5]
actually guarantee the ground state structure of the Hamiltonian to be 
captured by the MPS Hamiltonians given by $\ClassA$.
More precisely, we showed for any Hamiltonian in $\tilde H$, there is an MPS Hamiltonian given by $\ClassA$
satisfying the followings: {\it 1.} The ground state spaces  of the two Hamiltonians on the infinite intervals coincide.
{\it 2.} The spectral projections onto the ground state space of the original Hamiltonian on finite intervals
are well approximated by that of the MPS one. 
From the latter property we see that two Hamiltonians are in the same class in the classification problem of gapped Hamiltonians with open boundary conditions,
(the type II-$C^1$-classification).
Hence the classification problem of $\tilde H$ is reduced to the classification problem of 
MPS Hamiltonians given by $\ClassA$.
In this paper, we classify MPS Hamiltonians given by $\ClassA'$, where $\ClassA'$ is the ``non-degenerate'' part of $\ClassA$.
We show that the left and right degeneracies of edge ground states are the complete invariant for the type II-$C^1$-classification, for this class.

As an important application, we consider bulk-classification problem of gapped Hamiltonians.
It is believed that there exists only one gapped bulk ground state phase in one dimensional quantum spin systems.
It has been an open problem for a while, if two MPS-Hamiltonians from the class of \cite{Fannes:1992vq} can be connected to each other without closing the gap {\it in the bulk}, and 
{\it without breaking the translation invariance} (\cite{bo},\cite{SW}).
We solve this affirmatively.
More generally, we consider
Hamiltonians which 1. are given by translation invariant finite range interactions, 2. are gapped in the bulk, 3. are
frustration-free, 4. have uniformly bounded ground state degeneracy on finite intervals, and
5. have a unique bulk ground state,.
We show that for the bulk-classification, any such Hamiltonians are equivalent.

Throughout this article, $2\le n\in\nan$ is fixed as the dimension of the spin under consideration.
We use freely the notations and definitions given in Part I \cite{Ogata1} Subsection 1.1, 1.2, 1.3, and Appendix A.

\subsection{$\Class A'$ and its type II-$C^1$-classification}For $m\in\nan$, we denote by $\caJ_m$ the set of all positive translation invariant
interactions with interaction length less than or equal to $m$.
We also set $\caJ:=\cup_{m\in\nan} \caJ_m$.
A natural number $m\in\nan$ and an element $h\in\caA_{[0,m-1]}$,
define an interaction $\Phi_h$ by
\begin{align}\label{hamdef}
\Phi_h(X):=\left\{
\begin{gathered}
\tau_x\lmk h\rmk,\quad \text{if}\quad  X=[x,x+m-1] \quad \text{for some}\quad  x\in\bbZ\\
0,\quad\text{otherwise}
\end{gathered}\right.
\end{align}
for $X\in {\mathfrak S}_{\bbZ}$.
A Hamiltonian associated with $\Phi$ is a net of self-adjoint operators $H_{\Phi}:=\left((H_{\Phi })_\Lambda\right)_{\Lambda\in{\mathfrak I}_{\bbZ}}$ such that 
\begin{equation}\label{GenHamiltonian}
\lmk H_{\Phi}\rmk_{\Lambda}:=\sum_{X\subset{\Lambda}}\Phi(X).
\end{equation}
First we confirm what we mean by gapped Hamiltonian with open boundary conditions.
\begin{defn}
Let $\Phi$ be a translation invariant finite range interaction on $\caA_{\bbZ}$.
For each $N\in\nan$, let $G_{[0,N-1]}$ be the spectral projection of
$\lmk H_{\Phi}\rmk_{[0,N-1]}$  corresponding to the lowest eigenvalue $\inf\sigma\lmk \lmk H_{\Phi}\rmk_{[0,N-1]}\rmk$.
We say that the Hamiltonian $H_\Phi$ is {gapped with respect to the open boundary condition}  if there exist $\gamma>0$ and $N_0\in\nan$ such that
\[
\inf\sigma\lmk \lmk H_{\Phi}\rmk_{[0,N-1]}\rmk +\gamma\lmk \unit-G_{[0,N-1]]}\rmk \le
\lmk H_{\Phi}\rmk_{[0,N-1]},\quad
N_0\le N\in\nan.
\]
We call this $\gamma$, a lower bound of the gap.
\end{defn}

Let us recall the definition of $C^1$-classification I, II.
\begin{defn}[$C^1$-classification I of gapped Hamiltonians]\label{def:phafst}
Let $H_0,H_1$ be Hamiltonians gapped with respect to the open boundary conditions, associated with interactions
$\Phi_{0},\Phi_{1}\in{\caJ}$.
We say that  $H_0,H_1$ are 
typeI-$C^1$-equivalent if the following conditions are satisfied.
\begin{enumerate}
\item There exist an $m\in\nan$ and a continuous and piecewise $C^1$-path $\Phi:[0,1]\to {\caJ}_m$ such that $\Phi(0)=\Phi_0$, 
$\Phi(1)= \Phi_1$.
\item  Let   $H(t)$ be the Hamiltonian associated with $\Phi(t)$ for each $t\in[0,1]$.
There are $\gamma>0$, $N_0\in\nan$, and finite intervals $I(t)=[a(t), b(t)]$, whose endpoints $a(t), b(t)$ smoothly depending on $t\in[0,1]$,
such that for any $N_0\le N\in\nan$,
the smallest eigenvalue of $H(t)_{[0,N-1]}$ is in $I(t)$ and
the rest of the spectrum is in $[b(t)+\gamma,\infty)$.
\end{enumerate}
We write $H_0\simeq _IH_1$ when $H_0,H_1$ are 
type I-$C^1$-equivalent.
\end{defn}
\begin{defn}[$C^1$-classification II of gapped Hamiltonians]\label{def:phasec}
Let $H_0,H_1$ be Hamiltonians gapped with respect to the open boundary conditions, associated with interactions
$\Phi_{0},\Phi_{1}\in{\caJ}$.
We say that  $H_0,H_1$ are 
type II-$C^1$-equivalent if the following conditions are satisfied.
\begin{enumerate}
\item There exist $m\in\nan$ and a continuous and piecewise $C^1$-path $\Phi:[0,1]\to {\caJ}_m$ such that $\Phi(0)=\Phi_0$, 
$\Phi(1)= \Phi_1$. 
\item  Let $H(t)$ be the Hamiltonian associated with $\Phi(t)$ for each $t\in[0,1]$.
There are $\gamma>0$, $N_0\in\nan$, and finite intervals $I(t)=[a(t), b(t)]$, $t\in[0,1]$, satisfying the followings:
\begin{description}
\item[(i)] the endpoints $a(t), b(t)$ smoothly depend on $t\in[0,1]$, and
\item[(ii)]there exists a sequence $\{\varepsilon_N\}_{N\in\nan}$
of positive numbers with $\varepsilon_N\to 0$, for $N\to\infty$,
such that $\sigma\lmk H(t)_{[0,N-1]}\rmk\cap I(t)
=\sigma\lmk H(t)_{[0,N-1]}\rmk\cap 
[\lambda(t,N), \lambda(t,N)+\varepsilon_N]$,
and $\sigma\lmk H(t)_{[0,N-1]}\rmk\cap I(t)^c=
\sigma\lmk H(t)_{[0,N-1]}\rmk\cap [b(t)+\gamma,\infty)$
 for all $N\ge N_0$ and $t\in[0,1]$,
where $\lambda(t,N)$
is the smallest eigenvalue  of $H(t)_{[0,N-1]}$.
\end{description}
%
\end{enumerate}
We write 
$H_0\simeq_{II} H_1$, when $H_0$ and $H_1$ are $type II-C^1$-equivalent.
\end{defn}
In Part II \cite{Ogata2}, we discussed the type II-$C^1$-classification for a class of Hamiltonians $\tilde\caH$ satisfying five qualitative properties.
We showed that any element in $\tilde \caH$ is type II-$C^1$-equivalent to
an MPS Hamiltonian $H_{\Phi_{m,\bb}}$
which is given by some $\bb\in\ClassA$, the class which was introduced in Part I \cite{Ogata1}.
In other words, the problem of type II-$C^1$-classification for this family $\tilde \caH$  is reduced to the classification problem of
MPS Hamiltonians given by $\ClassA$.
In this paper, we classify Hamiltonians given by the ``non-degenerate'' part of $\ClassA$, which we call $\ClassA'$.

The subclass $\ClassA'$ is the set of elements of 
$\Class A$ with the non-degeneracy condition on $\lal$.
Namely we assume $\lal$ in Definition 1.14 Part I \cite{Ogata1} to be in
$\Wo'(K_R,k_L)$ defined in the following.
\begin{defn}\label{def:class}
For $k_R,k_L\in\nan\cup\{0\}$,
we denote by $\Wo'(k_R,k_L)$ the set of all
${\boldsymbol\lambda}=(\lambda_{-k_R},\ldots,\lambda_{-1},\lambda_0,\lambda_1,\ldots,\lambda_{k_L})\in
\cc^{k_R+k_L+1}$
satisfying
\begin{align*}
\lambda_0=1,\quad 0<\lv\lambda_{-k_R}\rv\le \lv\lambda_{-k_R+1}\rv\le \cdots
\lv\lambda_{-1}\rv<1,\quad 0<\lv\lambda_{k_L}\rv\le \lv\lambda_{k_L-1}\rv\le \cdots
\lv\lambda_{1}\rv<1,
\end{align*}
and $\lambda_i\neq\lambda_j$ if $-k_R\le i<j\le 0$ or $0\le i<j\le k_L$.
\end{defn}
The last condition in the definition corresponds to the non-degeneracy.
Recall the definitions in subsection 1.3 of \cite{Ogata1}.
In particular,  $\ClassA$ is defined as follows.
\begin{defn}
Let $n,n_0\in\nan$ with $n\ge 2$, $k_R,k_L\in \nan\cup\{0\}$ and
$(\lal,\bbD,\bbG,Y)\in\caT(k_R,k_L)$.
We denote by ${\mathfrak B}(n,n_0,k_R,k_L,\lal,\bbD,\bbG,Y)$
the set of all $n$-tuples $\bb=(B_1,\ldots,B_n)\in\mnz\otimes\lmk
 \caD(k_R,k_L,\bbD,\bbG)\Lambda_{\lal}\lmk 1+Y\rmk
\rmk$
satisfying 
\begin{align}\label{eq:lblb}
l_\bb=
l_{\bb}(n,n_0,k_R,k_L,\lal,\bbD,\bbG,Y)
:=\inf\left\{l\mid
\caK_{l'}(\bbB)=
\mnz\otimes\lmk
 \caD(k_R,k_L,\bbD,\bbG)\lmk \Lambda_{\lal}\lmk 1+Y\rmk\rmk^{l'}
\rmk
\text{ for all } l'\ge l
\right\}
<\infty,
\end{align}
and $r_{T_{\bb}}=1$.
We define $\ClassA$ by
\[
\ClassA:=\bigcup
\left\{ {\mathfrak B}(n,n_0,k_R,k_L,\lal,\bbD,\bbG,Y)
\mid n_0\in\nan, k_R,k_L\in \nan\cup\{0\}, (\lal,\bbD,\bbG,Y)\in\caT(k_R,k_L)
\right\}.
\]
\end{defn}We define $\ClassA'$, as the "non-degenerate" part of $\ClassA$.
The non-degeneracy is about $\lal$, i.e., $\lal$ belongs to $\Wo'(k_R,k_L)$, the non-degenerate part $\Wo(k_R,k_L)$.
\begin{defn}
We define $\ClassA'$ by
\begin{align*}
\Class A':=&\bigcup
\left\{ {\mathfrak B}(n,n_0,k_R,k_L,\lal,\bbD,\bbG,Y)
\mid n_0\in\nan, k_R,k_L\in\nan\cup\{0\},\;(\lal,\bbD,\bbG,Y)\in\caT(k_R,k_L),\lal\in\Wo'(k_R,k_L)
\right\}\\
=&\bigcup
\left\{ {\mathfrak B}(n,n_0,k_R,k_L,\lal,\bbD,\bbG,0)
\mid n_0\in\nan, k_R,k_L\in\nan\cup\{0\},\;(\lal,\bbD,\bbG,0)\in\caT(k_R,k_L),\lal\in\Wo'(k_R,k_L)
\right\}.
\end{align*}
\begin{rem}
The condition $\lal\in\Wo'(k_R,k_L)$ and $(\lal,\bbD,\bbG,Y)\in\caT(k_R,k_L)$ automatically implies $Y=0$
(Lemma \ref{lem:ad}).
This corresponds to the second inequality.
\end{rem}

Furthermore, for each $n_0\in\nan$ and $k_R,k_L\in\nan\cup\{0\}$, we set
\[
\Class(n,n_0,k_R,k_L):=\bigcup
\left\{ {\mathfrak B}(n,n_0,k_R,k_L,\lal,\bbD,\bbG,0)
\mid(\lal,\bbD,\bbG,0)\in\caT(k_R,k_L),\lal\in\Wo'(k_R,k_L)
\right\}.
\]
\end{defn}
We denote by $\caH(n,n_0,k_R,k_L) $, the set of Hamiltonians $H_{\Phi_{m,\bbB}}$ given by
$\bbB\in \Class(n,n_0,k_R,k_L)$, and $m\ge 2l_\bbB$.
Recall from Part I \cite{Ogata1} Theorem 1.18 (vi) that $n_0(k_R+1)$ (resp. $n_0(k_L+1)$) is the degree of ground state degeneracy on right (resp. left) half infinite chain.
The following theorem says that these numbers are complete invariant of the type II $C^1$-classification of MPS Hamiltonians given by $\ClassA'$.
\begin{thm}\label{singthm}
Let $n_0, n_0'\in\nan$, $k_R,k_L,k_R',k_L'\in\nan\cup\{0\}$,
 $\bb\in\Class(n,n_0,k_R,k_L)$, and
$\bb'\in\Class(n,n_0',k_R',k_L')$.
Then for any $m\ge 2l_\bbB$, and
$m'\ge 2l_{\bbB'}$,
we have
$H_{\Phi_{m,\bb}}\simeq_{II} H_{\Phi_{m',\bb'}}$
if and only of
\[
n_0(k_R+1)=n_0'(k_R'+1),\;\text{and}\;\;
n_0(k_L+1)=n_0'(k_L'+1).
\]
\end{thm}

\subsection{Bulk Classification}

Let us recall the definition of a ground state of $C^*$-dynamical system \cite{BR2}.
Let $\mathfrak A$ be a unital $C^*$-algebra and $\alpha$ a strongly continuous one parameter group on $\mathfrak A$.
We denote the generator of $\alpha$ by $\delta_{\alpha}$.
A state $\omega$ on $\mathfrak A$ is called an $\alpha$-ground state
if the inequality
$
-i\omega\lmk A^*{\delta_\alpha}\lmk A\rmk\rmk\ge 0
$
holds
for any element $A$ in the domain $\caD({\delta_\alpha})$ of ${\delta_\alpha}$.
Let $\omega$ be an $\alpha$-ground state, with the GNS triple $(\caH,\pi,\Omega)$.
Then there exists a unique positive operator $H_{\omega,\alpha}$ on $\caH$ such that
$e^{itH_{\omega,\alpha}}\pi\lmk A\rmk\Omega=\pi\lmk \alpha_t\lmk A\rmk\rmk\Omega$,
for all $A\in\mathfrak A$ and $t\in\mathbb R$.
Note that $\Omega$ is an eigenvector of $H_{\omega,\alpha}$ with eigenvalue $0$.
We say that $H_{\omega,\alpha}$ has a spectral gap if $0$ is a non-degenerate eigenvalue of $H_{\omega,\alpha}$ and
there exists  
 $\gamma>0$ such that
$\sigma(H_{\omega,\alpha})\setminus\{0\}\subset [\gamma,\infty)$.
(Here, $\sigma(H_{\omega,\alpha})$ denotes the spectrum of $H_{\omega,\alpha}$.)

We consider $C^*$-dynamical systems given by translation invariant finite range interactions on the quantum spin chain $\caA_{\bbZ}$. Let $\Phi$ be a  translation invariant finite range interaction
.
Our $C^*$-dynamics is given
by
$\alpha_{\Phi,t}\lmk A\rmk:=\lim_{\Lambda\to\bbZ}e^{it(H_\Phi)_\Lambda} A e^{-it(H_\Phi)_\Lambda}$,
$A\in \caA_{\bbZ}$, $t\in\mathbb R$ with $(H_\Phi)_\Lambda$ in
(\ref{GenHamiltonian}).
We denote the set of all $\alpha_\Phi$-ground states on 
$\caA_\bbZ$ by ${\mathfrak B}_{\Phi}$. 
Recall $\caS_{\bbZ}(H_\Phi)$ from subsection 1.1 of \cite{Ogata1}.
It is known that $\emptyset\neq \caS_{\bbZ}(H_\Phi)\subset {\mathfrak B}_{\Phi}$ (See Proposition 5.3.25 of \cite{BR2}).
For each $\varphi\in{\mathfrak B}_{\Phi}$, we call
the positive operator $H_{\varphi,\alpha_{\Phi}}$ the bulk Hamiltonian
associated to $\varphi$, $\Phi$.
We consider the following class of the interactions.
\begin{defn}\label{def:jb}
We denote by $\caJ_B$, the set of  all $\Phi\in\caJ$ which satisfiy the following conditions. 
\begin{enumerate}
\item
For any $\varphi\in {\mathfrak B}_{\Phi}$, 
$0$ is the non-degenerate eigenvalue of the bulk Hamiltonian
$H_{\varphi,\alpha_{\Phi}}$.
\item 
There is a constant $\gamma>0$, 
such that
\[
\sigma\lmk H_{\varphi,\alpha_{\Phi}}\rmk\setminus \left\{0\right\}
\subset [\gamma,\infty),
\]
for any $\varphi\in {\mathfrak B}_{\Phi}$.
\end{enumerate}
\end{defn}
Now we specify what we mean by bulk-classification in this paper.
\begin{defn}[Bulk Classification]\label{def:bulk}
Let 
$\Phi_{0},\Phi_{1}\in{\caJ}_B$.
We say that the Hamiltonians $H_{\Phi_0},H_{\Phi_1}$ are 
bulk equivalent if the following conditions are satisfied.
\begin{enumerate}
\item There exist $m\in\nan$ and a continuous path of interactions $\Phi:[0,1]\to {\caJ}_m\cap \caJ_B$ such that $\Phi(0)=\Phi_0$, 
$\Phi(1)= \Phi_1$. 
\item  
There is a constant $\gamma>0$, 
such that
\[
\sigma\lmk H_{\varphi_s,\alpha_{\Phi(s)}}\rmk\setminus \left\{0\right\}
\subset [\gamma,\infty),
\]
for any $s\in[0,1]$ and $\varphi_s\in{\mathfrak B}_{\Phi(s)}$.
\end{enumerate}
When $H_{\Phi_0},H_{\Phi_1}$ are 
bulk equivalent, we write
$H_{\Phi_0}\simeq_{B} H_{\Phi_1}$.
\end{defn}
The class we bulk-classify is the class of frustration free Hamiltonians whose degrees of local ground states degeneracy 
 are
uniformly bounded.
\begin{defn}Let $\Phi$ be a positive translation invariant finite range interaction on $\caA_{\bbZ}$.
We say $H_\Phi$ is frustration free if 
$1\le \dim \ker\lmk H_\Phi\rmk_{[0,N-1]}$
for all $N\in\nan$.
Furthermore, we say that the ground state degeneracy 
of $H_{\Phi}:=\left((H_{\Phi })_\Lambda\right)_{\Lambda\in{\mathfrak I}_{\bbZ}}$ is
uniformly bounded
if $\sup_{N}\dim  \ker\lmk H_\Phi\rmk_{[0,N-1]}<\infty$.
We denote by $\caJ_{FB}$, the set of all interactions $\Phi\in\caJ_B$ satisfying the followings:
\begin{enumerate}
\item There exists a unique $\alpha_\Phi$-ground state $\omega$ on $\caA_{\bbZ}$.
\item The Hamiltonian $H_\Phi$ is frustration free and the ground state degeneracy 
of the corresponding local Hamiltonians is 
uniformly bounded.
\end{enumerate}
\end{defn}
We prove the following.
\begin{thm}\label{thm:bulk}
For any $\Phi_0,\Phi_1\in \caJ_{FB}$, we have
$H_{\Phi_0}\simeq_{B}H_{\Phi_1}$.
\end{thm}

\section{Paths of gapped Hamiltonians}\label{general}
In this section, we introduce a sufficient condition for a path of Hamiltonians to be uniformly gapped.
We first introduce a set of conditions on sequences of subspaces.
\begin{defn}[\it Condition 5]
Let  $n,a\in\nan$. Let 
$\xi_j^l:[0,1]\to \bigotimes_{i=0}^{l-1}\cc^n$
be a map
 given for each $l\in\nan$ and $j=1,\ldots,a$.
For each $l\in\nan$ and $t\in[0,1]$, let 
$\caD_l(t)$ be the subspace of 
$\bigotimes_{i=0}^{l-1}\cc^n$
 spanned by $\{\xi_j^l(t)\}_{j=1}^a$, and $G_l(t)$ the orthogonal projection
 onto $\caD_l(t)$.
 We say $\{\xi_j^l\}_{j=1,\ldots,a,\;l\in\nan}$
 satisfies the {\it Condition 5} if  
 the following conditions are satisfied.
\begin{description}
\item[(i)]
For any $l\in\nan$ and $j=1,\ldots,a$, the map
$
\xi_j^l:[0,1]\to \bigotimes_{i=0}^{l-1}\cc^n
$
is continuous and piecewise $C^{\infty}$.
\item[(ii)]
There exists an $l'\in\nan$
such that for all $l\ge l'$ and $t\in[0,1]$,
the vectors $\{\xi_j^l(t)\}_{j=1}^a$ are linearly independent.
\item[(iii)]
There exists an $m'\in\nan$
such that for all 
$N\ge m'+1$ and $t\in[0,1]$,
\begin{align}
\caD_N(t)=\lmk\caD_{N-1}(t)\otimes\cc^n\rmk\cap\lmk\cc^n\otimes\caD_{N-1}(t)\rmk.\end{align}
\item[(iv)]
There exists an $l''\in\nan$
such that for all $l\ge l''$,
there exists $0<\varepsilon_l<\frac1{2\sqrt l}$
such that
\begin{align}
\sup_{t\in[0,1]}\lV
\lmk \unit_{[0,N-l]}\otimes G_l(t)\rmk\lmk G_N(t)\otimes\unit_{\{N\}}-G_{N+1}(t)\rmk
\rV<\varepsilon_l,
\end{align}
for all $N\ge 2l$.
\end{description}
We say $\{\xi_j^l\}_{j=1}^a$
satisfies the {\it Condition 5} for $(l',m',l'')$
when we would like to specify the numbers.
\end{defn}
\begin{rem}\label{rem:nt}
Suppose that  $\{\xi_j^l\}_{j=1}^a$ satisfies {\it Condition 5}.
As the dimensions of $\caD_N(t)$s are uniformly bounded by $a$,
(iii) of {\it Condition 5} implies $\caD_m(t)\neq\bigotimes_{i=0}^{m-1}\bbC^n$, for any $m\ge m'$, i.e., 
$\dim\caD_m(t)<n^m$, for all $m\ge m'$.
In particular, we have $a=\dim\caD_{\max\{l',m'\}}(t)<n^{\max\{l',m'\}}$.
For the same reason, for $\vv\in \mk^{\times n}$ with $m_\vv<\infty$, we have 
$\dim \caG_{m,\vv}<n^m$ for all $m\ge m_\vv$.
In particular, if $\bbB\in\ClassA$ with respect to 
$(n_0,k_R,k_L,\lal,\bbD,\bbG,Y)$, then we have $n_0^2(k_R+1)(k_L+1)<n^{2l_\bbB}$
by Proposition 3.1 of \cite{Ogata1}.
Therefore, in the statement of Theorem 1.18 of \cite{Ogata1}, we note
$2l_\bbB=\max\{ 2l_\bbB, \frac{\log\lmk n_0^2(k_R+1)(k_L+1)+1\rmk}{\log n}\}$.
\end{rem}

\begin{lem}\label{lem:vpath}
Let  $n,a\in\nan$ with $2\le n$, and $\xi_j^l: [0,1]\to \bigotimes_{i=0}^{l-1}\cc^n$
be maps
 given for each $l\in\nan$ and  $j=1,\ldots,a$.
For each $l\in\nan$ and $t\in[0,1]$, let $\caD_l(t)$ be the subspace of $\bigotimes_{j=0}^{l-1}\cc^n$
 spanned by $\{\xi_j^l(t)\}_{j=1}^a$ and $G_l(t)$ the orthogonal projection
 onto $\caD_l(t)$.
Suppose that $\{\xi_j^l\}_{j=1,\ldots,a,\;l\in\nan}$
 satisfies {\it Condition 5} for $(l',m',l'')$.
Then we have the followings.
\begin{enumerate}
\item For any $l\ge l'$, the map $G_l:[0,1]\to \bigotimes_{i=0}^{l-1}\Mat_n(\cc)$
is continuous and piecewise $C^\infty$.
\item For any $t\in[0,1]$, $\bbm_{\{\caD_N(t)\}}\le m'$,
(see Definition 1.3 of Part I \cite{Ogata1} for $\bbm_{\{\caD_N(t)\}}$).
\item For any $m\ge m'$ and $t\in[0,1]$,
$\ker\lmk H_{\Phi_{1-G_{m}(t)}}\rmk_{[0,N-1]}=\caD_N(t)$
for all $N\ge m$.
\item For any $m\ge m'$ and $l\ge\max\{ l',m\}$,
\[
\gamma_{l}^m:=\inf_{t\in[0,1]}\drr\lmk \sigma\lmk \lmk H_{\Phi_{1-G_{m}(t)}}\rmk_{[0,l-1]}   \rmk\setminus\{0\},\{0\} \rmk
>0.
\]
(Here, $\drr(A,B)$ denotes the Euclidean distance between $A,B\subset\bbR$.)
\item
For any $m\ge m'$, and $t\in[0,1]$,
\[
\frac{\gamma_l^m\wedge \gamma_{2l}^m}{4(l+2)}\lmk1-G_{N}(t)\rmk
\le
\lmk H_{\Phi_{1-G_{m}(t)}}\rmk_{[0,N-1]} ,\quad
\text{for all}\quad l\ge \max\{l'',l',m\},\; \text{and}\;N\ge 2l+1.
\]
\item For any  $m_0,m_1\ge m'$, we have
$H_{\Phi_{1-G_{m_0}(0)}} \simeq_{I} H_{\Phi_{1-G_{m_1}(1)}}$.
\end{enumerate}
\end{lem}
\begin{proof}
{\it 1.} is immediate from  Lemma \ref{li}.
It is clear that (iii) implies Property~(I,$m'$), namely {\it 2} holds.  (Recall (3) of \cite{Ogata1} for Property~(I,$m'$).)
Property~(I,$m'$) implies {\it 3}.
Applying Lemma \ref{poa} to $\lmk H_{\Phi_{1-G_{m}(t)}}\rmk_{[0,l-1]}$, we obtain
{\it 4.} (Here, we used $a<n^{\max\{l',m'\}}$ to guarantee $\rank \lmk H_{\Phi_{1-G_{m}(t)}}\rmk_{[0,l-1]}=n^l-a\in\nan$,
for $l\ge\max\{ l',m'\}$.)
Theorem 3 of \cite{Nachtergaele:1996vc} gives {\it 5.}
The properties {\it 1.-5.} implies $ H_{\Phi_{1-G_{m}(0)}} \simeq_{I} H_{\Phi_{1-G_{m}(1)}}$ for $m\ge m'$. By the argument of Lemma 4.2 \cite{bo}, we obtain
\[
H_{\Phi_{1-G_{m_0}(0)}} \simeq_{I} H_{\Phi_{1-G_{m}(0)}},\quad H_{\Phi_{1-G_{m_1}(1)}} \simeq_{I} H_{\Phi_{1-G_{m}(1)}}.\]
Hence we obtain $H_{\Phi_{1-G_{m_0}(0)}} \simeq_{I} H_{\Phi_{1-G_{m_1}(1)}}$.
\end{proof}
Recall the definition of $M_{\vv,p,q}$ given in (8), and {\it Condition 2,3,4} given by Definition 2.3, 2.4, 2.5 of Part I \cite{Ogata1}
\begin{prop}\label{prop:maingen}
Let $n,k,m_1,m_2,m_3\in\nan$ and $p,q\in\caP(\mk)$.
Let $v_{\mu}:[0,1]\to\mk$, $\mu=1,\ldots,n$ be $C^{\infty}$-maps.
Let $\{x_i\}_{i=1}^{\rank p\cdot\rank q}$ be a basis of $p\mk q$.
Suppose that for each $t\in[0,1]$, 
the pentad $(n,k,p,q,\vv(t))$ satisfies the {\it Condition 2}, and
the {\it  Condition 3} for $m_1$.
Furthermore, assume that the triple $(n,k,\vv(t))$
satisfies the {\it Condition 4} for $(m_2,m_3)$
for each $t\in[0,1]$.
Then $M_0:=\sup_{t\in[0,1]}\left\{
M_{\vv(t),p,q}\right\}<\infty$
and $\{\gvtr{l}(x_i)\}_{i=1,\ldots,\rank p\cdot\rank q,\; l\in\nan}$  satisfies 
the {\it Condition 5}.
Here, 
$l',m'\in\nan$ of (ii), (iii) in  {\it Condition 5}
can be taken $l'=M_0$ and $m'=m_2+m_3$.
\end{prop}

\begin{proof}
The first condition (i) is clear from the definition of $\gvtr{l}$
and the fact that $\vv(t)$ is $C^\infty$.
We recall the definitions of $a_\vv$, $c_\vv$,
$e_\vv$, $\rho_\vv$, $\varphi_\vv$,
$E_\vv$, $F_\vv$,
$L_\vv$ introduced in pp10--11 of Part I.
As $(n,k,p,q,\vv(t))$ satisfies the {\it Condition 2}, from Lemma 2.9 of \cite{Ogata1},
there exists a constant $0<s_{\vv(t)}<1$, a state $\varphi_{\vv(t)}$, 
and a positive element $e_{\vv(t)}\in\mk_{+}$ and 
$T_{\vv(t)}$ satisfies the Spectral Property II
with respect to $(s_{\vv(t)},e_{\vv(t)},\varphi_{\vv(t)})$.
Furthermore, 
we have
$s(e_{\vv(t)})=p,\quad s(\varphi_{_\vv(t)})=q$.
In particular, 
$[0,1]\ni t \mapsto T_{\vv(t)}$
is a $C^{\infty}$ -map satisfying
the (1), (2), (3) of Lemma \ref{lem:cinf} with $r_{T_{\vv(t)}}=1$.
Therefore, from Lemma \ref{lem:cinf},
there exists $0<s<1$ such that
$\sigma(T_{\vv(t)})\setminus\{1\}\subset \caB_s(0)$
for all $t\in [0,1]$.
Furthermore, $[0,1]\ni t\mapsto e_{\vv(t)}=P_{\{1\}}^{T_{\vv(t)}}(1)$
is a $C^{\infty}$-path of positive elements of $\mk$
whose rank is equal to $\rank p$.
As $a_{\vv(t)}=\drr\lmk \sigma(e_{\vv(t)})\setminus \{0\},\{0\}\rmk$,
we have $a:=\inf_{t\in[0,1]}a_{\vv(t)}>0$ by Lemma
\ref{poa}.
Similarly,  $[0,1]\ni t\mapsto 
\Tr\lmk P_{\{1\}}^{T_{\vv(t)}}(\cdot )\rmk\backslash \Tr e_{\vv(t)}
=\varphi_{\vv(t)}=\Tr\rho_{\vv(t)}\lmk \cdot\rmk$
is $C^{\infty}$ and we have
$c:=\inf_{t\in[0,1]}c_{\vv(t)}>0$.
Therefore, for any $s<s'<1$ fixed, we have
\begin{align*}
&C_1:=\lmk ac\rmk^{-1}
k^2
\lmk
\sup_{t\in[0,1]}\sup_{|z|=s'}
\lV(z-T_{\vv(t)})^{-1}\rV
\rmk
<\infty,\\
&C_2:=\sup_{t\in[0,1]}
F_{\vv(t)}=
\sup_{t\in[0,1]}
\lmk
\sup_{N\in\nan}\lV T_{\vv(t)}^N\lmk 1- P^{T_{\vv(t)}}_{\{1\}}\rmk\rV+\lV e_{\vv(t)}\rV\rmk<\infty.
\end{align*}
With these constants, we have
\begin{align}\label{eq:ltl}
\sup_{t\in[0,1]}E_{\vv(t)}(N)
=\sup_{t\in[0,1]}\lmk a_{\vv(t)}c_{\vv(t)}\rmk^{-1}
k^2\lV T_{\vv(t)}^N\lmk 1- P^{T_{\vv(t)}}_{\{1\}}\rmk\rV
\le C_1
 (s')^{N}.
\end{align}
Fix some $L_0\in\nan$ such that 
$C_1
 (s')^{L_0}<\frac 12$.
Then we have
\begin{align*}
L_{\vv(t)}=\inf \left\{L\in\nan\mid
\sup_{N\ge L}E_{\vv(t)}(N)<\frac 12\right\}
\le L_0,\quad t\in[0,1],
\end{align*}
by (\ref{eq:ltl}).
By Lemma 2.15 of Part I \cite{Ogata1},
we have
\begin{align*}
M_0:=\sup_{t\in[0,1]}M_{\vv(t),p,q}
\le \sup_{t\in[0,1]}L_{\vv(t)}\le L_0<\infty.
\end{align*}
This proves (ii) of {\it Condition 5}
for $l'=M_0$.

(iii) of {\it Condition 5} for $m'=m_2+m_3$
follows from Lemma 2.14 of Part I \cite{Ogata1}.

To prove (iv), note that 
\begin{align*}
\sup_{t\in[0,1]}2 F_{\vv(t)}E_{\vv(t)}(m)\lmk
F_{\vv(t)}^2E_{\vv(t)}(m)+1
\rmk
\le
2C_1C_2\lmk
C_2^2C_1+1
\rmk s'^m.
\end{align*}
Choose $L_1\in\nan$ such that
$8C_1C_2\lmk
C_2^2C_1+1
\rmk \sqrt {m+1} {s'}^m<1$
for any $m\ge L_1$.
We use Lemma 2.17 of Part I \cite{Ogata1}
replacing $(l,m,r)$ in Lemma 2.17 of Part I \cite{Ogata1}
by $(N-l+1,l-1,1)$
with $l-1\ge\max\{m_1,L_0,L_1\}=:l''$.
Then we obtain (iv) with this $l''$.
\end{proof}

\section{$C^1$-classification of $\caH(n,n_0,k_R,k_L)$}\label{classc1}
In this section, we show the following Proposition.
\begin{prop}\label{classthm}
Let $n,n_0\in\nan$ with $n\ge 2$ and $k_R,k_L\in\nan\cup\{0\}$.
Then for any $\bb_0,\bb_1\in \Class(n,n_0,k_R,k_L)$, 
and $m_0,m_1\in\nan$ with $m_0,m_1\ge 2n_0^6(k_R+1)(k_L+1)$,
we have
$H_{\Phi_{m_0,\bb_0}}\simeq_{I} H_{\Phi_{m_1,\bb_1}}$.
\end{prop}
First we check what the condition $\lal\in \Wo'(k_R,k_L)$ implies.
\begin{lem}\label{lem:ad}
Let $k_R,k_L\in\nan\cup\{0\}$, and $(\lal,\bbD,\bbG,Y)\in \caT(k_R,k_L)$.
If furthermore $\lal\in\Wo'(k_R,k_L)$, then $Y=0$ and $\bbD$, $\bbG$ satisfy the followings.:
\begin{enumerate}
\item
For any $1\le a_1, a_2\le k_R$,
either there exist $\sigma(a_1, a_2)\in\{1,\ldots,k_R\}$
and a nonzero $\kappa(a_1,a_2)\in\cc$ such that
$D_{a_1}D_{a_2}=\kappa(a_1,a_2)D_{\sigma(a_1,a_2)}$,
or $D_{a_1}D_{a_2}=0$.
\item
For any $1\le b_1,b_2\le k_L$,
either there exist $\sigma(b_1, b_2)\in\{1,\ldots,k_L\}$
and a nonzero $\kappa(b_1,b_2)\in\cc$ such that
$G_{b_1}G_{b_2}=\kappa(b_1,b_2)G_{\sigma(b_1,b_2)}$,
or $G_{b_1}G_{b_2}=0$.
\end{enumerate}
\end{lem}
\begin{proof}
As $Y$ belongs to $\UT_{0, k_R+k_L+1}$ and commutes with $\Lambda_\lal$ ((11) of \cite{Ogata1}) , the assumption $\lal\in \Wo'(k_R,k_L)$ implies that
$\pu Y \pu=\pd Y \pd=0$. As we also have $\pu Y\pd=0$, we get $Y=0$.
To see the property of $\bbD$ in the claim, let $1\le a_1, a_2\le k_R$.
As the linear span of 
$\{D_a\}_{a=1}^{k_R}$ is a subalgebra of $\UT_{0,k_R+1}$, 
$D_{a_1}D_{a_2}$ can be expanded as $D_{a_1}D_{a_2}=\sum_{a=1}^{k_R}c_a D_a$ with some coefficients $c_a\in\bbC$.
From (9) of \cite{Ogata1}, we have
\begin{align*}
\sum_{a=1}^{k_R}c_a \lambda_{-a}\eij{-a0}
={\Lambda_{\lal}}I_R^{(k_R,k_L)}(D_{a_1}D_{a_2}){\Lambda_{\lal}^{-1}}\eij{00}=\lambda_{-a_1}\lambda_{-a_2}I_R^{(k_R,k_L)}(D_{a_1}D_{a_2})\eij{00}
=\sum_{a=1}^{k_R}c_a \lambda_{-a_1}\lambda_{-a_2}\eij{-a0}
\end{align*}
From this, we have $c_a=0$ unless $\lambda_{-a_1}\lambda_{-a_2}=\lambda_{-a}$. As our $\lal$ is in $\Wo'(k_R,k_L)$,
such an $a$ is at most one. Hence we obtain the claim.
The assertion for $\bbG$ is proven analogously.
\end{proof}

\begin{lem}\label{lem:dbb}
Let $n,n_0\in\nan$ with $n\ge 2$, $k_R,k_L\in\nan\cup\{0\}$ 
and $\bb\in\Class(n,n_0,k_R,k_L)$.Then
$(\lal,\bbD,\bbG,0)\in{\caT}(k_R,k_L)$, $\lal\in\Wo'(k_R,k_L)$ with respect to which $\bbB$ belongs to $\Class(n,n_0,k_R,k_L)$
is uniquely determined.
Furthermore, there exists a unique $\oo\in\mnz^{\times n}$ such that 
\begin{align}\label{eq:loe}
\lambda_i\omega_{\mu}\otimes E^{(k_R,k_L)}_{ii}=\lmk \unit\otimes E^{(k_R,k_L)}_{ii}\rmk
B_{\mu}
\lmk \unit\otimes E^{(k_R,k_L)}_{ii}\rmk,\quad \mu=1,\ldots,n,\quad i=-k_R,\ldots,k_L.
\end{align}
This $\oo$ belongs to $\Prim(n,n_0)$. (Recall Definition 1.6 \cite{Ogata1} for $\Prim(n,n_0)$.)
\end{lem}
\begin{defn}\label{def:ldgob}
We call the $(\lal,\bbD,\bbG,\oo)$ determined uniquely for
$\bb\in\Class(n,n_0,k_R,k_L)$ by Lemma \ref{lem:dbb},
the quadruplet associated with $\bb$ and write it
as
$(\lal_\bb,\bbD_\bb,\bbG_\bb,\oo_{\bb})$.
\end{defn}
\begin{proof}
The formula (\ref{eq:loe}) with $i=0$ corresponds to the formula in Lemma 3.2 \cite{Ogata1}, and the latter Lemma implies 
$\oo\in \Prim(n,n_0)$. As we have $B_\mu\in\mnz\otimes
 \caD(k_R,k_L,\bbD,\bbG)\Lambda_{\lal}$, we obtain (\ref{eq:loe}) for all $i=-k_R,\ldots,k_L$.
Suppose that 
$\bbB$ is in $\Class(n,n_0,k_R,k_L)$ with respect to $(\lal,\bbD,\bbG,0)\in{\caT}(k_R,k_L)$
as well as $(\lal',\bbD',\bbG',0)\in{\caT}(k_R,k_L)$.
As $r_{T_{\oo}}=1\neq0 $, there exists a $\mu$ such that $\omega_{\mu}\neq 0$. 
For this $\mu$ and 
any $i=-k_R,\ldots,k_L$, we have
\begin{align}
\lambda_i\omega_\mu\otimes\eij{ii}=\lmk \unit\otimes E^{(k_R,k_L)}_{ii}\rmk
B_{\mu}
\lmk \unit\otimes E^{(k_R,k_L)}_{ii}\rmk
=\lambda_i'\omega_\mu\otimes\eij{ii},
\end{align}
and obtain $\lal=\lal'$.
If $k_R\in \nan$, then for any $a\in\{1,\ldots, k_R\}$ fixed,
we have
\[
\unit\otimes I_R^{(k_R,k_L)}(D_a')\Lambda_{\lal}^l\in 
\kl{l}(\bb)\hat P^{(n_0,k_R,k_L)}_R=
\mnz\otimes\spa\left\{I_R^{(k_R,k_L)}(D_{a})
\Lambda_{\lal}^l\right\}_{a=0}^{k_R}
\]
for all $l$ large enough. Here, we use a notation $D_0:=P_{R}^{(k_R,0)}$.
This means fixing $l$ large enough, there exist $\{Z_{il}\}_{i=0}^{k_R}$
such that
\[
\unit\otimes I_R^{(k_R,k_L)}(D_a')\Lambda_{\lal}^l
=\sum_{i=0}^{k_R} Z_{il}\otimes I_R^{(k_R,k_L)}(D_{i})\Lambda_{\lal}^l.
\]
As $\{\lambda_{-i}\}_{i=0}^{k_R}$ is a set of $k_R+1$
distinct nonzero complex numbers, from Lemma C.7 Part I \cite{Ogata1}, there exist
$\varsigma_{i}=(\varsigma_i(j))_{j=0}^{k_R}\in\cc^{k_R+1}$,
$i=0,\ldots,k_R$ such that
$\sum_{j=0}^{k_R}\varsigma_i(j)\lambda_{-i'}^j=\delta_{i,i'}$, $i,i'=0,\ldots,k_R$.
Using this $\varsigma_{i}$, 
we have
\begin{align*}
&\unit\otimes I_R^{(k_R,k_L)}(D_a')\Lambda_{\lal}^l
=\sum_{j=0}^{k_R} \varsigma_a(j) \lambda_{-a}^j
\lmk \unit\otimes I_R^{(k_R,k_L)}(D_{a}')\Lambda_{\lal}^{l} \rmk
=\sum_{j=0}^{k_R} \varsigma_a(j) 
\lmk \unit\otimes \Lambda_{\lal}^j\rmk
\lmk \unit\otimes I_R^{(k_R,k_L)}(D_{a}')\Lambda_{\lal}^{l} \rmk
\lmk \unit\otimes \Lambda_{\lal}^{-j}\rmk\\
&=\sum_{i=0}^{k_R}\sum_{j=0}^{k_R} \varsigma_a(j) 
\lmk \unit\otimes \Lambda_{\lal}^j\rmk
\lmk Z_{il}\otimes I_R^{(k_R,k_L)}(D_{i})\Lambda_{\lal}^{l} \rmk
\lmk \unit\otimes \Lambda_{\lal}^{-j}\rmk
=\sum_{i=0}^{k_R}\sum_{j=0}^{k_R} \varsigma_a(j) 
\lambda_{-i}^j\lmk Z_{il}\otimes I_R^{(k_R,k_L)}(D_{i})\Lambda_{\lal}^{l} \rmk\\
&=Z_{al}\otimes I_R^{(k_R,k_L)}(D_{a})\Lambda_{\lal}^l.
\end{align*}
This means $D_a'$ is proportional to $D_a$. By the normalization
$I_R^{(k_R,k_L)}(D_a)\eij{00}= I_R^{(k_R,k_L)}(D_a')\eij{00}=\eij{-a0}$, we conclude $D_a=D_a'$.
Similarly, we obtain $G_b=G_b'$. 
\end{proof}

\subsection{Constructive characterization of $\Class(n,n_0,k_R,k_L)$ }
In order to classify $\caH(n,n_0,k_R,k_L)$, we need to understand the property of $\Class(n,n_0,k_R,k_L)$. We carry it out in this subsection.
The number $l_\bbB$ in (\ref{eq:lblb}) for $\bbB\in\Class(n,n_0,k_R,k_L)$ has a uniform upper bound.
\begin{lem}\label{lem:klmain}
Let $n,n_0\in\nan$ with $n\ge 2$, $k_R,k_L\in\nan\cup\{0\}$. Let $\bb\in\Class(n,n_0,k_R,k_L)$ with respect to $(\lal,\bbD,\bbG,0)$.
Then $l_{\bb}(n,n_0,k_R,k_L,\lal,\bbD,\bbG,0)\le n_0^6(k_R+1)(k_L+1)$.
\end{lem}
\begin{proof}
We would like to apply Lemma \ref{lem:qw} to
see $\l_\bb(n,n_0,k_R,k_L,\lal,\bbD,\bbG,0)\le n_0^6(k_R+1)(k_L+1)$.
The dimension of  $\mnz\otimes \caD(k_R,k_L,\bbD,\bbG)\Lambda_{\lal}^l$ is $n_0^2(k_R+1)(k_L+1)$.
Therefore, $\dim\kl{l}(\bb)=
\dim\lmk \mnz\otimes \caD(k_R,k_L,\bbD,\bbG)\Lambda_{\lal}^l\rmk=n_0^2(k_R+1)(k_L+1)$
for all $l\ge l_\bb(n,n_0,k_R,k_L,\lal,\bbD,\bbG,0)$.

On the other hand, $\kl{l_{\oo}}(\bb)$ has an invertible element.
(Recall Definition 1.6 of Part I \cite{Ogata1} for $l_\oo$.)
To see this, note by the definition of $l_{\oo}$, that
$\unit_{n_0}\in\mnz=\kl{l_{\oo}}(\oo)$.
This means there exists a set of coefficients $\{c_{\mu^{(l_\omega)}}\}$
such that $\sum_{\mu^{(l_\omega)}}c_{\mu^{(l_\omega)}}\wo{l_{\omega}}=\unit_{n_0}$.
Therefore,
we have
\[
\kl{l_{\omega}}(\bb)\ni \sum_{\mu^{(l_\omega)}}c_{\mu^{(l_\omega)}}\wb{l_{\omega}}
=\unit_{n_0}\otimes\Lambda_{\lal}^{l_{\omega}}
+\text{an element of }\mnz\otimes \UT_{0,k_R+k_L+1}.
\]
From the right hand side, we see that
this matrix is invertible.

Note that 
$\dim\caK_l(\bb)=n_0^2(k_R+1)(k_L+1)$
if and only if 
$\caK_l(\bb)=\mnz\otimes \caD (k_R,k_L,\bbD,\bbG)\Lambda_{\lal}^l$, because
we always have $\caK_l(\bb)\subset \mnz\otimes \caD (k_R,k_L,\bbD,\bbG)\Lambda_{\lal}^l$, due to the fact that
$\caD (k_R,k_L,\bbD,\bbG)$ is an algebra invariant under operations $\Lambda_{\lal}^x\lmk \cdot\rmk\Lambda_{\lal}^{-x}$, $x\in\bbZ$.
 Applying Lemma \ref{lem:qw},
we obtain
\[
l_\bb(n,n_0,k_R,k_L,\lal,\bbD,\bbG,0)\le n_0^2(k_R+1)(k_L+1) l_{\oo}
\le n_0^6(k_R+1)(k_L+1).
\]
For the second inequality, we used Lemma \ref{lem:oqw}.
\end{proof}
We would like to give a constructive characterization of $\Class(n,n_0,k_R,k_L)$.
To do so,
we introduce the following notations.
\begin{defn}
Let $n,n_0\in\nan$ with $n\ge 2$, $\lambda\in\cc$, and
$\oo\in\Prim(n,n_0)$.
We define
a linear subspace $\caL(\oo,\lambda)$
of $\oplus_{\mu=1}^n\mnz$ by
\[
\caL(\oo,\lambda):=\left\{ (J\omega_{\mu}-
\lambda\omega_{\mu}J)_{\mu=1}^n\mid
J\in\mnz\right\}.
\]
We also define a linear map $\Delta_{\oo,\lambda}:\mnz\to \caL(\oo,\lambda)$ by
\[
\Delta_{\oo,\lambda}(J):=(J\omega_{\mu}-
\lambda\omega_{\mu}J)_{\mu=1}^n,\quad J\in\mnz.
\]
\end{defn}
\begin{defn}
Let $n,n_0\in\nan$ with $n\ge 2$, $k_R,k_L\in\nan\cup\{0\}$
and $(\lal,\bbD,\bbG,0)\in\caT(k_R,k_L)$ with $\lal\in\Wo'(k_R,k_L)$.
We define
\begin{align*}
&{\hfu{\bbD}}:=\{1\le a \le k_R \mid
\text{there are no }1\le a_1,a_2\le  k_R\;
\text{such that}\;
D_a\text{is a scalar multiple of}\; D_{a_1}D_{a_2}
\},\\
&{\hfd{\bbG}}:=\{1\le b\le k_L \mid
\text{there are no }1\le b_1,b_2\le k_L\;
\text{such that}\;G_b 
\text{is a scalar multiple of}\;G_{b_1}G_{b_2}
\},\\
&\ghu{l}:=
\mnz\otimes 
\spa\{\ir\lmk D_a\rmk \Lambda_{\lal}^l\}_{1\le a\notin{{\hfu{\bbD}}}},\\
&\ghd{l}:=\mnz\otimes \spa\{\il(G_b)\Lambda_{\lal}^l\}_{1\le b\notin{{\hfd{\bbG}}}}.
\end{align*}
\end{defn}
\begin{lem}\label{lem:Delta}
Let $n,n_0\in\nan$ with $n\ge 2$, $\lambda\in\cc$ with $\lambda\neq 1$, and
$\oo\in\Prim(n,n_0)$.
Then $\Delta_{\oo,\lambda}$ is an injection.
\end{lem}
\begin{proof}
If $J\in \ker \Delta_{\oo,\lambda}$, then
we have
$
J\omega_{\mu}=\lambda\omega_{\mu}J, 
$ for all $\mu=1,\ldots n$.
From this, we obtain
$
J\omega_{\mu^{(l)}}=\lambda^l\omega_{\mu^{(l)}}J,
$ for any $l\in\nan$ and $\mu^{(l)}\in\{1,\ldots n\}^{\times l}$.
As $\oo\in\Prim(n,n_0)$, we have $\kl{l}(\oo)=\mnz$ for $l$ large enough.
Hence, in particular, we have
$
J=\lambda^{l}J,
$
for $l$ large enough.
As $\lambda\neq 1$, this implies $J=0$.
\end{proof}

\begin{defn}\label{def:cl1}
Let $n,n_0\in\nan$ with $n\ge 2$ and $k_R,k_L\in \nan\cup\{0\}$.
Let $\bb=(B_1,\ldots,B_n)$ be 
an $n$-tuple of matrices in $\mnz\otimes \UT_{k_R+k_L+1}\subset\mnz\otimes\mkk$.
We say $\bb$ belongs to $\Class_1(n,n_0,k_R,k_L)$
if 
there exist $(\lal,\bbD,\bbG,0)\in\caT(k_R,k_L)$
with $\lal\in \Wo'(k_R,k_L)$
 and $\oo\in\Prim(n,n_0)$
satisfying the following conditions.:
\begin{description}
\item[(1)]For all $i=-k_R,\ldots,k_L$ and $\mu=1,\ldots,n$, we have
\[
\lambda_i\omega_{\mu}\otimes E^{(k_R,k_L)}_{ii}=\lmk \unit\otimes E^{(k_R,k_L)}_{ii}\rmk
B_{\mu}
\lmk \unit\otimes E^{(k_R,k_L)}_{ii}\rmk.
\]
\item[(2)]
For any $\mu=1,\ldots,n$, we have
\begin{align}\label{eq:cl1}
&B_{\mu} \hat P^{(n_0,k_R,k_L)}_R
=\omega_\mu\otimes\Lambda_{\lal}\pu
+\sum_{1\le a\le k_R}\caX_{\mu,a,\bb}\otimes \ir(D_a)\Lambda_{\lal},\nonumber\\
&\hpd B_{\mu} 
= \omega_\mu\otimes\Lambda_{\lal}\pd
+\sum_{1\le b\le k_L}
\caY_{\mu ,b,\bb}\otimes \il(G_b)\Lambda_{\lal},
\end{align}
where
$\bbX_{a,\bb}:=(\caX_{\mu,a,\bb})_{\mu=1}^n\notin \caL(\oo,\lambda_{-a})$
and $\bbY_{b,\bb}:=
(\caY_{\mu,b,\bb})_{\mu=1}^n\notin \caL(\oo,\lambda_{b}^{-1})$,
for any $a\in {\mathfrak H}_{\bbD}^R$ and 
$b\in{\mathfrak H}_{\bbG}^L$.
\end{description}
\end{defn}
\begin{rem}\label{rem:xy}
We say $\bb$ belongs to $\Class_1(n,n_0,k_R,k_L)$
with  respect to $(\lal,\bbD,\bbG,\oo)$,
when we would like to state explicitly.
We denote 
$\boldsymbol{\caX}_\bb:=\{{\caX}_{\mu,a,\bb}\}_{a,\mu}$,
$\boldsymbol{\caY}_\bb:=\{\caY_{\mu,b,\bb}\}_{b,\mu}$,
and call them the sets of matrices associated with $\bb$.
(Note that they are determined uniquely by $\bbB$ due to the independence of
$\{\unit\}\cup\{I_R^{(k_R,k_L)}(D_a)\}_{a=1}^{k_R}
\cup\{I_L^{(k_R,k_L)}(G_b)\}_{b=1}^{k_L}\cup\{ E_{-a,b}^{(k_R,k_L)}\}_{a=1,\ldots,k_R, b=1,\ldots,k_L}$.)
\end{rem}
The rest of this subsection is devoted to showing the following Lemma:
\begin{lem}\label{lem:cc1}
Let $n,n_0\in\nan$ with $n\ge 2$ and $k_R,k_L\in \nan\cup\{0\}$.
Let $(\lal,\bbD,\bbG,0)\in\caT(k_R,k_L)$ with $\lal\in\Wo'(k_R,k_L)$ and $\oo\in\Prim(n,n_0)$.
Then $\bb\in \lmk\mnz\otimes \UT_{k_R+k_L+1}\rmk^{\times n}$ belongs to 
$\Class(n,n_0,k_R,k_L)$ with associated quadraple  $(\lal,\bbD,\bbG,\oo)$ if and only if
it belongs to $\Class_1(n,n_0,k_R,k_L)$ with respect to $(\lal,\bbD,\bbG,\oo)$.
\end{lem}
\begin{rem}
The advantage of this characterization is that the conditions for $\Class_1(n,n_0,k_R,k_L)$ is easier to check, when we construct a path $\bbB(t)$.
\end{rem}

\begin{lem}\label{lem:bdecom}
Let $n,n_0\in\nan$ with $n\ge 2$, $k_R,k_L\in\nan$
and $(\lal,\bbD,\bbG,0)\in\caT(k_R,k_L)$ with $\lal\in\Wo'(k_R,k_L)$.
For any $b\in\{1,\ldots, k_L\}$,
there exists an $m\in\nan$ with $m\le k_L$,
$b_1,\ldots,b_m\in {\mathfrak H}_{\bbG}^L$,
and $c_b\in\cc$
such that
$G_b=c_bG_{b_1}\cdots G_{b_m}$.
\end{lem}
\begin{proof}
Set $n_b:=\min\{j-i:(G_b)_{ij}\neq 0\}$.
By the definition,
we have
$1\le n_b\le k_L$.
The claim is trivial when $b\in {\mathfrak H}_{\bbG}^L$.
If $b\notin {\mathfrak H}_{\bbG}^L$, then
there exist $b_1,b_2\in\{1,\ldots, k_L\}$
such that $\sigma(b_1,b_2)=b$, and 
$G_b=\kappa(b_1,b_2)^{-1}G_{b_1}G_{b_2}$
with nonzero $\kappa(b_1,b_2)\in \cc$.
From this, we see
$2\le n_{b_1}+n_{b_2}\le n_b$.
If furthermore $b_1\notin {\mathfrak H}_{\bbG}^L$,
we repeat the same argument
to obtain $b_{11}',b_{12}' \in\{1,\ldots, k_L\}$
such that
$ G_{b}\propto G_{b_{11}'}G_{b_{12}'}G_{b_2}$.
We repeat this procedure 
and obtain $b_1,b_2,\ldots, b_m$
such that
$G_b\propto G_{b_1}G_{b_2}\cdots G_{b_m}$
and $m\le n_{b_1}+n_{b_2}+\cdots+n_{b_m}\le n_b$.
If some of $b_i$ is not in ${\mathfrak H}_\bbG^L$,
we repeat the same argument to split $b_i$ into two.
However, this procedure stops in finite time
because of the bound $m\le n_b\le k_L$.
Namely, at some point, all of $b_1,\ldots,b_m$ will be in
${\mathfrak H}_\bbG^L$.
\end{proof}

\begin{lem}\label{lem:hcc1}
Let $n,n_0\in\nan$ with $n\ge 2$, $k_R,k_L\in \nan\cup\{0\}$.
Let $\bb\in\Class_1(n,n_0,k_R,k_L)$ with respect to
$(\lal,\bbD,\bbG,\oo)$.
Then we have $\mnz\otimes P_L^{(k_R,k_L)}\Lambda_{\lal}^l\subset 
\hpd \kl{l}(\bb)$
(resp. $\mnz\otimes P_R^{(k_R,k_L)}\Lambda_{\lal}^l\subset 
\kl{l}(\bb)\hpu$),
for $l$ large enough.
Furthermore, 
if $k_L\in\nan$, (resp. $k_R\in\nan$),
for any $b\in\{1,\ldots,k_L\}$
( resp. $a\in\{1,\ldots,k_R\}$),
we have $\mnz\otimes \il(G_b)\Lambda_{\lal}^l\subset 
\hpd \kl{l}(\bb)$
(resp. $\mnz\otimes \ir(D_a)\Lambda_{\lal}^l\subset 
\kl{l}(\bb)\hpu$),
for $l$ large enough.
\end{lem}
\begin{proof}
Note that the first statement is trivial if $k_L=0$ (resp. $k_R=0$),
from the primitivity of $\oo$.
We prove the Lemma for $k_L\in\nan$, $b\in\{1,\ldots,k_L\}$.
The proof for $k_R\in\nan$, $a\in\{1,\ldots,k_R\}$ is the same.
Assume $k_L\in\nan$.
For each
$b\in\{1,\ldots,k_L\}$, we define
\[
\check l_b:=\inf\left\{l\in\nan\mid 
\mnz\otimes \il(G_b)\Lambda_{\lal}^{l'}
\subset \hpd \kl{l'}(\bb)+\ghdb{l'},
\quad \text{for all } \quad l'\ge l
\right\}.
\]
First we see $\check l_b<\infty$ for all
$b\in\hfd{\bbG}$.
From (\ref{eq:cl1}) and $(\lal,\bbD,\bbG,0)\in{\caT}(k_R,k_L)$, we can check inductively
that for any $l\in\nan$ and $\mu^{(l)}\in\{1,\ldots,n\}^{\times l}$,
$\hpd\wb{l}$ is of the form
\[
\hpd \wb{l}=\wo{l}\otimes\Lambda_{\lal}^l\pd
+\sum_{b=1}^{k_L}
\caZ_{\mu^{(l)},b}\otimes \il(G_b)\Lambda_{\lal}^l,
\]
with $\caZ_{\mu^{(l)},b}\in\mnz$.

Recall that
we have $\mnz=\kl{l}(\oo)$ for all $l\ge l_{\oo}$.
Therefore, for any 
$l\ge l_{\oo}$ and $\alpha,\beta\in\{1,\ldots,n_0\}$, there exists
$X_{\alpha,\beta}^{(l)}\in\kl{l}(\bb)$
such that $\hpd X_{\alpha,\beta}^{(l)}$ is the of the form
\[
\hpd X_{\alpha,\beta}^{(l)}
=\zeij{\alpha,\beta}\otimes\Lambda_{\lal}^l\pd
+\sum_{b=1}^{k_L}
\caW_{\alpha,\beta,b}^{(l)}\otimes \il(G_b)\Lambda_{\lal}^l,
\]
with some $\caW_{\alpha,\beta,b}^{(l)}\in\mnz$.

For any $\alpha_1,\beta_1,
\alpha_2,\beta_2\in\{1,\ldots,n_0\}$
and $l_1,l_2\ge l_{\oo}$,
we have
\begin{align}\label{eq:xabd}
&
\hpd{\caK}_{l_1+l_2}(\bb)\ni \hpd X_{\alpha_1,\beta_1}^{(l_1)}
 X_{\alpha_2,\beta_2}^{(l_2)}
=\hpd X_{\alpha_1,\beta_1}^{(l_1)}
\hpd X_{\alpha_2,\beta_2}^{(l_2)}\nonumber\\
&=\delta_{\beta_1,\alpha_2}\zeij{\alpha_1,\beta_2}
\otimes\Lambda_{\lal}^{l_1+l_2}\pd
+
\sum_{b=1}^{k_L}
\lmk
\caW_{\alpha_1,\beta_1,b}^{(l_1)}\zeij{\alpha_2,\beta_2}
+\lambda^{-l_1}_{b}\zeij{\alpha_1\beta_1}
\caW_{\alpha_2,\beta_2,b}^{(l_2)}
\rmk\otimes \il(G_b)\Lambda_{\lal}^{(l_1+l_2)}\nonumber\\
&\quad +\text{an element of }\ghdb{l_1+l_2}.
\end{align}
We use the following Lemma which can be proven by the same argument as Lemma 7.14 of \cite{Ogata2}.
\begin{lem}\label{lem:yzd}
Let $n,n_0\in\nan$, $k_R\in\nan\cup \{0\}$ and $k_L\in\nan$. Let
$\bb\in\Class_1(n,n_0,k_R,k_L)$ with respect to $(\lal,\bbD,\bbG,\oo)$.
Suppose that there exist an $l'\in\nan$ and matrices 
$y_b\in \mnz$, $b=1,\ldots, k_L$,
such that
\begin{align*}
\sum_{b=1}^{k_L}y_b\otimes \il(G_b)\Lambda_{\lal}^{l'}\in
\hpd \kl{l'}(\bb).
\end{align*}
Then for any $b\in\hfd{\bbG}$ with $y_b\neq 0$,
we have
\[
\mnz\otimes \il(G_b)\Lambda_{\lal}^l\subset \hpd
\kl{l}(\bb)+
\ghdb{l},
\]
for all
$l\ge 2l_{\oo}+k_L+l'-1$.
\end{lem}
Suppose that $\check l_b=\infty$ for some $b\in\hfd{\bbG}$.
Applying Lemma \ref{lem:yzd} to
(\ref{eq:xabd}), we see that
$
x_{\alpha,\beta}^{(l)}:=\caW_{\alpha,\beta,b}^{(l)}
$
and $\lambda:=\lambda_b$
satisfy
the conditions in Lemma 6.6 of Part II \cite{Ogata2}
with $l_0=l_\oo$.
By the latter Lemma,
there exists $J\in\mnz$ such that  $\caW_{\alpha,\beta,b}^{(l)}
=J\zeij{\alpha,\beta}-\lambda_b^{-l}\zeij{\alpha\beta}J$, $l\ge l_\omega$, $\alpha,\beta=1,\ldots,n_0$, $l\ge l_\oo$.

Substituting this and (\ref{eq:cl1}), we have
\begin{align*}
\hpd{\caK}_{l_1+l_2+1}(\bb)\ni& 
\sum_{\alpha_1,\alpha_2=1}^{n_0}\lmk
\hpd X_{\alpha_1,\alpha_1}^{(l_1)}
B_{\mu} X_{\alpha_2,\alpha_2}^{(l_2)}
-\braket{\cnz{\alpha_1}}{\omega_{\mu}\cnz{\alpha_2}} X_{\alpha_1,\alpha_2}^{(l_1+l_2+1)}\rmk\\
&=\lambda_b^{-l_1}
\lmk
-J\omega_{\mu}+\caY_{\mu ,b,\bb}+\lambda_b^{-1}\omega_{\mu}J
\rmk\otimes \il(G_b)\Lambda_{\lal}^{l_1+l_2+1}\\
&+\text{elements of }\mnz\otimes\spa\{\il(G_{b'})\Lambda_{\lal}^{l_1+l_2+1}\}_{b'\neq b,b'=1,\ldots k_L}
\end{align*}
Applying Lemma \ref{lem:yzd} and from the assumption $\check l_b=\infty$,
we get 
\[
\caY_{\mu ,b,\bb}=J\omega_{\mu}-\lambda_b^{-1}\omega_{\mu}J,\quad \mu=1,\ldots,n. 
\]
This means $\bbY_{b,\bb}:=
(\caY_{\mu,b,\bb})_{\mu=1}^n\in \caL(\oo,\lambda_{b}^{-1})$,
which contradict the assumption $\bb\in\Class_1(n,n_0,k_R,k_L)$.
Hence we have $\check l_b<\infty$ for all $b\in\hfd{\bbG}$.

Let $\tilde l:=\max\{\check l_b\mid b\in\hfd{\bbG}\}<\infty$.
Then for each $b\in\hfd{\bbG}$, $\alpha,\beta\in\{1,\ldots,n_0\}$ 
and $l\ge\tilde l$,
there exists $C_{l,\alpha,\beta,b}\in \ghdb{l}$
such that
$\zeij{\alpha,\beta}\otimes
\il(G_b)\Lambda_{\lal}^l+C_{l,\alpha,\beta,b}\in 
\hpd \kl{l}(\bb)$.
For each $i\in\{1,\ldots,k_L\}$, we set
\[
m_i:=
\max 
\left\{
m\in\nan\mid
\text{there exist }
b_1,\ldots, b_m\in \hfd{\bbG}
\text{ such that }
G_i\propto G_{b_1}\ldots G_{b_m}
\right\}.
\]
From the proof of Lemma \ref{lem:bdecom},
we see that $m_i\le k_L$.

Let us consider $i\in \{1,\ldots,k_L\}$,
with a decomposition
$G_i=c_iG_{b_1}\cdots G_{b_{m_i}}$,
$b_1,\ldots, b_{m_i}\in \hfd{\bbG}$.
Then for all
$l_1,\ldots,l_{m_i}\ge\tilde l$
and $\alpha,\beta\in\{1\ldots,n_0\}$,
we have
\begin{align*}
&\lmk \zeij{\alpha,\beta}\otimes
 \il(G_{b_1})\Lambda_{\lal}^{l_1}+C_{l_1,\alpha,\beta,b_1} \rmk
\lmk \zeij{\beta,\beta}\otimes \il(G_{b_2})\Lambda_{\lal}^{l_2}+C_{l_2,\beta,\beta,b_2} \rmk
\cdots
\lmk  \zeij{\beta,\beta}\otimes \il(G_{b_{m_i}})\Lambda_{\lal}^{l_{m_i}}+C_{l_{m_i},\beta,\beta,b_{m_i}} \rmk\\
&\in \hpd\kl{l_1+l_2+\cdots+l_{m_i}}(\bb).
\end{align*}

From this, we see that for all
$i\in \{1,\ldots,k_L\}$,
$l\ge m_i\tilde l$
and $\alpha,\beta\in\{1,\ldots,n_0\}$,
we have
\begin{align}
\zeij{\alpha,\beta}\otimes
\il(G_i)\Lambda_{\lal}^{l}
+\text{an element of}
\mnz\otimes 
\spa\{\il(G_{b'})\Lambda_{\lal}^l\}_{b'\in\{1,\ldots,k_L\},m_{b'}>m_i}
\in \hpd\kl{l}(\bb).
\end{align}
In particular, if there is no $b'$ such that $m_{b'}>m_i$, then 
$\mnz\otimes \il(G_i)\Lambda_\lal^l\subset \hpd\caK_l(\bb)$ for 
$l$ large enough.
By induction in $m_i$ (from large $m_i$ to small ones), we prove the second part of the Lemma.
By this, for any $\alpha,\beta\in\{1,\ldots, n_0\}$, we have 
\begin{align}
\zeij{\alpha,\beta}\otimes\Lambda_{\lal}^l\pd=
\hpd X_{\alpha,\beta}^{(l)}
-\sum_{b=1}^{k_L}
\caW_{\alpha,\beta,b}^{(l)}\otimes \il(G_b)\Lambda_{\lal}^l
\in \hpd\kl{l}(\bb),
\end{align}
for $l$ large enough.
%
\end{proof}
Now we give a proof of Lemma \ref{lem:cc1}
\begin{proofof}[Lemma \ref{lem:cc1}]
Suppose that $\bb$ belongs to 
$\Class_1(n,n_0,k_R,k_L)$ with respect to 
$(\lal,\bbD,\bbG,\oo)$. We would like to show that
$\bb$ belongs to $\Class(n,n_0,k_R,k_L)$
with associated quadraple $(\lal,\bbD,\bbG,\oo)$.
By Lemma \ref{lem:hcc1} and Definition \ref{def:cl1},
there exists $\l_{\bb}'\in\nan$ such that 
\begin{align}\label{eq:ulm}
&\kl{l}(\bb)\hat P^{(n_0,k_R,k_L)}_R=\mnz\otimes\spa\left\{\ir({D_{a})}
{\Lambda_{\lal}}^l\right\}_{a=0}^{k_R},\nonumber\\
&\hat P^{(n_0,k_R,k_L)}_D\kl{l}(\bb)=\mnz\otimes\spa\left\{
\il(G_{b}){\Lambda_{\lal}}^l\right\}_{b=0}^{k_L}.
\end{align}for all $l\ge l_{\bb}'$.
Here we denoted $D_0:=\prr$, $G_0:=\pll$.
Note that this corresponds to the property {\it 1.} of Lemma 7.2 PartII \cite{Ogata2}.
The only difference here is that our $B_\mu$s may have term in $\CN$.
However
the existence of such terms does not affect 
the argument in Section 7 of Part II \cite{Ogata2} and
statements of Lemma 7.3,7.6, 7.7 can be proven.
This implies that $\bbB$ is in $\Class(n,n_0,k_R,k_L)$ with the associated quadraple
$(\lal,\bbD,\bbG,\oo)$.

Next let us prove $\Class(n,n_0,k_R,k_L)\subset\Class_1(n,n_0,k_R,k_L)$.
Assume that $\bb$ belongs to $\Class(n,n_0,k_R,k_L)$
with the associated quadraple $(\lal,\bbD,\bbG,\oo)$.
We would like to show that $\bb$ belongs to 
$\Class_1(n,n_0,k_R,k_L)$ with respect to $(\lal,\bbD,\bbG,\oo)$.
(1) of Definition \ref{def:cl1} corresponds to Lemma \ref{lem:dbb} and Definition \ref{def:ldgob}.
We have to show (2) of Definition
\ref{def:cl1}.
Note that $\hpd B_{\mu}$ is of the form
\begin{align*}
\hpd B_{\mu} 
= \omega_\mu\otimes\Lambda_{\lal}\pd
+\sum_{1\le b'\le k_L}
\caY_{\mu ,b',\bb}\otimes \il(G_{b'})\Lambda_{\lal},
\end{align*}
for any $\mu=1,\ldots,n$.
Suppose that
$\bbY_{b,\bb}=(\caY_{\mu,b,\bbB})_{\mu=1}^n\in \caL(\oo,\lambda_{b}^{-1})$
for some $b\in{\mathfrak H}_{\bbG}^L$, $\mu=1,\ldots,n$.
Then there exists $J_b$
such that 
$\caY_{\mu ,b,\bb}=J_b\omega_{\mu}
-\lambda_b^{-1}\omega_{\mu}J_b$.
We can check inductively
that for any $l\in\nan$ and $\mu^{(l)}\in\{1,\ldots,n\}^{\times l}$,
$\hpd\wb{l}$ is of the form
\[
\hpd \wb{l}=\wo{l}\otimes\Lambda_{\lal}^l\pd
+\sum_{b'=1}^{k_L}
\caZ_{\mu^{(l)},b'}\otimes \il(G_{b'})\Lambda_{\lal}^l,
\]
with $\caZ_{\mu^{(l)},b'}\in\mnz$
and in particular, we have
\[
\caZ_{\mu^{(l)},b}
=J_b\wo{l}
-\lambda_b^{-l}\wo{l}J_b.
\]
Now, as $\bbB\in\Class(n,n_0,k_R,k_L)$,
for $l\in\nan$ large enough,
there exist coefficients 
$\{\alpha_{\mu^{(l)}}\}_{\mu^{(l)}\in\{1,\ldots,n\}^{\times l}}$
such that
\begin{align*}
&\unit\otimes \il(G_b)\Lambda_{\lal}^l=
\hpd \sum_{\mu^{(l)}\in\{1,\ldots,n\}^{\times l}}
\alpha_{\mu^{(l)}}
\wb{l}\\
&=\sum_{\mu^{(l)}\in\{1,\ldots,n\}^{\times l}}
\alpha_{\mu^{(l)}}
\wo{l}\otimes\Lambda_{\lal}^l\pd
+\sum_{b'=1}^{k_L}
\sum_{\mu^{(l)}\in\{1,\ldots,n\}^{\times l}}
\alpha_{\mu^{(l)}}
\caZ_{\mu^{(l)},b'}\otimes \il(G_{b'})\Lambda_{\lal}^l.
\end{align*}
From the linearly independence of $\{G_b\}_{b=1}^{k_L}$,
we obtain
\[
\sum_{\mu^{(l)}\in\{1,\ldots,n\}^{\times l}}
\alpha_{\mu^{(l)}}
\wo{l}=0,\quad
\unit=\sum_{\mu^{(l)}\in\{1,\ldots,n\}^{\times l}}
\alpha_{\mu^{(l)}}
\caZ_{\mu^{(l)},b}
=\sum_{\mu^{(l)}\in\{1,\ldots,n\}^{\times l}}
\alpha_{\mu^{(l)}}\lmk
J_b\wo{l}
-\lambda_b^{-l}\wo{l}J_b\rmk=
0.\]
This is a contradiction.
Therefore, we have
$\bbY_{b,\bb}\notin \caL(\oo,\lambda_{b}^{-1})$
for all $b\in{\mathfrak H}_{\bbG}^L$.
Similarly, we have 
$\bbX_{a,\bb}\notin \caL(\oo,\lambda_{-a})$
for all $a\in {\mathfrak H}_{\bbD}^R$.
Hence the condition (2) of Definition \ref{def:cl1}
holds.
\end{proofof}

\subsection{Proof of Proposition \ref{classthm}}
In this subsection, we prove 
Proposition \ref{classthm}.
For this, we prove the following Lemma.
\begin{lem}\label{lem:ccpath}
Let $n,n_0\in\nan$ with $n\ge 2$ and $k_R,k_L\in\nan\cup\{0\}$.
Then for any $\bb_0,\bb_1\in \Class(n,n_0,k_R,k_L)$,
there exists a continuous and piecewise
$C^\infty$-map
$\bb:[0,1]\to\lmk \mnz\otimes\mkk\rmk^{\times n}$
such that $\bb(0)=\bb_0$, $\bb(1)=\bb_1$
and $\bb(t)\in\Class(n,n_0,k_R,k_L)$
for all $t\in[0,1]$.
\end{lem}

First we connect each element in 
$\CL$ to an element in $\CLn$, defined as follows.
\begin{defn}
Let $n,n_0\in\nan$ with $n\ge 2$, and $k_R,k_L\in\nan\cup\{0\}$.
We say
$\bb\in\lmk\mnz\otimes\mkk\rmk^{\times n}$
belongs to $\CLn$ if 
$\bb\in\Class_1(n,n_0,k_R,k_L)(=\Class(n,n_0,k_R,k_L))$
with respect to some 
$(\lal,\bbD,\bbG,\oo)$,
and the associated set of matrices 
$\{\caX_{\mu,a,\bb}\}_{a\mu},\{\caY_{\mu,b,\bb}\}_{b,\mu}$
(see Remark \ref{rem:xy})
satisfy
$
\bbX_{a,\bb}=(\caX_{\mu,a,\bb})_{\mu=1}^n\notin \caL(\oo,\lambda_{-a}),
\bbY_{b,\bb}=(\caY_{\mu,b,\bb})_{\mu=1}^n\notin \caL(\oo,\lambda_{b}^{-1})$
for any $1\le a\le k_R$
and $1\le b\le k_L$.
\end{defn}

\begin{lem}\label{lem:c02}
Let $n,n_0\in\nan$ with $n\ge 2$, $k_R,k_L\in\nan\cup\{0\}$, and 
$\bb_0\in\CL$.
Then there exist $\bb_1\in\CLn$ and
a 
$C^\infty$-map
$\bb:[0,1]\to\lmk \mnz\otimes\mkk\rmk^{\times n}$
such that $\bb(0)=\bb_0$, $\bb(1)=\bb_1$
and $\bb(t)\in\Class(n,n_0,k_R,k_L)$
for all $t\in[0,1]$.
\end{lem}
\begin{proof}
Let $(\lal_0,\bbD_0,\bbG_0,\oo_0)$
be the quadruplet associated with $\bb_0
\in \Class(n,n_0,k_R,k_L)$.
Let
$\{\caX_{0, \mu,a}\},\{\caY_{0,\mu,b}\}$,
be the set of matrices associated with $\bb_0$.
Set $\caX_{1,\mu,a}:=\caX_{0,\mu,a}$
and $\caY_{1,\mu,b}:=\caY_{0,\mu,b}$
for each $a\in \hfu{\bbD_0}$ and $b\in \hfd{\bbG_0}$.
As $\bb_0\in\Class_1(n,n_0,k_R,k_L)$,
by Lemma \ref{lem:cc1},
we have $
(\caX_{0,\mu,a})_{\mu}=(\caX_{1,\mu,a})_\mu\notin \caL(\oo_0,\lambda_{-a})
$ for $a\in \hfu{\bbD_0}$, and
$(\caY_{0,\mu,b})_\mu=(\caY_{1,\mu,b})_\mu\notin 
\caL(\oo_0,\lambda_{b}^{-1})$, for 
$b\in \hfd{\bbG_0}$.
Note that for any $\lambda\in\cc$,
$\dim\caL(\oo_0,\lambda)
\le n_0^2<\dim(\bigoplus_{\mu=1}^n\mnz)=n_0^2 n$,
as $n\ge 2$.
Therefore, for each $a\in \{1,\ldots,k_R\}\setminus \hfu{\bbD_0}$
and $b\in \{1,\ldots, k_L\}\setminus \hfd{\bbG_0}$, we can find
$(\caX_{1,\mu,a})_\mu\in \lmk \bigoplus_{\mu=1}^n\mnz\rmk \setminus \caL(\oo_0,\lambda_{-a})$,
$(\caY_{1,\mu,b})_\mu\in \lmk \bigoplus_{\mu=1}^n\mnz\rmk \setminus\caL(\oo_0,\lambda_{b}^{-1})$.

For each $t\in[0,1]$,
define $\bb(t)\in(\mnz\otimes\mkk)^{\times n}$
by
\begin{align*}
B_{\mu}(t)
:=&\omega_{0,\mu}\otimes\Lambda_{\lal_0}
+\sum_{a=1}^{k_R}
\lmk (1-t)\caX_{0,\mu,a}+t \caX_{1,\mu,a})\rmk
\otimes \ir(D_{0,a})\Lambda_{\lal_0}\\
&+\sum_{b=1}^{k_L}
\lmk (1-t)\caY_{0,\mu,b}+t \caY_{1,\mu,b})\rmk
\otimes \il(G_{0,b})\Lambda_{\lal_0}
+\overline{\hpd} B_{0,\mu}\overline{\hpu}
,\quad
\quad \mu=1,\ldots, n.
\end{align*}
Clearly, 
$\bb:[0,1]\to\lmk \mnz\otimes\mkk\rmk^{\times n}$is
a 
$C^\infty$-map
such that $\bb(0)=\bb_0$.
It is easy to check for all $t\in[0,1]$
that
$\bb(t)\in\Class_1(n,n_0,k_R,k_L)=\Class(n,n_0,k_R,k_L)$
such that
\begin{align*}
&\lal_{\bb(t)}=\lal_{0},\;
\bbD_{\bb(t)}=\bbD_{0},\;
\bbG_{\bb(t)}=\bbG_{0},\;
\oo_{\bb(t)}=\oo_{0},\\
&\caX_{\mu,a,\bb(t)}=(1-t)\caX_{0,\mu,a}+t\caX_{1,\mu,a},\;
\caY_{\mu,b,\bb(t)}=(1-t)\caY_{0,\mu,b}+t\caY_{1,\mu,b}.
\end{align*}
As we have 
$(\caX_{1,\mu,a})_\mu\notin \caL(\oo_0,\lambda_{-a})$,
$(\caY_{1,\mu,b})_\mu\notin \caL(\oo_0,\lambda_{b}^{-1})$,
$\bb_1:=\bb(1)$ belongs to $\CLn$.
\end{proof}

Next we consider elements in $\CLn$
given as follows.
\begin{lem}\label{lem:tb}
Let $n,n_0\in\nan$ with $n\ge 2$, $k_R,k_L\in\nan\cup\{0\}$, 
and $\bb\in\CLn$.
Let $(\lal,\bbD,\bbG,\oo)$
be the quadruplet associated with $\bb$.
Let $\boldsymbol{\caX},\boldsymbol{\caY}$
be the set of matrices associated with $\bb$.
Define $\tilde \bb=(\tilde B_1,\ldots,
\tilde B_n)\in(\mnz\otimes\mkk)^{\times n}$
by
\[
\tilde B_{\mu}=
\omega_{\mu}\otimes{\Lambda_{\lal}}
+\sum_{a=1}^{k_R}
 \caX_{\mu,a}
\otimes \eij{-a,0}{\Lambda_{\lal}}
+\sum_{b=1}^{k_L}
\caY_{\mu,b}
\otimes \eij{0,b}{\Lambda_{\lal}}
,\quad
\mu=1,\ldots,n.
\]
Then $\tilde \bb\in\CLn$, with
\begin{align*}
&\lal_{\tilde \bb}=\lal,\;
\bbD_{\tilde \bb}=(\eijr{-1,0},\ldots,\eijr{-k_R,0}),\;
\bbG_{\tilde \bb}=(\eijl{0,1},\ldots,\eijl{0,k_L}),\;
\oo_{\tilde \bb}=\oo,\\
&\boldsymbol{\caX}_{\tilde \bb}=
\boldsymbol{\caX},\;
\boldsymbol{\caY}_{\tilde \bb}=
\boldsymbol{\caY}.
\end{align*}
\end{lem}
\begin{proof}
As $\bb\in\CLn$, we have
$\caX_{\mu,a}\notin \caL(\oo,\lambda_{-a})$,
$\caY_{\mu,b}\notin \caL(\oo,\lambda_{b}^{-1})$.
With this observation, it is easy to check 
that $\tilde \bb\in\CLn$.
\end{proof}
\begin{defn}
For $\bb\in\CLn$, 
we denote the $\tilde \bb$
given in Lemma \ref{lem:tb}
by $\bbS_{\bb}:=\tilde \bb$.
\end{defn}

Next we connect $\bb_0\in\CLn$ to $\bbS_{\bb_0}$.
\begin{lem}\label{lem:bsb}
Let $n,n_0\in\nan$ with $n\ge 2$, $k_R,k_L\in\nan\cup\{0\}$, and 
$\bb_0\in\CLn$.
Set $\bb_1:=\bbS_{\bb_0}$.
Then there exists 
a 
$C^\infty$-map
$\bb:[0,1]\to\lmk \mnz\otimes\mkk\rmk^{\times n}$
such that $\bb(0)=\bb_0$, $\bb(1)=\bb_1$
and $\bb(t)\in\Class_2(n,n_0,k_R,k_L)$
for all $t\in[0,1]$.
\end{lem}
\begin{proof}
Let $(\lal,\bbD,\bbG,\oo)$
be the quadruplet associated with $\bb_0$.
Let $\boldsymbol{\caX},\boldsymbol{\caY}$
be the set of matrices associated with 
$\bb_0$.

Define for each $t\in[0,1]$,
$1\le a\le k_R$, and $1\le b\le k_L$,
\begin{align*}
D_a(t):=\eijr{-a,0}+(1-t)D_{a}\qur{-1},\quad
G_b(t):=\eijl{0,b}+(1-t)\qdl{1}G_b.
\end{align*}
We consider $\bbD(t):=(D_1(t),\ldots,D_{k_R}(t))$ and
$\bbG(t):=(G_1(t),\ldots,G_{k_L}(t))$.

We have
$\bbD(0)=\bbD$ and $\bbG(0)=\bbG$.
We also have
$\bbD(1)=(
\eijr{-1,0},\eijr{-2,0},\cdots,\eijr{-{k_R},0})$ 
and $\bbG(1)=(
\eijl{0,1},\eijl{0,2},\cdots,\eijl{0,k_L})$.
Furthermore, we have
\begin{align*}
&\Lambda_\lal \il\lmk G_{b}(t)\rmk
=\Lambda_\lal \lmk
\eij{0,b}+(1-t)\qd{1}\il\lmk G_b\rmk
\rmk\\
&=\lambda_{b}^{-1}
\lmk
\eij{0,b}+(1-t)\qd{1}\il\lmk G_b\rmk
\rmk{\Lambda_{\lal}}
=\lambda_{b}^{-1}
\il\lmk G_b(t)\rmk \Lambda_\lal,
\end{align*}
for $1\le b\le k_L$ and $t\in[0,1]$.
Similarly, we have
$\Lambda_\lal \ir\lmk D_a(t)\rmk=\lambda_{-a}\ir\lmk D_a(t)\rmk \Lambda_\lal$
for $1\le a\le k_R$ and $t\in[0,1]$.

We claim
$\bbD(t):=(D_1(t),\ldots,D_{k_R}(t))\in
\caC^R(k_R)$, and
$\bbG(t):=(G_1(t),\ldots,G_{k_L}(t))\in
\caC^L(k_L)$, for each $t\in[0,1]$.
We check $\bbG(t):=(G_1(t),\ldots,G_{k_L}(t))\in
\caC^L(k_L)$. That $\bbD(t)\in\caC^R(k_R)$ 
can be checked in the same way.
It is clear from the definition
that $\eijl{00}G_b(t)=\eijl{0b}$
and $G_b(t)\in \UT_{0,k_L+1}$.

We have to check the condition 2.
of Definition 1.7 Part I \cite{Ogata1}.
Take arbitrary $1\le b_1,b_2\le k_L$.
Note from Lemma 3.3 \cite{Ogata1} that
\begin{align}\label{eq:ggen}
&G_{b_1}G_{b_2}
=\lmk \eijl{0,b_1}+\qdl{1}G_{b_1}\rmk
\lmk \eijl{0,b_2}+\qdl{1}G_{b_2}\rmk\nonumber\\
&=\eijl{0,b_1}G_{b_2}
+\qdl{1}G_{b_1}G_{b_2}.
\end{align}
Recall Lemma \ref{lem:ad}.
Suppose that there exist $\sigma(b_1, b_2)\in\{1,\ldots,k_L\}$
and nonzero $\kappa(b_1,b_2)\in\cc$ such that
\begin{align}\label{eq:gg}
G_{b_1}G_{b_2}=\kappa(b_1,b_2)G_{\sigma(b_1,b_2)}.
\end{align}
The equations(\ref{eq:ggen}), (\ref{eq:gg}), 
imply
\begin{align}\label{eq:zgg}
\kappa(b_1,b_2)\eijl{0,\sigma(b_1,b_2)}=
\kappa(b_1,b_2)\eijl{0,0}G_{\sigma(b_1,b_2)}
=\eijl{0,0}G_{b_1}G_{b_2}
=\eijl{0,b_1}G_{b_2},
\end{align}
and
\begin{align}\label{eq:igg}
\kappa(b_1,b_2)\qdl{1}G_{\sigma(b_1,b_2)}
=\qdl{1}G_{b_1}G_{b_2}.
\end{align}
Substituting (\ref{eq:zgg}) and (\ref{eq:igg}),
we have
\begin{align*}
&G_{b_1}(t)G_{b_2}(t)
=\lmk\eijl{0,b_1}+(1-t)\qdl{1}G_{b_1}\rmk
\lmk\eijl{0,b_2}+(1-t)\qdl{1}G_{b_2}\rmk\\
&=(1-t)\eijl{0,b_1}G_{b_2}
+(1-t)^2\qdl{1}G_{b_1}G_{b_2}
=(1-t)\lmk
\kappa(b_1,b_2)\eijl{0\sigma(b_1,b_2)}
+(1-t)\kappa(b_1,b_2)\qdl{1}G_{\sigma(b_1,b_2)}
\rmk\\
&=(1-t)\kappa(b_1,b_2)
G_{\sigma(b_1,b_1)}(t).
\end{align*}
Hence we have $G_{b_1}(t)G_{b_2}(t)\in\spa\{G_b(t)\}_{b=1}^{k_L}$, $t\in [0,1]$, when $G_{b_1}G_{b_2}\neq 0$.

Suppose $G_{b_1}G_{b_2}=0$.
Then 
from (\ref{eq:ggen}),
we have
\begin{align*}
\eijl{0,b_1}G_{b_2}=0,\quad
\qdl{1}G_{b_1}G_{b_2}=0.
\end{align*}
Substituting this, we have
\begin{align*}
G_{b_1}(t)G_{b_2}(t)
=\lmk\eijl{0,b_1}+(1-t)\qdl{1}G_{b_1}\rmk
\lmk\eijl{0,b_2}+(1-t)\qdl{1}G_{b_2}\rmk
=
0.
\end{align*}
This proves $G_{b_1}(t)G_{b_2}(t)\in\spa\{G_b(t)\}_{b=1}^{k_L}$, $t\in [0,1]$, when $G_{b_1}G_{b_2}=0$.
Hence we obtain $\bbG(t)\in
\caC^L(k_L)$.

Now we define the path $\bb(t)$.
We set $\bb(t):=(B_1(t),\ldots, B_{n}(t))$
with
\begin{align*}
B_{\mu}(t)
:=\omega_{\mu}\otimes{\Lambda_{\lal}}
+\sum_{a=1}^{k_R}
\caX_{\mu,a}
\otimes \ir\lmk D_{a}(t)\rmk {\Lambda_{\lal}}+\sum_{b=1}^{k_L}
\caY_{\mu,b}
\otimes \il\lmk G_{b}(t)\rmk {\Lambda_{\lal}}
+(1-t)\overline{\hpd} B_{\mu}\overline{\hpu},
\end{align*}
for $\mu=1,\ldots, n$ and $t\in[0,1]$.
It is trivial that $\bb:[0,1]\to (\mnz\otimes\mkk)^{\times n}$
is $C^\infty$.
Furthermore,
we have
$\bb(0)=\bb_0$, and $\bb(1)=\bbS(\bb)=\bb_1$.

For all $t\in[0,1]$, it is easy to check
$\bb(t)\in\Class_2(n,n_0,k_R,k_L)$
with
$\lal_{\bb(t)}=\lal$,
$\bbD_{\bb(t)}=\bbD(t)$,
$\bbG_{\bb(t)}=\bbG(t)$, $\oo_{\bb(t)}=\oo$,
$\boldsymbol{\caX_{\bb(t)}}=\boldsymbol{\caX}$,
and $\boldsymbol{\caY_{\bb(t)}}=\boldsymbol{\caY}$.
\end{proof}

Now we connect elements of $\CLn$.
\begin{lem}\label{lem:c2c2}
Let $n,n_0\in\nan$ with $n\ge 2$ and $k_R,k_L\in\nan\cup\{0\}$. 
Then for any $\bb_0,\bb_1\in \Class_2(n,n_0,k_R,k_L)$,
there exists a continuous and piecewise
$C^\infty$-map
$\bb:[0,1]\to\lmk \mnz\otimes\mkk\rmk^{\times n}$
such that $\bb(0)=\bbS_{\bb_0}$, $\bb(1)=\bbS_{\bb_1}$
and $\bb(t)\in\Class_2(n,n_0,k_R,k_L)$
for all $t\in[0,1]$.
\end{lem}
\begin{proof}
Let $(\lal_i,\bbD_i,\bbG_i,\oo_i)$
be the quadruplet associated with $\bb_i$,
and  $\boldsymbol{\caX}_i,\boldsymbol{\caY}_i$
the set of matrices associated with $\bb_i$,
for $i=0,1$.

From \cite{bo}, Theorem 5.1,
there exists a continuous and piecewise
$C^{\infty}$-map
$\tilde \oo:[0,1]\to\lmk \mnz\rmk^{\times n}$
such that $\tilde \oo(0)=\oo_0$, $\tilde
\oo(1)=\oo_1$
and $\tilde \oo(t)\in\Primz(n,n_0)$.
(Note that the proof there provides continuous and
 piecewise $C^\infty$-path.)
By Appendix D of \cite{bo}, we see that
$T_{\tilde \oo(t)}$ satisfies the conditions 
of Lemma \ref{lem:cinf}.
In particular, we have $r_{T_{\tilde \oo(t)}}>0$ and
$\oo(t):=r_{T_{\tilde \oo(t)}}^{-\frac 12}\tilde\oo(t)$
belongs to $\Prim(n,n_0)$. By Lemma \ref{lem:cinf}, the path
$ \oo:[0,1]\to\lmk \mnz\rmk^{\times n}$
is continuous and piecewise
$C^\infty$, and $\oo(0)=\oo_0$, $
\oo(1)=\oo_1$.

Furthermore, there exists a
continuous and piecewise $C^\infty$-path
$\lal:[0,1]\to\bbC^{k_R+k_L+1}$
with  $\lal(0)=\lal_0$, $\lal(1)=\lal_1$
such that
$\lal(t)\in\Wo'(k_R,k_L)$
for all $t\in[0,1]$.
To see this, choose ${{\boldsymbol r}}:=(r_j)_{j=-k_R}^{k_L}\in\bbC^{k_R+k_L+1}$ such that
$r_0=1$, $0<r_{-k_R}<\cdots<r_{-1}<1$, and $0<r_{k_L}<\cdots<r_1<1$.
Decompose each $\lambda_{j,i}\in\cc$ as 
$\lambda_{j,i}=e^{i\theta_{j,i}}\lv \lambda_{j,i}\rv$, $\theta_{j,i}\in[0,2\pi)$,
for each $i=0,1$ and $j=-k_R,\ldots, k_L$.
Define $\lal(t):=\lmk\lambda_j(t)\rmk_{j=-k_R}^{k_L}$, $t\in[0,1]$ by
\begin{align}
\lambda_{j}(t)
:=\left\{
\begin{gathered}
e^{i\theta_{j,0}}\lmk (1-3t)\lv \lambda_{j,0}\rv
+3t r_j\rmk
,\quad t\in\left[0,\frac13\right],\\
e^{i\lmk
(-3t+2)\theta_{j,0}+(3t-1)\theta_{j,1}
\rmk}r_j,\quad t\in\left[\frac13,\frac 23\right],\\
e^{i\theta_{j,1}}\lmk (3t-2)\lv \lambda_{j,1}\rv
+3(1-t) r_j\rmk
,\quad t\in\left[\frac 23,1\right].
\end{gathered}
\right.
\end{align}
It is straight forward to check that $\lal(t):=(\lambda_{-k_R}(t),\ldots,\lambda_{k_L}(t))$
satisfies the claimed properties.

Next we consider paths of vectors in $\bigoplus_{\mu=1}^n\mnz$.
For each $\alpha,\beta\in n_0$,
$a=1,\ldots, k_R$, $b=1,\ldots,k_L$ and $t\in[0,1]$,
define
\begin{align*}
\zeta^R_{\alpha,\beta,a}(t):=\Delta_{\oo(t),\lambda_{-a}(t)}\lmk\zeij{\alpha,\beta}\rmk,\quad
\zeta^L_{\alpha,\beta,b}(t):=\Delta_{\oo(t),\lambda_{b}(t)^{-1}}\lmk\zeij{\alpha,\beta}\rmk.
\end{align*}
These define continuous and piecewise $C^\infty$-paths $\zeta^R_{\alpha,\beta,a}$, $\zeta^L_{\alpha,\beta,b}$
in $\bigoplus_{\mu=1}^n\mnz$.
By Lemma \ref{lem:Delta},  for each $a=1,\ldots, k_R$(resp.  $b=1,\ldots,k_L$) and $t\in[0,1]$,
$\{\zeta^R_{\alpha,\beta,a}(t)\}_{\alpha,\beta}$ (resp. $\{\zeta^L_{\alpha,\beta,b}(t)\}_{\alpha,\beta}$ )
are linearly independent.
Regarding $\bigoplus_{\mu=1}^n\mnz$ a Hilbert space with inner product
$\braket{\bigoplus_{\mu} a_{\mu}}{\bigoplus_{\mu} b_{\mu}}=\sum_{\mu}\Tr a_{\mu}^*b_{\mu}$,
let $P_a(t)$ be the orthogonal projection onto $\spa\{\zeta^R_{\alpha,\beta,a}(t)\}_{\alpha,\beta}$ for each $a=1,\ldots,k_R$.

We have $(1-P_a(0))(\caX_{0,\mu,a})_{\mu=1}^n\neq 0$ and 
$(1-P_a(1))(\caX_{1,\mu,a})_{\mu=1}^n\neq 0$ for all $a=1,\ldots,k_R$ because $\bb_0,\bb_1\in
\Class_2(n,n_0,k_R,k_L)$.
Applying Lemma \ref{lem:cio} to $\spa\{\zeta^R_{\alpha,\beta,a}(t)\}_{\alpha,\beta}$,
we obtain
a continuous and piecewise $C^\infty$-path $\bbX_a:[0,1]\to \bigoplus_{\mu=1}^n\mnz$
such that $(1-P_a(t))\bbX_a(t)\neq 0$ for all $t\in[0,1]$ and $\bbX_a(0)=(\caX_{0,\mu,a})_{\mu=1}^n$, $\bbX_a(1)=(\caX_{1,\mu,a})_{\mu=1}^n$ for each $a=1,\ldots,k_R$.
In other words, we obtain a continuous and piecewise $C^\infty$-path $[0,1]\ni t \mapsto \bbX_a(t)=(\caX_{\mu,a}(t))_{\mu=1}^n\in\lmk \bigoplus_{\mu=1}^n\mnz\rmk \setminus \caL(\oo(t),\lambda_{a}(t))$
such that $\bbX_a(0)=\bbX_{0,a}$ and $\bbX_a(1)=\bbX_{1,a}$ for each $a=1,\ldots,k_R$.
Similarly we obtain a continuous and piecewise $C^\infty$-path $[0,1]\ni t \mapsto \bbY_a(t)=(\caY_{\mu,b}(t))_{\mu=1}^n\in \lmk \bigoplus_{\mu=1}^n\mnz\rmk \setminus\caL(\oo(t),\lambda_{b}^{-1}(t))$
such that $\bbY_b(0)=\bbY_{0,b}$ and $\bbY_b(1)=\bbY_{1,b}$ for each $b=1,\ldots,k_L$.

Now we define the path $\bbB(t)$.
We set $\bb(t):=(B_1(t),\ldots, B_{n}(t))$
with
\begin{align*}
B_{\mu}(t)
:=\omega_{\mu}(t)\otimes{\Lambda_{\lal(t)}^{(k_R,k_L)}}
+\sum_{a=1}^{k_R}
\caX_{\mu,a}(t)
\otimes \eij{-a,0}{\Lambda_{\lal(t)}^{(k_R,k_L)}}+\sum_{b=1}^{k_L}
\caY_{\mu,b}(t)
\otimes \eij{0,b}{\Lambda_{\lal(t)}^{(k_R,k_L)}}
\end{align*}
for $\mu=1,\ldots, n$ and $t\in[0,1]$.

It is trivial that $\bb:[0,1]\to (\mnz\otimes\mkk)^{\times n}$
is continuous and piecewise $C^\infty$.
Furthermore,
we have
$\bb(0)=\bbS_{\bb_0}$ and $\bb(1)=\bbS_{\bb_1}$.
For all $t\in[0,1]$, it is easy to check
$\bb(t)\in\Class_2(n,n_0,k_R,k_L)$
with
$\lal_{\bb(t)}=\lal(t)$,
$\oo_{\bb(t)}=\oo(t)$,
$\bbD_{\bb(t)}=(\eijr{-a,0})_{a=1}^{k_R}$,
$\bbG_{\bb(t)}=(\eijl{0,b})_{b=1}^{k_L}$,
$\boldsymbol{\caX_{\bb(t)}}=\boldsymbol{\caX(t)}$,
and $\boldsymbol{\caY_{\bb(t)}}=\boldsymbol{\caY(t)}$.
\end{proof}
\begin{proofof}[Lemma \ref{lem:ccpath}]
For $\bbA_0,\bbA_1\in\Class(n,n_0,k_R,k_L)$, we denote
$\bbA_0\thickapprox\bbA_1$ if there exists a continuous and
piecewise $C^\infty$-map $\bbA:[0,1]\to(\mnz\otimes\mkk)^{\times n}$
such that $\bbA(0)=\bbA_0$, $\bbA(1)=\bbA_1$
and $\bbA(t)\in\Class(n,n_0,k_R,k_L)$
for all $t\in[0,1]$.
Let $\bb_0,\bb_1\in\Class(n,n_0,k_R,k_L)$.
By Lemma \ref{lem:c02}, there exists $\bb_0',\bb_1'\in\CLn$
such that $\bb_0\thickapprox \bb_0'$,
$\bb_1\thickapprox \bb_1'$.
Lemma \ref{lem:bsb} implies 
$\bb_0'\thickapprox\bbS_{\bb_0'}$,
$\bb_1'\thickapprox\bbS_{\bb_1'}$.
Finally, Lemma \ref{lem:c2c2} implies
$\bbS_{\bb_0'}\thickapprox \bbS_{\bb_1'}$,
proving Lemma \ref{lem:ccpath}.
\end{proofof}
\begin{proofof}[Proposition \ref{classthm}]
Let $\bb_0,\bb_1\in\Class(n,n_0,k_R,k_L)$.
Let $2n_0^6(k_R+1)(k_L+1)\le m_0,m_1\in\nan$.
By Lemma \ref{lem:ccpath},
there exists a continuous and piecewise $C^{\infty}$-path 
$\bb:[0,1]\to \lmk\mnz\otimes\mkk\rmk^{\times n}$
such that $\bb(0)=\bb_0$, $\bb(1)=\bb_1$
and $\bb(t)\in\Class(n,n_0,k_R,k_L)$, $t\in[0,1]$.
From Lemma 3.6 Part I \cite{Ogata1}, 
$(n,n_0(k_L+k_R+1),\hpu,\hpd,\bbB(t))$ satisfies {\it Condition 2}.
By Lemma 3.7 Part I \cite{Ogata1},
$(n,n_0(k_L+k_R+1),\hpu,\hpd,\bbB(t))$ satisfies {\it Condition 3}
for $l_{\bb(t)}=l_{\bb(t)}(n,n_0,k_R,k_L,\lal(t),\bbD(t), \bbG(t),0)$.
By Lemma 3.8 of \cite{Ogata1}, $(n,n_0(k_L+k_R+1),\bbB(t))$
satisfies {\it Condition 4} for $(l_{\bb(t)}, l_{\bb(t)})$.
These facts and the estimate $l_{\bb(t)}\le n_0^6(k_R+1)(k_L+1)$ from
 Lemma \ref{lem:klmain} gurantee the condition of Proposition \ref{prop:maingen} with $m_1=m_2=m_3=n_0^6(k_R+1)(k_L+1)$.
From Proposition \ref{prop:maingen}, the paths 
$[0,1]\ni t\to \Gamma_{m,\bbB(t)}^{(R)}\lmk \zeij{\alpha,\beta}\otimes \eij{-a,b}\rmk$
satisfies {\it Condition 5}.
From Lemma \ref{lem:vpath}, this implies Proposition \ref{classthm}. 
\end{proofof}

\section{Classification of $\caH(n)$}\label{sec:first}
In this section, we prove Theorem \ref{singthm}.
Given Proposition \ref{classthm}, it suffices to find a concrete $C^1$-path of gapped Hamiltonians connecting 
$H_{\Phi_{m_1,\bbB_1}}$ with
$\bbB_1\in\Class(n,n_0,k_R,k_L)$  and $H_{\Phi_{m_2,\bbB_2}}$ with
$\bbB_2\in\Class(n,1,n_0(k_R+1)-1,n_0(k_L+1)-1)$. 

\subsection{The path with singularity}
We introduce several notations.
Throughout Section \ref{sec:first}, we fix
numbers $0<\kappa<1$ and $\theta\in 2\pi \lmk \bbR\setminus \bbQ\rmk$. 
Let $2\le n_0\in\nan$, $k_R,k_L\in\nan\cup\{0\}$.
Define vectors
${\boldsymbol r}=(r_\alpha)_{\alpha=1}^{n_0}\in\cc^{n_0}$,
$\lal_L=(\lambda_{L,\alpha})_{\alpha=-(n_0-1)}^{n_0-1}\in\cc^{2n_0-1}$, 
and $\lal_{R}=(\lambda_{R,i})_{i=-k_R}^{k_L}
\in\cc^{k_R+k_L+1}$
 by 
\begin{align}\label{eq:rd}
&r_{\alpha}:=\lmk \kappa e^{i\theta}\rmk^{\alpha-1},\quad
\text{for }\alpha=1,\ldots,n_0,\notag\\
&\lambda_{L,\alpha}=\lmk \kappa e^{i\theta}\rmk^{|\alpha|},\quad
\text{for }\alpha=-(n_0-1),\ldots, -1,0,1,\ldots,(n_0-1),\notag\\
&\lambda_{R,j}:=\lmk \kappa e^{i\theta}\rmk^{|j|n_0},\quad
\text{for }j=-k_R,\ldots, -1,0,1,\ldots,k_L.
\end{align}
Note that we have
$\lal_L\in\Wo'(n_0-1,n_0-1)$
and $\lal_R\in\Wo'(k_R,k_L)$.
We also define diagonal matrices
\begin{align*}
R:=\sum_{\alpha=1}^{n_0}
r_{\alpha}\zeij{\alpha,\alpha}\in\mnz,\quad
\Lambda_{\lal_L}:=
\sum_{\alpha=-(n_0-1)}^{n_0-1}\lambda_{L,\alpha}
\eijz{\alpha,\alpha}\in \Mat_{2n_0-1},\quad
\Lambda_{\lal_R}:
=\sum_{i=-k_R}^{k_L}\lambda_{R,i}\eij{ii}\in \mkk.
\end{align*}
We define vectors 
\begin{align*}
&\eta_0^{(n_0)}:=\sum_{\alpha=2}^{n_0}
\cnz{\alpha}\in \bbC^{n_0},\\
&\ezu:
=\sum_{\alpha=-(n_0-1)}^{-1}{\fiz{\alpha}},\quad
\ezd:
=\sum_{\beta=1}^{n_0-1}{\fiz{\beta}}\in \bbC^{2n_0-1},\\
&\eu:=\sum_{i=-k_R}^{-1}{\fii{i}},\quad
\ed:=\sum_{i=1}^{k_L}{\fii{i}}\in \bbC^{k_L+k_R+1}.
\end{align*}(Recall the definition of $\fii{i}$ from Appendix A of \cite{Ogata1}.)

From these vectors, we define
\begin{align}\label{eq:vd}
V:=\ket{\eu}\bra{\fii{0}}+\ket{\fii{0}}\bra{\ed}\in \Mat_{k_R+k_L+1},
\end{align}
and
\begin{alignat*}{3}
K_1^{(L)}&:=\ket{\cnz{1}}\bra{\eta_0^{(n_0)}},&\quad
K_1^{(R)}&:=\Lambda_{\lal_R},&\quad
\tilde K_1^{(L)}&:=\ket{\fiz{0}}\bra{\ezd},\\
K_2^{(L)}&:=R,&\quad
K_2^{(R)}&:=V \Lambda_{\lal_R},&\quad
\tilde K_2^{(L)}&:=\Lambda_{\lal_L},\\
K_3^{(L)}&:=\ket{\eta_0^{(n_0)}}\bra{\cnz{1}},&\quad
K_3^{(R)}&:=\Lambda_{\lal_R},&\quad 
\tilde K_3^{(L)}&:=\ket{\ezu}\bra{\fiz{0}}.
\end{alignat*}
We then define matrices
\begin{align*}
K_i=K_i^{(L)}\otimes K_i^{(R)}\in\mnz\otimes \mkk,\quad
\tilde K_i=\tilde K_i^{(L)}\otimes  K_i^{(R)}
\in\Mat_{2n_0-1}\otimes \mkk,
\end{align*}
for $i=1,2,3$.
We set for each $t\in[0,1]$, matrices $\tilde \omega_\mu(t)$, $\mu=1,\ldots,n$ by
\begin{align*}
\tilde \omega_1(t)
:=R,\quad
\tilde \omega_2(t)
:=\ket{\cnz{1}}\bra{\eta_0^{(n_0)}}+t\ket{\eta_0^{(n_0)}}\bra{\cnz{1}},\quad
\tilde\omega_{\mu}(t)=0,\quad \mu\ge 3.
\end{align*}
We claim that $\tilde\oo(t)$ is primitive, for $t\in(0,1]$.
To prove this, we use the same argument as in \cite{bo}.
Applying Lemma C.7 of Part I \cite{Ogata1} to the distinct numbers
to 
$\{r_{\alpha}\}_{\alpha=2}^{n_0}\cup \{r_{\alpha}^{-1}\}_{\alpha=2}^{n_0}$,
we obtain 
$\varsigma_{\alpha}
=(\varsigma_{\alpha}(j))_{j=0}^{2n_0-3}
\in\cc^{2(n_0-1)}$
$\alpha=2,\ldots, n_0$
such that
\begin{align}\label{eq:klo}
\kl{l}(\tilde \oo(t))\ni \sum_{j=0}^{2n_0-3}
\varsigma_{\alpha}(j)
\tilde\omega_1(t)^{l-1-j}
\tilde \omega_2(t)\tilde \omega_1(t)^{j}
=\zeij{1,\alpha},\\
\kl{l}(\tilde \oo(t))\ni \sum_{j=0}^{2n_0-3}
\varsigma_{\alpha}(j)\tilde \omega_1(t)^{j}
\tilde \omega_2(t)\tilde\omega_1(t)^{l-1-j}
=t \zeij{\alpha,1}\nonumber,
\end{align}
for all $t\in[0,1]$, $\alpha=2,\ldots,n_0$,
and $l\ge 2(n_0-1)$.
Hence for $t\in(0,1]$ and $l\ge 4(n_0-1)$,
we have
$\kl{l}{(\tilde\oo(t))}=\mnz$, i.e., 
$\tilde \oo(t)$ is primitive.
By Lemma C.6 of Part I \cite{Ogata1}, for all $t\in(0,1]$,
$T_{\tilde\oo(t)}$ satisfies all the conditions (1),(2),(3)
in Lemma \ref{lem:cinf}.

For $t=0$, the
pentad $(n,n_0, \zeij{11},
\unit, \tilde\oo(0))$
satisfies the {\it Condition 2} (Definition 2.3 \cite{Ogata1}).
This can be checked in the same way as the proof of
Lemma 3.6 \cite{Ogata1}, using (\ref{eq:klo}).
Then, by Lemma 2.9 \cite{Ogata1},
$T_{\tilde\oo(0)}$ satisfies all the conditions (1), (2), (3)
in Lemma \ref{lem:cinf}.
In particular, we have $r_{T_{\tilde \oo(0)}}=1>0$.

Hence we can apply Lemma \ref{lem:cinf} to $T_{\tilde \oo(t)}$
and from the latter Lemma, $[0,1]\ni t\mapsto {r_{{T_{\tilde \oo(t)}}}}\in\cc$ 
is $C^\infty$.
We define a $C^\infty$-path $[0,1]\ni t\mapsto \oo(t)\in\Mat_{n_0}^{\times n}$ by
\[
\omega_{\mu}(t):=
r_{T_{\tilde \oo(t)}}^{-\frac 12}
\tilde\omega_{\mu}(t),\quad
t\in[0,1],\quad \mu=1,\ldots,n.
\]
We have $r_{T_{ \oo(t)}}=1$, for $t\in [0,1]$.
We define a path of $n$-tuple of elements in 
$\mnz\otimes \mkk$,
$\bb(t)$ by
\begin{align}\label{eq:pbp}
B_{1}(t):=\omega_1(t)\otimes \Lambda_{\lal_R},\quad
B_{2}(t):=
\omega_2(t)\otimes \Lambda_{\lal_R}
+\lmk r_{{T_{\tilde \oo(t)}}}\rmk^{-\frac12} R\otimes V\Lambda_{\lal_R},\quad
B_{\mu}(t):=0,\quad 3\le \mu\le n,
\end{align}
for each $t\in[0,1]$.

\begin{lem}\label{lem:tspec}Let $2\le n_0\in\nan$, $k_R,k_L\in\nan\cup\{0\}$, and $\bb(t)$ be defined by (\ref{eq:pbp}).
For any $t\in(0,1]$, $\bb(t)$ belongs to 
$\Class(n,n_0,k_R,k_L)$ with
$\lal_{\bb(t)}=\lal_R$, $\oo_{\bb(t)}=\oo(t)$, 
$D_{a,\bb(t)}=\eijr{-a,0}$,and
$G_{b,\bb(t)}=\eijl{0,b}$.
In particular, by Proposition 3.1 of Part I \cite{Ogata1}, $T_{\bb(t)}$ satisfies the Spectral Property II
with respect to a triple $(s_{\bb(t)},e_{\bb(t)},\varphi_{\bb(t)})$
and $0<s_{\bb(t)}<1$,
$s(e_{\bb(t)})=\hpu$, $s(\varphi_{\bb(t)})=\hpd$.
\end{lem}
\begin{proof}
We check that $\bb(t)$ satisfies the condition of 
$\Class(n,n_0,k_R,k_L)$
with respect to $(\lal_R, \bbD:=(\eijr{-a,0})_{a=1}^{k_R},
\bbG:=(\eijl{0,b})_{b=1}^{k_L}, 0)\in\caT(k_R,k_L)$.
It is trivial that  $(\lal_R, \bbD:=(\eijr{-a,0})_{a=1}^{k_R},
\bbG:=(\eijl{0,b})_{b=1}^{k_L}, 0)$ belongs to $\caT(k_R,k_L)$,
with $\lal_R\in\Wo'(k_R,k_L)$.
Furthermore, we have
$\bb(t)\in \lmk\mnz\otimes \caD(k_R,k_L,\bbD,\bbG)\Lambda_{\lal_R}\rmk^{\times n}$,
hence
$\caK_l(\bb(t))\subset \mnz\otimes \caD(k_R,k_L,\bbD,\bbG)\Lambda_{\lal_R}^l$ for all $l\in\nan$.
We claim  
$\caK_l(\bb(t))=\mnz\otimes \caD(k_R,k_L,\bbD,\bbG)\Lambda_{\lal_R}^l$
for $l$ large enough.
As $\{r_\alpha\}_{\alpha=2}^{n_0}\cup\{\lambda_{R i}^{-1}\}_{i=-k_R}^{-1}\cup
\{ r_\alpha^{-1}\}_{\alpha=2}^{n_0}\cup\{\lambda_{R j}\}_{j=1}^{k_L}$
are distinct, by the routine argument using Lemma C.7 of Part I \cite{Ogata1},
we obtain
\[
\zeij{1\alpha}\otimes\Lambda_{\lal_R}^l,\;
\zeij{\beta,1}\otimes\Lambda_{\lal_R}^l\;,
R^l\otimes\eij{-a,0},\;
R^l\otimes\eij{0b}
\in\caK_l(\bbB(t)),
\]
for $2\le \alpha,\beta\le n_0$, $a=1,\ldots,k_R$, $b=1,\ldots,k_L$
when 
$l$ is large enough.
By multiplying these elements,
we obtain
\[
\mnz\otimes \caD(k_R,k_L,\bbD,\bbG)\Lambda_{\lal_R}^l
\subset\caK_l(\bb(t)),
\]
for $l$ large enough. This proves the claim.
\end{proof}
Therefore, for $t\in(0,1]$, the Hamiltonians $H_{m,\Phi(t)}$
are gapped by Theorem 1.18 \cite{Ogata1}.
By an analogous argument of Proposition 3.1 \cite{Ogata1}, we obtain the following for $t=0$. 
 \begin{lem}\label{lem:0spec}Let $2\le n_0\in\nan$, $k_R,k_L\in\nan\cup\{0\}$, and $\bb(t)$ be defined by (\ref{eq:pbp}).
The quadraplet $(n,n_0(k_R+k_L+1), \zeij{11}\otimes \pu, \unit_{n_0}\otimes \pd,\bb(0))$ satisfies
the {\it Condition 2} of  Definition 2.3 \cite{Ogata1}.
In particular, by Lemma 2.9 of Part I \cite{Ogata1} , $T_{\bb(0)}$ satisfies the Spectral Property II
with respect to a triple $(s_{\bb(0)},e_{\bb(0)},\varphi_{\bb(0)})$
and $0<s_{\bb(0)}<1$,
$s(e_{\bb(0)})=\zeij{00}\otimes \pu$, $s(\varphi_{\bb(0)})=\hpd$.
\end{lem}
From the last two Lemmas, we can apply Lemma \ref{lem:cinf} and obtain the following.
\begin{lem}\label{lem:s1}
There exists an $0<s_1<1$
such that 
$\sigma\lmk T_{\bbB(t)}\rmk\setminus \{1\}\subset \caB_{s_1}(0)$, $t\in[0,1]$.
\end{lem}

\subsection{The $t\downarrow 0$ limit}Throughout this subsection, let $2\le n_0\in\nan$, $k_R,k_L\in\nan\cup\{0\}$, and $\bb(t)$ be defined by (\ref{eq:pbp}).
The path $\bb(t)$ has a singularity, i.e., 
although $\bb(t)$ belongs to $\Class(n,n_0,k_R,k_L)$ for
$t\in (0,1]$, its $t\downarrow 0$ limit $\bbB(0)$
does not.
In this subsection, we consider the $t\downarrow 0$ limit of $\caG_{l,\bbB(t)}$.
To describe the limit, we need some preparation.
We define $n$-tuple $\bbA$ of elements 
in $\Mat_{2n_0-1}\otimes \mkk$ by
\begin{align}\label{eq:defa}
&A_{1}:=\Lambda_{\lal_L}\otimes \Lambda_{\lal_R},\quad
A_{2}:=\tilde K_{1}+\tilde K_{2}+\tilde K_3,\quad
A_{\mu}:=0,\quad 3\le \mu\le n.
\end{align}
We set projections
\begin{align}\label{eq:pa}
{\puuz}:=P_{R}^{(n_0-1,n_0-1)}\otimes P^{(k_R,k_L)}_R,\quad
{\pddz}:={P_{L}^{(n_0-1,n_0-1)}\otimes P^{(k_R,k_L)}_L}.
\end{align}
Furthermore, for $\alpha=-(n_0-1),\ldots,n_0-1$ and $i=-k_R,\ldots,k_L$, we denote
\begin{align}\label{eq:gaa}
\gaa{\alpha}{i}:=\fiz{\alpha}\otimes \fii{i}.
\end{align}

For each $l,k\in\nan$ with $k\le l$ and
$i_1,i_2,\ldots,i_k\in \{1,2,3\}$,
define $\zeta_{i_1\cdots i_k}^{(k)l}
:\mnz\otimes
\mkk\to \bigotimes_{i=0}^{l-1}\cc^n$
by
\begin{align}\label{eq:zf}
&\zeta_{i_1\cdots i_k}^{(k)l}(X)
:=
\sum_{1\le m_1<\cdots<m_k\le l}
\Tr\lmk
X\lmk B_1(0)^{m_1-1}K_{i_1}B_1(0)^{m_2-m_1-1}K_{i_2}
\cdots
B_1(0)^{m_k-m_{k-1}-1}K_{i_k} B_1(0)^{l-m_k}\rmk^*
\rmk\nonumber\\
&\quad\quad\quad\quad\quad\quad \psi_1^{\otimes m_1-1}\otimes
\psi_{2}\otimes\psi_{1}^{\otimes m_2-m_1-1}
\otimes \psi_{2}\otimes 
\cdots \psi_1^{\otimes m_k-m_{k-1}-1}\otimes \psi_{2}
\otimes \psi_1^{\otimes l-m_k},
\end{align}
for $X\in \mnz\otimes\mkk$.
For each $4\le l\in\nan$, define 
$\Upsilon_l:\mkk\to\otimes_{i=0}^{l-1}\cc^n$
by
\begin{align}
\Upsilon_l(x):=&
\lmk
\zeta_{13}^{(2)(l)}
+\zeta_{213}^{(3)(l)}+\zeta_{123}^{(3)(l)}+\zeta_{132}^{(3)(l)}
+
\zeta_{2213}^{(4)(l)}+\zeta_{2123}^{(4)(l)}+\zeta_{2132}^{(4)(l)}+
\zeta_{1223}^{(4)(l)}+\zeta_{1232}^{(4)(l)}+\zeta_{1322}^{(4)(l)}
\rmk
\lmk\zeij{11}\otimes x\rmk,\nonumber\\
& \quad\quad\quad\quad\quad \quad\quad\quad\quad\quad \quad\quad\quad\quad\quad \quad\quad\quad\quad\quad \quad\quad\quad\quad\quad x\in\mkk \label{eq:el1}.
\end{align}

For $l\in\nan$,
$\alpha,\beta\in\{1,\ldots,n_0\}$,
$a=0,1,\ldots, k_R$, and $b=0,1,\ldots,k_L$, we define
\begin{align}\label{eq:sndp}
\xi_{\alpha,\beta,a,b}^{(l)}
:=
\Gamma_{l,\bbA}^{(R)}\lmk\eijz{-(\alpha-1),\beta-1} 
\otimes \eij{-a,b}\rmk
-\delta_{\alpha,\beta}(1-\delta_{\alpha,1})
 \bar r_{\alpha}^l
\Upsilon_l(\eij{-a,b}).
\end{align}
For each $l\in\nan$, we denote
 the linear subspace of $\bigotimes_{i=0}^{l-1}\bbC^n$ spanned by $\xi_{\alpha,\beta,a,b}^{(l)}$, 
$\alpha,\beta\in\{1,\ldots,n_0\}$,
$a=0,1,\ldots, k_R$, and $b=0,1,\ldots,k_L$,
by $\caD_l$.
Furthermore, we denote the orthogonal  projection onto 
$\caD_l$ by $G_l$.
, i.e.,
\begin{align}\label{eq:gldef}
G_l:=\text{orthogonal projection onto} \spa\{\xi_{\alpha,\beta,a,b}^{(l)}\}_{\alpha,\beta=1,\ldots,n_0, a=0,1,\ldots, k_R,b=0,1,\ldots,k_L}.
\end{align}


Set $\hee:=\zeij{11}\otimes \unit_{k_R+k_L+1}$ and 
$
a_t:=\hee+\frac 1t\overline{\hee},
$
for each $t\in(0,1]$.
For each $l\in \nan$,  and $t\in(0,1]$, we define maps
$\varpi_l,\Xi_l,\Theta_{l,t}:\mnz\otimes \pu\Mat_{k_R+k_L+1}\pd\to \mnz\otimes \pu\Mat_{k_R+k_L+1}\pd$,
by
\begin{align}
&\varpi_l\lmk X\rmk:=\zeij{11}\otimes \Tr_{n_0}\lmk\lmk \overline{ \zeij{11}}{R^*}^l\otimes \unit_{k_R+k_L+1}\rmk X\rmk,\quad \Xi_l\lmk X\rmk:= \overline{\hee} X-\varpi_l(X),
&\Theta_{l,t}\lmk X\rmk:={ \hee}X+\frac{1}t \Xi_l (X),
\end{align}
for $X\in \mnz\otimes \pu\Mat_{k_R+k_L+1}\pd$.
Here, $\Tr_{n_0}:\mnz\otimes\mkk\to\mkk$ indicates the partial trace.
Note that $\Theta_{l,t}$ is a bijection from $ \mnz\otimes \pu\Mat_{k_R+k_L+1}\pd$ onto itself.
We have
\begin{align}
\lim_{l\to\infty}\Theta_{l,t}\lmk X\rmk=a_t X,\quad X\in \mnz\otimes \pu\Mat_{k_R+k_L+1}\pd,\quad t\in(0,1].
\end{align}

Now we consider the $t\downarrow 0$ limit.
\begin{lem}\label{lem:flim}
For each $l\in\nan$, $\alpha,\beta\in\{1,\ldots,n_0\}$,
$a=0,1,\ldots, k_R$, and $b=0,1,\ldots,k_L$, we have
the limit 
\begin{align}
\lim_{t\downarrow 0} \gbtr{l}\lmk
\Theta_{l,t}\lmk \zeij{\alpha\beta}\otimes \eij{-a,b}
\rmk\rmk
=
\xi_{\alpha,\beta,a,b}^{(l)}.
\end{align}
\end{lem}
\begin{proof}
First, we have
$\lim_{t\downarrow 0} \gbtr{l}\lmk
\zeij{1\beta}\otimes \eij{-a,b}
\rmk
=\Gamma_{l,(B_1(0),B_2(0),0,\cdots,0)}^{(R)}
\lmk
\zeij{1\beta}\otimes \eij{-a,b}\rmk$.
Let us consider the monomials $B_{\mu^{(l)}}(0)$, $\mu^{(l)}\in \{1,\ldots,n\}^{\times l}$
 which appears in the representation of $\Gamma_{l,\bbB(0)}^{(R)}$.
Note that $B_1(0)=R\otimes \Lambda_{\lal_R}$ is diagonal
and $B_2(0)=K_1+K_2$.
Therefore, monomials of $B_1(0)$ and $B_2(0)$
are sum of monomials of
 $B_1(0)$, $K_1$, and $K_2$.
A monomial of $B_1(0)$, $K_1$, and $K_2$
vanishes if there are more than one $K_1$.
(Note $K_1^{(L)}R^{l-1}K_1^{(L)}=0$, $l\in\nan$.)
It also vanishes if there are more than two
$K_2$.(Note $V\Lambda_{\lal_R}^{l_1}V\Lambda_{\lal_R}^{l_2}V=0=0$, $l_1,l_2\in\nan$.)
Furthermore, if $\beta\neq 1$, in order to give a nontrivial contribution, the monomial
has to own one $K_1$.
On the other hand, if $\beta=1$, the number of $K_1$ should be zero.
From this kind of  observations we obtain
\begin{align*}
&\lim_{t\downarrow 0} \gbtr{l}\lmk
\zeij{1\beta}\otimes \eij{-a,b}
\rmk\\
&=\left\{
\begin{gathered}
\lmk \Tr \eij{-a,b} \lmk \Lambda_{\lal_R}^*\rmk^l\rmk
\psi_1^{\otimes l}
+
\lmk
\zeta_2^{(1)(l)}+\zeta_{2,2}^{(2)(l)}
\rmk
\lmk\zeij{11}\otimes \eij{-a,b}\rmk,\quad
\text{if}\;\; 
\beta=1,\\
\lmk
\zeta_1^{(1)(l)}+\zeta_{21}^{(2)(l)}+\zeta_{12}^{(2)(l)}+
\zeta_{221}^{(3)(l)}+\zeta_{212}^{(3)(l)}+\zeta_{122}^{(3)(l)}
\rmk
\lmk\zeij{1\beta}\otimes \eij{-a,b}\rmk,\quad 
\text{if }\beta\in\{2,\ldots,n_0\}.
\end{gathered}
\right.
\end{align*}
Similar kind of observation about $\bbA$
shows that the right hand side
coincides with $\Gamma_{l,\bbA}^{(R)}\lmk\eijz{0,\beta-1} 
\otimes \eij{-a,b}\rmk$.

Next let us consider the case
with $\alpha\neq 1$ and $\alpha\neq \beta$.
Note that in the representation of $\frac 1t \gbtr{l}\lmk
\zeij{\alpha\beta}\otimes \eij{-a,b}
\rmk$,
terms with more than two 
$r_{T_{\tilde \oo(t)}}^{-\frac 12}tK_3$ vanish
in the $t\downarrow 0$ limit.
On the other hand, all the
monomials of $B_1(0)$, $K_1$, and $K_2$ 
are either of the form ''$\text{diagonal matrix}\otimes \mkk$" or in $(\zeij{11}\mnz)\otimes \mkk$.
Therefore,  we have 
\[
\Gamma_{l,(B_1(0),B_2(0),0,\cdots)}^{(R)}\lmk
\zeij{\alpha,\beta}\otimes \eij{-a,b}
\rmk=0,
\]
because $\alpha\neq\beta$ and $\alpha\neq 1$.
Hence, what is left is only the monomials of
$B_1(0)$, $K_1$, $K_2$ and $K_3$ with 
exactly one $K_3$.
The contribution of such a monomial is zero if there exists
$K_1$ left to $K_3$.
This is because such a term belongs to 
$(\zeij{11}\mnz)\otimes\mkk$. 
Terms with two $K_1$ without having $K_3$ between vanish
by the same argument as $\alpha=1$ case.
Also, terms with more than two $K_2$ also disappear.
From these observation we obtain
for $\alpha\neq \beta$,
\begin{align*}
&\lim_{t\downarrow 0} 
\frac 1t \gbtr{l}\lmk
\zeij{\alpha\beta}\otimes \eij{-a,b}
\rmk
\\
&=
\lmk
\zeta_3^{(1)(l)}+\zeta_{23}^{(2)(l)}+\zeta_{32}^{(2)(l)}
+\zeta_{223}^{(3)(l)}+\zeta_{232}^{(3)(l)}+\zeta_{322}^{(3)(l)}+
\zeta_{31}^{(2)(l)}+\zeta_{231}^{(3)(l)}+\zeta_{321}^{(3)(l)}+\zeta_{312}^{(3)(l)}
\rmk
\lmk\zeij{\alpha\beta}\otimes \eij{-a,b}\rmk\\
&+
\lmk
\zeta_{2231}^{(4)(l)}+\zeta_{2321}^{(4)(l)}+\zeta_{2312}^{(4)(l)}+
\zeta_{3221}^{(4)(l)}+\zeta_{3212}^{(4)(l)}+\zeta_{3122}^{(4)(l)}
\rmk\lmk\zeij{\alpha\beta}\otimes \eij{-a,b}\rmk.
\end{align*}
Similar observation about $\bbA$
shows that the right hand side
 coincides with 
$\Gamma_{l,\bbA}^{(R)}\lmk\eijz{-(\alpha-1),\beta-1} 
\otimes \eij{-a,b}\rmk$.

Finally, we consider the case
with $\alpha\neq 1$ and $\alpha=\beta$.
Note that 
terms in $\frac 1t \gbtr{l}\lmk
\lmk \zeij{\alpha\alpha}-\bar r_\alpha^l\zeij{11}\rmk\otimes \eij{-a,b}
\rmk$
with more than two 
$r_{T_{\tilde \oo(t)}}^{-\frac 12}tK_3$ vanish
in the $t\downarrow 0$ limit.
Now let us see that the terms given by
monomials of $B_1(0)$, $K_1$, and $K_2$ 
vanish.
There can be at most one $K_1$ for
nonzero monomial of $B_1(0)$, $K_1$, and $K_2$ as above.
However, if there is one, then
the monomial belongs to 
$\lmk\zeij{11}\mnz\overline{\zeij{11}}\rmk\otimes \mkk$.
Therefore, these terms are orthogonal to
$\lmk \zeij{\alpha,\alpha}-\bar r_{\alpha}^l\zeij{11}\rmk\otimes \mkk$
with respect to the inner product 
$\braket{}{}_{\Tr}$ given by $\braket{X}{Y}_{\Tr}=\Tr X^*Y$.
Hence we conclude that the number of $K_1$ should be zero.
In case there is no $K_1$, the monomial belongs to
$R^l\otimes \mkk$.
But such an element is again orthogonal to 
$\lmk \zeij{\alpha,\alpha}-\bar r_{\alpha}^l\zeij{11}\rmk\otimes \mkk$
with respect to $\braket{}{}_{\Tr}$.
Hence any term given by a monomial of $B_1(0)$, $K_1$, and $K_2$ vanishes.

For the reasons above,
terms which are left in $t\downarrow 0$ limit
are only the monomials of
$B_1(0)$, $K_1$, $K_2$ and $K_3$ with 
exactly one $K_3$.
Terms with two $K_1$ without having $K_3$ in between vanish by observation in $\alpha\neq 1$, $\alpha\neq \beta$ case.
Therefore, there can be
at most two $K_1$.
Under these conditions, if furthermore
there is no  $K_1$ in the monomial,
it belongs to $\overline{\zeij{11}}\mnz\zeij{11}\otimes \mkk$.
If there are two  $K_1$ in the monomial, then from the above observation, they have to be ordered as $K_1,K_3,K_1$, and
the monomial belongs to $\zeij{11}\mnz
\overline{\zeij{11}}\otimes \mkk$.
In both cases, the monomial is orthogonal to 
$\lmk \zeij{\alpha,\alpha}-\bar r_{\alpha}^l\zeij{11}\rmk\otimes\mkk$
with respect to $\braket{}{}_{\Tr}$.
Therefore, to have a non-zero contribution, the number of $K_1$ has to be one.
When there is one $K_1$, the monomial belongs to
$\lmk
\zeij{11}\mnz{\zeij{11}}\rmk
\otimes \mkk$
if $K_1$ is left to $K_3$,
and it  belongs to
$\lmk
\overline{\zeij{11}}\mnz\overline{\zeij{11}}\rmk
\otimes \mkk$
if $K_3$ is left to $K_1$.Terms with more than two $K_2$s disappear as usual.
From these observations, we obtain 
for $\alpha\neq \beta$,
\begin{align*}
&\lim_{t\downarrow 0} \frac 1t \gbtr{l}\lmk
\lmk \zeij{\alpha\alpha}-\bar r_\alpha^l\zeij{11}\rmk\otimes \eij{-a,b}
\rmk
\\
&=
\lmk
\zeta_{31}^{(2)(l)}
+\zeta_{231}^{(3)(l)}+\zeta_{321}^{(3)(l)}+\zeta_{312}^{(3)(l)}
\rmk\lmk\zeij{\alpha\alpha}\otimes \eij{-a,b}\rmk\\
&+\lmk
\zeta_{2231}^{(4)(l)}+\zeta_{2321}^{(4)(l)}+\zeta_{2312}^{(4)(l)}+
\zeta_{3221}^{(4)(l)}+\zeta_{3212}^{(4)(l)}+\zeta_{3122}^{(4)(l)}
\rmk
\lmk\zeij{\alpha\alpha}\otimes \eij{-a,b}\rmk\\
&-\bar r_\alpha^l\lmk
\zeta_{13}^{(2)(l)}
+\zeta_{213}^{(3)(l)}+\zeta_{123}^{(3)(l)}+\zeta_{132}^{(3)(l)}
\rmk\lmk\zeij{11}\otimes \eij{-a,b}\rmk\\
&-\bar r_\alpha^l
\lmk
\zeta_{2213}^{(4)(l)}+\zeta_{2123}^{(4)(l)}+\zeta_{2132}^{(4)(l)}+
\zeta_{1223}^{(4)(l)}+\zeta_{1232}^{(4)(l)}+\zeta_{1322}^{(4)(l)}
\rmk
\lmk\zeij{11}\otimes \eij{-a,b}\rmk.
\end{align*}
It is straight forward to check that
the right hand side coincides with 
$\Gamma_{l,\bbA}^{(R)}\lmk\eijz{-(\alpha-1),\alpha-1} 
\otimes \eij{-a,b}\rmk
-\delta_{\alpha,\alpha}(1-\delta_{\alpha,1})
\bar r_{\alpha}^l
\Upsilon_l(\eij{-a,b})$.
\end{proof}

\subsection{The properties of $\bbA$}
Throughout this subsection, let $2\le n_0\in\nan$, $k_R,k_L\in\nan\cup\{0\}$, and $\bbA$ be given by (\ref{eq:defa}). 
By a routine argument, using Lemma C.7 of \cite{Ogata1}, we obtain the following Lemma.
\begin{lem}\label{lem:asp}
Set
\[
\caD_{\bbA}:=
\spa\lmk \left\{
\unit, \eijz{-\alpha,0}, \eijz{0,\beta}, \eij{-\alpha,\beta}
\right\}_{\alpha,\beta=1,\ldots, n_0-1}
\otimes \left\{
\unit,\eij{-a,0}, \eij{0,b},\eij{-a,b}
\right\}_{a=1,\ldots,k_R,b=1,\ldots,k_L}\rmk.
\]
There exists an $l_\bbA\in\nan$ such that
\begin{align*}
\caK_l(\bbA)=\caD_{\bbA}\lmk
\Lambda_{\lal_L}^l\otimes \Lambda_{\lal_R}^l\rmk,
\end{align*}
for all $l\ge l_\bbA$.
\end{lem}

Using this, we obtain the following Lemma. The proof is the same as that of Proposition 3.1 \cite{Ogata1}.
\begin{lem}\label{lem:aprop}
The pentad $(n, (2n_0-1)(k_L+k_R+1), \puuz, \pddz,\bbA)$ satisfies {\it Condition 2}, defined in Definition 2.3 \cite{Ogata1}.
Therefore, by Lemma 2.9 \cite{Ogata1}, 
there exist a number $0<s_\bbA<1$, a state $\varphi_{\bbA}$ on $\Mat_{2n_0-1}\otimes\mkk$
and a positive element $e_{\bbA}\in\lmk \Mat_{2n_0-1}\otimes\mkk\rmk_+$,
such that $T_\bbA$ satisfies the Spectral Property II with respect to 
$(s_\bbA, e_\bbA,\varphi_\bbA)$ (see Definition 2.7 \cite{Ogata1}) and 
$s(e_{\bbA})=\puuz$, and $s(\varphi_{\bbA})=\pddz$.
Furthermore, the pentad $(n, (2n_0-1)(k_L+k_R+1), \puuz, \pddz,\bbA)$ satisfies {\it Condition 3} for $l_\bbA$, and 
the triple  $(n, (2n_0-1)(k_L+k_R+1), \bbA)$ satisfies {\it Condition 4}, for $(l_\bbA,l_\bbA)$.
Therefore, by Lemma 2.14, Proposition 2.6 of \cite{Ogata1},  $m_\bbA\le 2l_\bbA<\infty$ and
$M_{\bbA,\puuz,\pddz}<\infty$.  (Here $m_\bbA$, $M_{\bbA,\puuz,\pddz}$ defined in Definition 1.4 and (8) of \cite{Ogata1}.) 
From  Proposition 2.6 \cite{Ogata1}, $\caG_{l,\bbA}$ satisfies {\it Condition 1}(see Definition2.1 of \cite{Ogata1}).
\end{lem}
By the argument parallel to  Lemma 3.11-13 and Lemma 3.15, 3.16 of \cite{Ogata1}, we obtain the following.
\begin{lem}
There exists a completely positive map from the half-infinite chain ${\mathcal A}_{(-\infty,-1]}$
 to $\Mat_{2n_0-1}\otimes \mkk$, such that 
\begin{align*}
\bbL_\bbA(B):=\sum_{\mu^{(l)},\nu^{(l)}\in \{1,\ldots,n\}^{\times l}}
\braket{\ws{l}}{B\wsn{l}}
\lmk \wan{l}\rmk^{*}\rho_{\bbA} \wa{l},
\end{align*} 
if $B\in {\mathcal A}_{[-l,-1]}\simeq \otimes_{i=0}^{l-1}\mn$ for $l\in\nan$.
We have
\begin{align}\label{eq:lbran}
\Ran \bbL_\bbA=\pddz \lmk\Mat_{2n_0-1}\otimes \mkk\rmk\pddz.
\end{align}
For each $\sigma_L\in {\mathfrak E}_{n_0(k_L+1)}$,
under the identification $\Mat_{n_0(k_L+1)}\simeq \pddz\lmk \Mat_{2n_0-1}\otimes \mkk\rmk\pddz$,
define
$\Xi_{L,\bbA}(\sigma_L):\caA_{(-\infty,-1]}\to \cc$
by
\[
\Xi_{L,\bbA}(\sigma_L)(B)
:=\sigma_{L}(y_{\bbA}^{\frac 12} \bbL_\bbA(B)y_{\bbA}^{\frac 12}),\quad
B\in\caA_{(-\infty,-1]}.
\]
Then for $m\ge {\bbm}_{\bbA}$, 
the map
$\Xi_{L,\bbA}:{\mathfrak E}_{n_0(k_L+1)}\to 
\caS_{(-\infty,-1]}(H_{\Phi_{m,\bbA}})$ defines an
affine bijection.
An analogous statement holds for the right half infinite chain and defines
$\Xi_{R,\bbA}$, $\bbR_\bbA$.
The set $\caS_{\bbZ}(H_{\Phi_{m,\bbA}})$ consists of a unique state 
\begin{align}\label{eq:ainf}
\omega_{\bbA,\infty}=\bigotimes_{\bbZ}\rho_0,
\end{align}
where $\rho_0$ is a state on $\mn$ such that
$\rho_0=\braket{\chi^{(n)}_1}{\lmk\cdot\rmk \chi_1^{(n)}}$. 
\end{lem}
The statement of Lemma 3.25 of \cite{Ogata1} also holds for $\bbA$.
\begin{lem}\label{lem:mme}
Let $m\ge \bbm_{\bbA}$. There exists a constant $C_\bbA^{'''}>0$ satisfying the following. : Let $M\in\nan$ 
and $\varphi$ be a state on $\caA_{\bbZ}$. Assume that we have $\varphi(\tau_{i}(1-G_{m,\bbA}))=0$ for all $i\in\bbZ$ with
$[i,i+m-1]\subset [-M,M]^c$.
Then for any $L\in\nan$ with $M+1\le L$ and $A\in\caA_{[-L+1,L-1]^c}$, we have
\begin{align}
\lv \varphi \lmk A\rmk-\lmk \left. \omega_{\bbA,\infty}\rv_{\caA_{(-\infty,-1]}}\otimes \left. \omega_{\bbA,\infty}\rv_{\caA_{[0,\infty)}}\rmk\lmk A\rmk\rv
\le C_\bbA^{'''} s_\bbA^{L-M}\lV A\rV.
\end{align}
\end{lem}

Furthermore, as in Lemma 3.15-22 of \cite{Ogata1}, we obtain the followings.
\begin{lem}\label{lem:a1a5a}
 [A1]-[A5] in \cite{Ogata2} hold for $H_{\Phi_{m, \bbA}}$ for $m\ge 2l_\bbA$.
\end{lem}
\begin{proof}
The proof is a straight forward application of 
the argument of Part I \cite{Ogata1}.
\end{proof}
\begin{notation}\label{nota:bbk}
From Theorem 1.2 of \cite{Ogata2}, for each $m\ge 2l_\bbA$, there exists a $\bbV\in\ClassA$ with respect to
some $(n_{0,\bbV},k_{R,\bbV},k_{L,\bbV},\lal_\bbV,\bbD_{\bbV},\bbG_\bbV,Y_\bbV)$ 
 such that 
\begin{align}\label{eq:aakk}
\caS_{(-\infty,-1]}(H_{\Phi_{m, \bbA}})=\caS_{(-\infty,-1]}(H_{\Phi_{m_1,\bbV}}),\quad
\caS_{[0,\infty)}(H_{\Phi_{m, \bbA}})=\caS_{[0,\infty)}(H_{\Phi_{m_1,\bbV}}),\quad
\caS_{\bbZ}(H_{\Phi_{m, \bbA}})=\caS_{\bbZ}(H_{\Phi_{m_1,\bbV}})=\{\omega_{\bbV,\infty}\}=\{\omega_{\bbA,\infty}\},
\end{align}
for all $m_1\ge 2l_\bbV$.
Furthermore, $H_{\Phi_{m, \bbA}}$ and $H_{\Phi_{m_1,\bbV}}$
are type II-$C^1$-equivalent, by Corollary 1.4 of \cite{Ogata2}.
\end{notation}
\begin{lem}
This $\bbV$ in Notation \ref{nota:bbk} belongs to $\Class(n,1,n_0(k_R+1)-1, n_0(k_L+1)-1)$.
\end{lem}
\begin{proof}
We already know that $\bbV\in\Class A$.
Therefore, there exist
$n_{0,\bbV}\in\nan$, $k_{R\bbV},k_{L\bbV}\in\nan\cup\{0\}$, $\;(\lal_\bbV,\bbD_\bbV,\bbG_\bbV,Y_\bbV)\in\caT(k_{R\bbV},k_{L\bbV})$
such that 
$\bbV\in {\mathfrak B}(n,n_{0,\bbV},k_{R\bbV},k_{L\bbV},\lal_\bbV,\bbD_\bbV,\bbG_\bbV,Y_\bbV)$.
We have to show that $n_{0,\bbV}=1$, $k_{R\bbV}=n_0(k_R+1)-1$, $k_{L\bbV}=n_0(k_L+1)-1$, $Y_\bbV=0$, and $\lal_\bbV\in\Wo'(k_{R\bbV},k_{L\bbV})$.

First we show $n_{0,\bbV}=1$.
Recall from Part I \cite{Ogata1} Lemma 3.2, we obtain $\oo_{\bbV}\in\Prim(n,n_{0,\bbV})$ by
\[
\omega_{\mu,\bbV}\otimes \eij{00}=\lmk\unit\otimes \eij{00}\rmk
V_{\mu}\lmk\unit\otimes \eij{00}\rmk,\quad \mu=1,\ldots,n.
\]
By the primitivity, from C.6 of \cite{Ogata1}, there exist a faithful positive linear functional $\varphi_{\oo_\bbV}$ on $\Mat_{n_{0,\bbV}}$ and 
a strictly positive element $e_{\oo_\bbV}\in \Mat_{n_{0,\bbV}}$ such that 
$\lim_{N\to\infty} T_{\oo_\bbV}^N(A)=\varphi_{\oo_\bbV}(A)e_{\oo_\bbV}$ for all $A\in \Mat_{n_{0,\bbV}}$.
In particular, we have $\varphi_{\oo_\bbV}(e_{\oo_\bbV})=1$.

On the other hand, by Lemma 2.9 \cite{Ogata1}, we have
$\lim_{N\to\infty} T_{\bbV}^N(A)=\varphi_{\bbV}(A)e_\bbV$ for all $A\in \Mat_{n_{0,\bbV}}\otimes \Mat_{k_{R\bbV}+k_{L\bbV}+1}$.
Sandwiching this by $\widehat{E_{00}^{(k_{R\bbV},k_{L\bbV})}} =\unit\otimes E_{00}^{(k_{R\bbV},k_{L\bbV})}$,
we obtain
\begin{align}\label{eq:ov}
&\varphi_\bbV\lmk X\otimes  E_{00}^{(k_{R\bbV},k_{L\bbV})}\rmk\cdot
\widehat{E_{00}^{(k_{R\bbV},k_{L\bbV})}} 
e_\bbV
\widehat{E_{00}^{(k_{R\bbV},k_{L\bbV})}} \nonumber\\
&=
\lim_{N\to\infty} 
\widehat{E_{00}^{(k_{R\bbV},k_{L\bbV})}} T_{\bbV}^N\lmk X\otimes  E_{00}^{(k_{R\bbV},k_{L\bbV})}\rmk\widehat{E_{00}^{(k_{R\bbV},k_{L\bbV})}} 
=
\lim_{N\to\infty} T_{\oo_\bbV}^N\lmk X\rmk \otimes E_{00}^{(k_{R\bbV},k_{L\bbV})}\nonumber\\
&=\varphi_{\oo_\bbV}\lmk X\rmk e_{\oo_\bbV}\otimes E_{00}^{(k_{R\bbV},k_{L\bbV})},
\end{align}
for any $X\in\Mat_{n_{0,\bbV}}$.
In the second equality, we used the property $\widehat{P_L^{(n_{0,\bbV}k_{R\bbV},k_{L\bbV})}}\widehat{V_{\mu^{(l)}}} \widehat{P_R^{(n_{0,\bbV},k_{R\bbV},k_{L\bbV})}} 
=\widehat{\omega_{\mu^{(l)},\bbV}}\otimes E_{00}^{(n_{0},k_{R\bbV},k_{L\bbV})}$, for $\mu^{(l)}\in\{1,\ldots,n\}^{\times l}$.
Substituting $X=e_{\oo_\bbV}$ to this equality, we obtain
\begin{align}\label{eq:ove}
 e_{\oo_\bbV}\otimes E_{00}^{(k_{R\bbV},k_{L\bbV})}
=\varphi_\bbV\lmk e_{\oo_\bbV} \otimes  E_{00}^{(k_{R\bbV},k_{L\bbV})}\rmk\cdot
\widehat{E_{00}^{(n_{0,\bbV},k_{R\bbV},k_{L\bbV})}} 
e_\bbV
\widehat{E_{00}^{(n_{0,\bbV},k_{R\bbV},k_{L\bbV})}}.
\end{align}
We used $\varphi_{\oo_\bbV}(e_{\oo_\bbV})=1$ on the left hand side.
Substituting (\ref{eq:ove}) to (\ref{eq:ov}), we obtain 
\begin{align}\label{eq:ovp}
\varphi_{\oo_\bbV}\lmk X\rmk
=\frac{\varphi_\bbV\lmk X\otimes  E_{00}^{(k_{R\bbV},k_{L\bbV})}\rmk}{\varphi_\bbV\lmk e_{\oo_\bbV}\otimes  E_{00}^{(k_{R\bbV},k_{L\bbV})}\rmk},\quad X\in\Mat_{n_{0,\bbV}}.
\end{align}
We set $\uu\in \Mat_{n_{0,\bbV}}^{\times n}$ by $u_\mu:=e_{\oo_\bbV}^{-\frac 12}\lmk \omega_{\mu,\bbV}\rmk  e_{\oo_\bbV}^{\frac 12}$, $\mu=1,\ldots,n$.
Then $T_\uu$ is a unital CP map on $\Mat_{n_{0,\bbV}}$ and 
 $\varphi_\uu:=\varphi_{\oo_\bbV}\circ \Ad e_{\oo_\bbV}^{\frac 12}$ is a faithful $T_\uu$-invariant state on $\Mat_{n_{0,\bbV}}$.
From Lemma 3.16 of \cite{Ogata1} and (\ref{eq:ove}), (\ref{eq:ovp}), the unique element $\omega_{\bbV,\infty}$ in 
$\caS_{\bbZ}(H_{\Phi_{m, \bbA}})=\caS_{\bbZ}(H_{\Phi_{m_1,\bbV}})$ is
\begin{align*}
&\omega_{\bbV,\infty}\lmk
A
\rmk=
\sum_{\mu^{(l)},\nu^{(l)}\in \{1,\ldots,n\}^{\times l}}\braket{\ws{l}}{A\wsn{l}}
\varphi_{\bbV}\lmk \widehat{V_{\mu^{(l)}}}
e_{\bbV}
\lmk \widehat{V_{\nu^{(l)}}}\rmk^*\rmk\nonumber\\
&=
\sum_{\mu^{(l)},\nu^{(l)}\in \{1,\ldots,n\}^{\times l}}\braket{\ws{l}}{A\wsn{l}}
\varphi_{\bbV}\lmk \widehat{P_L^{(n_{0,\bbV},k_{R\bbV},k_{L\bbV})}} \widehat{V_{\mu^{(l)}}}\widehat{P_R^{(n_{0,\bbV},k_{R\bbV},k_{L\bbV})}} 
e_{\bbV}\widehat{P_R^{(n_{0,\bbV},k_{R\bbV},k_{L\bbV})}} 
\lmk \widehat{V_{\nu^{(l)}}}\rmk^*
\widehat{P_L^{(n_{0,\bbV},k_{R\bbV},k_{L\bbV})}} \rmk\\
&=
\sum_{\mu^{(l)},\nu^{(l)}\in \{1,\ldots,n\}^{\times l}}\braket{\ws{l}}{A\wsn{l}}
\varphi_{\bbV}\lmk \lmk\widehat{\omega_{\mu^{(l)},\bbV}}\otimes  E_{00}^{(k_{R\bbV},k_{L\bbV})}\rmk
e_{\bbV}
\lmk\widehat{\omega_{\nu^{(l)},\bbV}}\otimes  E_{00}^{(k_{R\bbV},k_{L\bbV})}\rmk^*\rmk\nonumber\\
&=
\sum_{\mu^{(l)},\nu^{(l)}\in \{1,\ldots,n\}^{\times l}}\braket{\ws{l}}{A\wsn{l}}
\varphi_{\oo_\bbV}\lmk \widehat{\omega_{\mu^{(l)},\bbV}}
e_{\oo_\bbV}
\lmk\widehat{\omega_{\nu^{(l)},\bbV}}\rmk^*\rmk\nonumber\\
&=
\sum_{\mu^{(l)},\nu^{(l)}\in \{1,\ldots,n\}^{\times l}}\braket{\ws{l}}{A\wsn{l}}
\varphi_{\uu}\lmk \widehat{u_{\mu^{(l)}}}
\lmk\widehat{u_{\nu^{(l)}}}\rmk^*\rmk,\nonumber\\
&A\in\caA_{[i,i+l-1]},\;i\in \bbZ,\;l\in \nan .
\end{align*}
In the second equality we used the fact that the supports of $e_\bbV$, $\varphi_\bbV$ are
$\widehat{P_R^{(n_{0,\bbV},k_{R\bbV},k_{L\bbV})}} $, $\widehat{P_L^{(n_{0,\bbV},k_{R\bbV},k_{L\bbV})}}$.
In the third equality, we used the property $\widehat{P_L^{(n_{0,\bbV},k_{R\bbV},k_{L\bbV})}}\widehat{V_{\mu^{(l)}}} \widehat{P_R^{(n_{0,\bbV},k_{R\bbV},k_{L\bbV})}} 
=\widehat{\omega_{\mu^{(l)},\bbV}}\otimes E_{00}^{(k_{R\bbV},k_{L\bbV})}$.
The eigenspace of $T_\uu$ for the eigenvalue $1$ is $\bbC\unit$ and $(\Mat_{n_{0,\bbV}}, \uu,\varphi_\uu)$ is a minimal standard triple which right generates $\omega_{\bbV,\infty}$.
(Recall Appendix C of \cite{Ogata2}.)
By (\ref{eq:ainf}), $\omega_{\bbV,\infty}=\omega_{\bbA,\infty}$ is right-generated by 
a minimal standard triple $(\Mat_1=\bbC,\uup,\id_{\bbC})$, where
$\uup=(\upsilon_{\mu})\in\bbC^{\times n}$ is given by 
$\upsilon_{1}=1$ and $\upsilon_{\mu}=0$ for $\mu\ge 2$.
From this and the uniqueness of the minimal representation proven in \cite{fnwpure} (see Theorem C.3 of \cite{Ogata2}), we conclude that there exists a $C^*$-isomorphism $\Theta: \Mat_{n_{0,\bbV}}\to \bbC$
such that $\upsilon_{\mu}\upsilon_{\mu}^*=\Theta(u_{\mu}u_{\mu}^*)$ for all $\mu=1,\ldots, n$.
Therefore, $n_{0,\bbV}=1$ and $\Mat_{n_{0,\bbV}}=\bbC$.

Second, note that from the first and second equality of  (\ref{eq:aakk}) and Lemma 3.13 of \cite{Ogata1},
the map
$\Xi_{L,\bbA}^{-1}\circ \Xi_{L\bbV}$ (resp. $\Xi_{R,\bbA}^{-1}\circ \Xi_{R\bbV}$) is a well-defined affine bijection from
$\caE_{k_{L\bbV}+1}$ (resp. $\caE_{k_{R\bbV}+1}$) to $\caE_{n_0(k_L+1)}$ (resp. $\caE_{n_0(k_R+1)}$). Recall that
 $\caE_k$ denotes the set of all states on $\mk$.
In particular, we have $k_{R\bbV}+1=n_0(k_R+1)$ and $k_{L\bbV}+1=n_0(k_L+1)$.

Third, we show $Y_\bbV=0$. 
We note that from the choice of $\lal_L$, $\lal_R$,
the numbers 
\begin{align}\label{eq:llll}
\{\overline{\lambda_{L\beta}}\lambda_{L\beta'}
\overline{\lambda_{Rb}}\lambda_{Rb'}\mid 
\beta,\beta'=0,\ldots, n_0-1,\; b,b'=0,\ldots,k_L
\}
\end{align}
are distinct. Furthermore, $\overline{\lambda_{L\beta}}\lambda_{L\beta'}
\overline{\lambda_{Rb}}\lambda_{Rb'}$ is a positive real number if and only if
$\beta=\beta'$ and $b=b'$.

From $\Mat_{n_{0,\bbV}}=\bbC$, we have $\omega_{1,\bbV}=u_{1}$ and it is an element in $\bbC$ with absolute value $1$.
For $\mu\ge 2$, we have $\omega_{\mu,\bbV}=0$.
Note, in particular, that
$V_1$ is of the form
\begin{align}\label{eq:k1}
V_1=\omega_1\Lambda_{\lal_\bbV}\lmk \unit+Y_\bbV\rmk +W_1,
\end{align}
with some element $W_1$ , such that 
\begin{align}\label{eq:W}
P_L^{(k_{R\bbV}, k_{L\bbV})}W_1=\sum_{\stackrel{i,j=0,\ldots k_{L\bbV}}{\lv\lambda_{i,\bbV}\rv>\lv \lambda_{j,\bbV}\rv}}
E_{ii}^{(k_{R\bbV}, k_{L\bbV})}W_1E_{jj}^{(k_{R\bbV}, k_{L\bbV})},
\end{align}
(Recall that $\bbV\in \Mat_{n_{0\bbV}}\otimes \caD\lmk k_{R\bbV},k_{L\bbV},\bbD_\bbV,\bbG_\bbV\rmk\Lambda_{\lal_\bbV}\lmk 1+Y_\bbV\rmk$ from Definition 1.13 of \cite{Ogata1}.)
From (\ref{eq:k1}) and (\ref{eq:W}),
we have 
\begin{align}\label{eq:k1ll}
\lmk\sum_{i: i\ge 0, \lambda_{i\bbV}=\lambda}E_{ii}^{(k_{R\bbV},k_{L\bbV})}\rmk
 V_1^N\lmk\sum_{i: i\ge 0, \lambda_{i\bbV}=\lambda}E_{ii}^{(k_{R\bbV},k_{L\bbV})}\rmk
=
\lmk \omega_1\lambda\lmk \unit+Y_\bbV\rmk\rmk^N
\lmk\sum_{i: i\ge 0,\lambda_{i\bbV}=
\lambda}E_{ii}^{(k_{R\bbV},k_{L\bbV})}\rmk,
\end{align}
for any $N\in\nan$ and $\lambda\in\{\lambda_{i\bbV}\}_{i=1}^{k_{L\bbV}}$.

For each nonzero $\xi\in \bbC^{k_{L\bbV}+1}$, we define a state $\sigma_\xi$ on  $\Mat_{k_{L\bbV}+1}$
given by a density matrix 
\[
C_{\xi}^{-1}{y_{\bbV}^{-\frac 12}\ket{\xi}\bra{\xi}y_{\bbV}^{-\frac 12}}
\]
where $C_\xi={\braket{\xi}{y_{\bbV}^{-1}\xi}}$.
(Recall the definition of  $y_\bbB$ in \cite{Ogata1}, after Remark 2.10.)
Furthermore, we denote by $\tilde \sigma_\xi$ the image of $\sigma_\xi$ by $\Xi_{L,\bbA}^{-1}\circ \Xi_{L\bbV}$ ,
i.e., $\tilde \sigma_\xi=\Xi_{L,\bbA}^{-1}\circ \Xi_{L\bbV}\lmk \sigma_\xi\rmk$.
Let $\tilde D_\xi$ be the density matrix of $\tilde \sigma_\xi$.
We set $D_\xi:=y_{\bbA}^{\frac 12} \tilde D_\xi y_{\bbA}^{\frac 12}\in \pddz\lmk \Mat_{2n_0-1}\otimes \Mat_{K_R+k_L+1}
\rmk \pddz$. (Recall (\ref{eq:pa}) for the definition of $\pddz$.)
For each 
$X\in \Mat_{k_{L\bbV}+1}$, 
there exists an $A_X\in \caA_{(-\infty,-1]}^{\rm loc}$ such that $\bbL_\bbV(A_X)=X$.
This can be seen from the proof of Lemma 3.11 of \cite{Ogata1}.
From the definitions, we have
\begin{align}\label{eq:genak}
& \sigma_\xi\lmk y_\bbV^{\frac 12}\lmk \Ad V_1^*\rmk^N \lmk \bbL_\bbV\lmk A_X\rmk\rmk y_\bbV^{\frac 12}\rmk
=\Xi_{L\bbV}\lmk \sigma_\xi\rmk\lmk A_X\otimes \lmk e_{11}^{(n)}\rmk^{\otimes N}\rmk
\nonumber\\
&=\Xi_{L,\bbA}\lmk \tilde \sigma_\xi\rmk\lmk A_X\otimes \lmk e_{11}^{(n)}\rmk^{\otimes N}\rmk= 
\tilde \sigma_\xi\lmk 
 y_\bbA^{\frac 12}\lmk \Ad A_1^*\rmk^N \lmk \bbL_\bbA\lmk A_X\rmk\rmk y_\bbA^{\frac 12}\rmk
\nonumber\\
&=\sum_{\beta,\beta'=0}^{n_0-1}\sum_{b,b'=0}^{k_{L}}
\lmk \overline{\lambda_{L\beta}}\lambda_{L\beta'}
\overline{\lambda_{Rb}}\lambda_{Rb'}\rmk^{N}
\braket{\gaa{\beta}{b}}{ \bbL_\bbA\lmk A_X\rmk \gaa{\beta'}{b'}}
\braket{\gaa{\beta'}{b'}}{D_\xi \gaa{\beta}{b}},
\end{align}
for all $N\in\nan$ and $X\in \Mat_{k_{L\bbV}+1}$. (Recall (\ref{eq:gaa}) for the definition of $\gaa{\beta}{b}$.)

Let
$\lambda\in\{\lambda_{i\bbV }\}_{i=1}^{k_{L\bbV}}$,
 $X\in \lmk
\sum_{i: i\ge 0,\lambda_{i\bbV}=
\lambda}E_{ii}^{(k_{R\bbV},k_{L\bbV})}\rmk
\Mat_{k_{R\bbV}+k_{L\bbV}+1}
\lmk
\sum_{i: i\ge 0, \lambda_{i\bbV}=
\lambda}E_{ii}^{(k_{R\bbV},k_{L\bbV})}\rmk$,
and $0\neq \xi\in \lmk\sum_{i: i\ge 0,\lambda_{i\bbV}=
\lambda}E_{ii}^{(k_{R\bbV},k_{L\bbV})}\rmk\bbC^{k_{R\bbV}+k_{L\bbV}+1}$.
Substituting this to (\ref{eq:genak}) and using (\ref{eq:k1ll}), we get
\begin{align}\label{eq:vac}
&C_\xi^{-1}\sum_{k=0}^{2(k_{L\bbV}+1)}
\lv \lambda\rv^{2N} {}_NC_k \braket{\xi}{\lmk \Ad\lmk 1+Y_{\bbV}^*\rmk -\unit \rmk^{k}\lmk X\rmk \xi}
=
C_\xi^{-1}\lv \lambda\rv^{2N}\braket{\xi}{\lmk \Ad\lmk 1+Y_{\bbV}^*\rmk\rmk^N\lmk X\rmk \xi}\nonumber\\
&=
 \sigma_\xi\lmk y_\bbV^{\frac 12}\lmk \Ad V_1^*\rmk^N \lmk X\rmk y_\bbV^{\frac 12}\rmk
=
 \sigma_\xi\lmk y_\bbV^{\frac 12}\lmk \Ad V_1^*\rmk^N \lmk \bbL_\bbV\lmk A_X\rmk\rmk y_\bbV^{\frac 12}\rmk
\nonumber\\
&=\sum_{\beta,\beta'=0}^{n_0-1}\sum_{b,b'=0}^{k_{L}}
\lmk \overline{\lambda_{L\beta}}\lambda_{L\beta'}
\overline{\lambda_{Rb}}\lambda_{Rb'}\rmk^{N}
\braket{\gaa{\beta}{b}}{ \bbL_\bbA\lmk A_X\rmk \gaa{\beta'}{b'}}
\braket{\gaa{\beta'}{b'}}{D_\xi \gaa{\beta}{b}},
\end{align}
for all $N\in\nan$.
Applying Lemma C.7 of \cite{Ogata1}, we see that
$\braket{\xi}{\lmk \Ad\lmk 1+Y_{\bbV}^*\rmk -\unit \rmk^{k}\lmk X\rmk \xi}=0$
for all $k\in\nan$. As this holds for all $\lambda$, $\xi$ and $X$ as above, we conclude that
$Y_\bbV=0$.

Finally, we show that $\lal_\bbV\in \Wo'(n_0(k_R+1)-1,n_0(k_L+1)-1)$. Namely, we show that the numbers
$\{\lambda_{i,\bbV}\}_{i=1}^{k_{L\bbV}}$ (resp. $\{\lambda_{i,\bbV}\}_{-k_{R\bbV}}^{-1}$) are distinct.
Substituting $Y_\bbV=0$ to (\ref{eq:vac}), we have 
\begin{align}\label{eq:lka}
&C_\xi^{-1}
\lv \lambda\rv^{2N} \braket{\xi}{X\xi}
&=\sum_{\beta,\beta'=0}^{n_0-1}\sum_{b,b'=0}^{k_{L}}
\lmk \overline{\lambda_{L\beta}}\lambda_{L\beta'}
\overline{\lambda_{Rb}}\lambda_{Rb'}\rmk^{N}
\braket{\gaa{\beta}{b}}{ \bbL_\bbA\lmk A_X\rmk \gaa{\beta'}{b'}}
\braket{\gaa{\beta'}{b'}}{D_\xi \gaa{\beta}{b}},
\end{align} for all $N\in\nan$,
$\lambda\in\{\lambda_{i\bbV}\}_{i=1}^{k_{L\bbV}}$,
 $X\in \lmk
\sum_{i: i\ge 0,\lambda_{i\bbV}=
\lambda}E_{ii}^{(k_{R\bbV},k_{L\bbV})}\rmk
\Mat_{k_{R\bbV}+k_{L\bbV}+1}
\lmk
\sum_{i: i\ge 0,\lambda_{i\bbV}=
\lambda}E_{ii}^{(k_{R\bbV},k_{L\bbV})}\rmk$,
and $0\neq \xi\in \lmk\sum_{i: i\ge 0,\lambda_{i\bbV}=
\lambda}E_{ii}^{(k_{R\bbV},k_{L\bbV})}\rmk\bbC^{k_{R\bbV}+k_{L\bbV}+1}$.
In particular, for each $\lambda\in\{\lambda_{i\bbV}\}_{i=1}^{k_{L\bbV}}$, take any unit vector 
$\xi\in \sum_{i: i\ge 0,\lambda_{i\bbV}=\lambda}E_{ii}^{(k_{R\bbV},k_{L\bbV})}\bbC^{k_{R\bbV}+k_{L\bbV}+1}$
and set $X:=\ket{\xi}\bra{\xi}$. Then we obtain
\begin{align*}
C_\xi^{-1}
\lv \lambda\rv^{2N} 
=\sum_{\beta,\beta'=0}^{n_0-1}\sum_{b,b'=0}^{k_{L}}
\lmk \overline{\lambda_{L\beta}}\lambda_{L\beta'}
\overline{\lambda_{Rb}}\lambda_{Rb'}\rmk^{N}
\braket{\gaa{\beta}{b}}{ \bbL_\bbA\lmk A_{\ket{\xi}\bra{\xi}}\rmk \gaa{\beta'}{b'}}
\braket{\gaa{\beta'}{b'}}{D_\xi \gaa{\beta}{b}},
\end{align*}
for all $N\in\nan$.
As the coefficient $C_\xi^{-1}$ of $\lv \lambda\rv^{2N}$ on the left hand side
is nonzero, another application of Lemma C.7 of \cite{Ogata1} and distinction of numbers in (\ref{eq:llll}) implies that 
for each $\lambda\in \{\lambda_i\}_{i=1}^{k_{L\bbV}}$, there exists 
$\overline{\lambda_{L\beta}}\lambda_{L\beta'}
\overline{\lambda_{Rb}}\lambda_{Rb'}$ such that $\lv\lambda\rv^2=\overline{\lambda_{L\beta}}\lambda_{L\beta'}
\overline{\lambda_{Rb}}\lambda_{Rb'}$.
By the consideration around (\ref{eq:llll}), we see that 
$\beta=\beta'$, $b=b'$, $|\lambda|^2=|\lambda_{L\beta}|^2|\lambda_{Rb}|^2$.
Furthermore, as the numbers in (\ref{eq:llll}) are distinct, such a pair $(\beta,b)$ is unique.
We denote the pair by $(\beta_\lambda,b_\lambda)$.
For any
unit vector 
$\xi\in \sum_{i: i\ge 0\lambda_{i\bbV}=\lambda}E_{ii}^{(k_{R\bbV},k_{L\bbV})}$,
we have
\begin{align}\label{eq:nzl}
\braket{\gaa{\beta_\lambda}{b_\lambda}}{ \bbL_\bbA\lmk A_{\ket{\xi}\bra{\xi}}\rmk \gaa{\beta_\lambda}{b_\lambda}}
\braket{\gaa{\beta_\lambda}{b_\lambda}}{D_\xi \gaa{\beta_\lambda}{b_\lambda}}=
C_\xi^{-1}\neq 0.
\end{align}

Now we are ready to show that $\lambda_{i, \bbV} \neq \lambda_{j,\bbV}$ for any $1\le i<j\le k_{L\bbV}$. This proves that
the numbers
$\{\lambda_{i,\bbV}\}_{i=1}^{k_{L\bbV}}$ are distinct.
We assume that  $\lambda_{i, \bbV} =\lambda_{j,\bbV}$ for some $1\le i<j\le k_{L\bbV}$ and show a contradiction.
First, substituting  $\lambda:=\lambda_{i, \bbV}=\lambda_{j,\bbV}$,
$X=E_{j,j}^{(k_{R\bbV},k_{L\bbV})}$ and $\xi=f_{i}^{(k_{R\bbV},k_{L\bbV})}$
to (\ref{eq:lka}), we have
\begin{align*}
0=\sum_{\beta,\beta'=0}^{n_0-1}\sum_{b,b'=0}^{k_{L}}
\lmk \overline{\lambda_{L\beta}}\lambda_{L\beta'}
\overline{\lambda_{Rb}}\lambda_{Rb'}\rmk^{N}
\braket{\gaa{\beta}{b}}{ \bbL_\bbA\lmk A_{E_{j,j}^{(k_{R\bbV},k_{L\bbV})}}\rmk \gaa{\beta'}{b'}}
\braket{\gaa{\beta'}{b'}}{D_{f_{i}^{(k_{R\bbV},k_{L\bbV})}} \gaa{\beta}{b}},
\end{align*} for all $N\in\nan$.
As the numbers in (\ref{eq:llll}) are distinct, this equality implies
\begin{align}\label{eq:ijld}
\braket{\gaa{\beta}{b}}{ \bbL_\bbA\lmk A_{E_{j,j}^{(k_{R\bbV},k_{L\bbV})}}\rmk \gaa{\beta'}{b'}}
\braket{\gaa{\beta'}{b'}}{D_{f_{i}^{(k_{R\bbV},k_{L\bbV})}} \gaa{\beta}{b}}=0,
\end{align}
for all $\beta,\beta'=0,\ldots, n_0-1$ and $ b,b'=0,\ldots,k_L$.
For $\lambda=\lambda_{i, \bbV} =\lambda_{j,\bbV}$, there exists a unique pair $(\beta_\lambda,b_\lambda)$
such that $\lv \lambda\rv=\lv \lambda_{L\beta_\lambda}\lambda_{Rb_\lambda}\rv$.
From (\ref{eq:nzl}) with $\xi=f_{i}^{(k_{R\bbV},k_{L\bbV})}$, and $\xi=f_{j}^{(k_{R\bbV},k_{L\bbV})}$, we have
\begin{align}\label{eq:nzl1}
&\braket{\gaa{\beta_\lambda}{b_\lambda}}{ \bbL_\bbA\lmk A_{{E_{ii}^{(k_{R\bbV},k_{L\bbV})}}}\rmk \gaa{\beta_\lambda}{b_\lambda}}
\braket{\gaa{\beta_\lambda}{b_\lambda}}{D_{f_{i}^{(k_{R\bbV},k_{L\bbV})}} \gaa{\beta_\lambda}{b_\lambda}}
\neq 0,\nonumber\\
&\braket{\gaa{\beta_\lambda}{b_\lambda}}{ \bbL_\bbA\lmk A_{{E_{j,j}^{(k_{R\bbV},k_{L\bbV})}}}\rmk \gaa{\beta_\lambda}{b_\lambda}}
\braket{\gaa{\beta_\lambda}{b_\lambda}}{D_{f_{j}^{(k_{R\bbV},k_{L\bbV})}} \gaa{\beta_\lambda}{b_\lambda}}
\neq 0.
\end{align}
On the other hand, we have
\begin{align*}
\braket{\gaa{\beta_\lambda}{b_\lambda}}{ \bbL_\bbA\lmk A_{E_{j,j}^{(k_{R\bbV},k_{L\bbV})}}\rmk \gaa{\beta_\lambda}{b_\lambda}}
\braket{\gaa{\beta_\lambda}{b_\lambda}}{D_{f_{i}^{(k_{R\bbV},k_{L\bbV})}} \gaa{\beta_\lambda}{b_\lambda}}=0,
\end{align*} from  (\ref{eq:ijld}) .
This and (\ref{eq:nzl1}) can not hold simultaneously. Hence we obtain a contradiction.
Similarly, we can show that  the numbers
$\{\lambda_{i,\bbV}\}_{-k_{R\bbV}}^{-1}$ are distinct.
\end{proof}

\subsection{Properties of $\caD_l$}Throughout this subsection, let $2\le n_0\in\nan$, $k_R,k_L\in\nan\cup\{0\}$, and $\bbA$, $\bb(t)$ be given by (\ref{eq:defa}), (\ref{eq:pbp}).
Recall that $\caD_l$ is the
 the linear subspace of $\bigotimes_{i=0}^{l-1}\bbC^n$ spanned by $\xi_{\alpha,\beta,a,b}^{(l)}$, 
$\alpha,\beta\in\{1,\ldots,n_0\}$,
$a=0,1,\ldots, k_R$, $b=0,1,\ldots,k_L$, and that $G_l$ is 
the orthogonal  projection onto 
$\caD_l$.
In this subsection, we investigate properties of $\caD_l$.
First we show that the projections $G_l$ and $G_{l,\bbA}$ are close.
In order for that, we estimate $\Upsilon_l$ in (\ref{eq:el1}).
\begin{lem}\label{lem:ups}
There exists a constant $C_7>0$ such that 
\[
\lV  r_{\alpha}^l\Upsilon_l(x)\rV\le C_7\kappa^{\frac l2}\lV x\rV,
\]we
for all $\alpha=2,\ldots,n_0$, $l\in\nan$, and $x\in \Mat_{k_R+k_L+1}$.
\end{lem}
\begin{proof}
From (\ref{eq:el1}), it suffices to show that each $\zeta_{i_1\cdots i_k}^{(k)l}(\zeij{11}\otimes x)$,
is bounded by
${\rm const}\times  l^k\times \lV x\rV$.
By the definition 
\begin{align*}
&\lV \zeta_{i_1\cdots i_k}^{(k)l}\lmk \zeij{11}\otimes x\rmk\rV^2\\
&=
\sum_{1\le m_1<\cdots<m_k\le l}\sum_{j,j'=-k_R}^{k_L}\\
&\braket{\cnz{1}\otimes \fii{j}}
{
T_0^{m_1-1}\circ \Ad K_{i_1}\circ\cdots\circ T_0^{m_k-m_{k-1}-1}\circ\Ad{K_{i_k}}\circ T_0^{l-m_k}
\lmk
\zeij{11}\otimes x^*\eij{jj'}x
\rmk
\lmk \cnz{1}\otimes \fii{j'}\rmk}.
\end{align*}
Here, $T_0$ is a CP map given by $T_0:=\Ad\lmk R\otimes \Lambda_{\lal_R}\rmk$, and 
we have a bound $\lV T_0\rV=\lV T_0(1)\rV\le 1$.
From this and the above formula, we obtain the bound
\begin{align*}
\lV \zeta_{i_1\cdots i_k}^{(k)l}\lmk \zeij{11}\otimes x\rmk\rV^2
\le \sum_{1\le m_1<\cdots<m_k\le l}\sum_{j,j'=-k_R}^{k_L}
\lmk\max_{i=1,2,3}\left\{ \lV \Ad K_i\rV\right\} \rmk^k
\lV x\rV^2\le{\rm const}\cdot l^k\cdot \lV x\rV^2.
\end{align*}
This completes the proof.
\end{proof}

\begin{lem}\label{lem:dfsn}
There exist $l_3\in\nan$ and a constant $C_8>0$
such that for all $l\ge l_3$, 
\begin{enumerate}
\item the vectors $\xi_{\alpha,\beta,a,b}^{(l)}$, 
$\alpha,\beta\in\{1,\ldots,n_0\}$,
$a=0,1,\ldots, k_R$, and $b=0,1,\ldots,k_L$,,
are linearly independent, and
\item \begin{align*}
\lV
G_l-G_{l,\bbA}
\rV
\le C_8\kappa^{\frac l2}.
\end{align*}
\end{enumerate}

\end{lem}
\begin{proof}This follows from the fact that there exist $C'>0$ and $L_\bbA\in\nan$ such that
\begin{align}\label{eq:gee}
C'\lV X\rV_2
\le
\lV\Gamma_{l,\bbA}^{(R)}\lmk
X
\rmk\rV
\le (C')^{-1}\lV X\rV_2,\;\; \text{for all } L_{\bbA}\le l,\;\;
X\in\puuz\lmk
\mnz\otimes \mkk
\rmk\pddz
\end{align}
and 
Lemma \ref{lem:ups}.
(Recall Lemma 2.16 of \cite{Ogata1}.)
\end{proof}
\begin{lem}\label{lem:gtz}
For any $l\ge \max\{l_3, n_0^6(k_R+1)(k_L+1)\}$ (where $l_3$ is given in Lemma \ref{lem:dfsn}), $t\in[0,1]$,
and  $\alpha,\beta\in\{1,\ldots,n_0\}$,
$a=0,1,\ldots, k_R$, and $b=0,1,\ldots,k_L$, set
\begin{align}\label{eq:defhx}
\hat \xi_{\alpha,\beta,a,b}^{(l)}\lmk t\rmk:=\left\{
\begin{gathered}
\gbtr{l}\lmk
\Theta_{l,t}\lmk \zeij{\alpha\beta}\otimes \eij{-a,b}
\rmk\rmk,\quad \text{if}\quad t\in(0,1]\\
 \xi_{\alpha,\beta,a,b}^{(l)},\quad \text{if}\quad t=0
\end{gathered}
\right..
\end{align}
Let $\hat G_l(t)$ be the orthogonal projection onto
the subspace of $\bigotimes_{i=0}^{l-1}\bbC^n$ spanned by $\hat \xi_{\alpha,\beta,a,b}^{(l)}\lmk t\rmk$,
$\alpha,\beta\in\{1,\ldots,n_0\}$,
$a=0,1,\ldots, k_R$, and $b=0,1,\ldots,k_L$,
for each $t\in[0,1]$ and $l\ge \max\{l_3, n_0^6(k_R+1)(k_L+1)\}$.
Then the maps
$[0,1]\ni t\mapsto \hat \xi_{\alpha,\beta,a,b}^{(l)}\lmk t\rmk\in\bigotimes_{i=0}^{l-1}\bbC^n$,
$[0,1]\ni t\mapsto \hat G_l(t)\in \bigotimes_{i=0}^{l-1}\Mat_n$ are $C^\infty$.
\end{lem}
\begin{proof}
That $[0,1]\ni t\mapsto \hat \xi_{\alpha,\beta,a,b}^{(l)}\lmk t\rmk\in\bigotimes_{i=0}^{l-1}\bbC^n$ is $C^\infty$, is immediate from the proof of Lemma \ref{lem:flim}.
By Lemma \ref{lem:dfsn},  the vectors $\xi_{\alpha,\beta,a,b}^{(l)}$, 
are linearly independent for any $l\ge l_3$.
For any $t\in(0,1]$, $\bb(t)$ belongs to 
$\Class(n,n_0,k_R,k_L)\subset \ClassA$ from Lemma \ref{lem:tspec}.
Therefore, by Proposition 3.1 of \cite{Ogata1} and Lemma \ref{lem:klmain},
for any $l\ge  n_0^6(k_R+1)(k_L+1)\ge l_{\bb(t)}$, 
the vectors 
$
\left\{
\gbtr{l}\circ \Theta_{l,t}\lmk\zeij{\alpha\beta}\otimes \eij{-a,b}\rmk
\right\}_{\alpha,\beta=1,\ldots,n_0,a=0,\ldots,k_R,b=0,\ldots,k_L}
$
are linearly independent.
Therefore, by Lemma \ref{li}, $\hat G_l$ is $C^{\infty}$ for 
$l\ge \max\{l_3,n_0^6(k_R+1)(k_L+1)\}$.
\end{proof}
By the bijectivity of $\Theta_{l,t}$ on $\mnz\otimes \pu\mkk\pd$ and the definition of $G_l$, we have
$\hat G_l(t)=G_{l,\bbB(t)}$ for $t\in(0,1]$ and $\hat G_l(0)=G_l$.

Next we show the intersection property of $\caD_l$.
It is convenient to represent the vectors $\xi_{\alpha,\beta,a,b}^{(l)}$ in a matrix product form.
We set $\bbF\in \lmk\Mat_{n_0+1}\otimes\mkk\rmk^{\times n}$ by
\begin{align*}
&F_1:=\Lambda_{\lal_F}\otimes \Lambda_{\lal_R},\\
&F_2:= \ket{\chi_1^{(n_0+1)}}\bra{\eta_F}\otimes \Lambda_{\lal_R}
+\ket{\eta_F}\bra{\chi_{n_0+1}^{(n_0+1)}}\otimes \Lambda_{\lal_R}
+\Lambda_{\lal_F}\otimes V\Lambda_{\lal_R},\\
&F_\mu=0,\quad \mu\ge 3.
\end{align*}
Here, we set
\[
\eta_F:=\sum_{i=2}^{n_0}\chi_i^{(n_0+1)}\in\bbC^{n_0+1},\quad \lal_{F}:=\lmk 1,r_2,\ldots,r_{n_0}, 1\rmk\in\bbC^{n_0+1}.
\]
(Recall $r_\alpha$, $V$ from (\ref{eq:rd}), and (\ref{eq:vd}).)
With this definition, it is straight forward to check
\[
\Upsilon_l\lmk X\rmk=\Gamma_{l,\bbF}^{(R)}\lmk e_{1,n_0+1}^{(n_0+1)}\otimes X\rmk,
\]
for $X\in\pu\mkk\pd$, $l\in\nan$.
In order to represent  $\xi_{\alpha,\beta,a,b}^{(l)}$, we set ${\bbA'}$ as
\begin{align}\label{eq:defap}
{A'}_{\mu}:=A_\mu\oplus F_\mu\in 
\lmk\Mat_{2n_0-1}\otimes \mkk\rmk\oplus \lmk\Mat_{n_0+1}\otimes \mkk\rmk
\simeq
\lmk \Mat_{2n_0-1}\oplus\Mat_{n_0+1}\rmk \otimes \mkk,\quad \mu=1,\ldots,n.
\end{align}
Using these matrices, our $\xi_{\alpha,\beta,a,b}^{(l)}$ is represented as
\begin{align*}
\xi_{\alpha,\beta,a,b}^{(l)}
=\Gamma_{l{\bbA'}}^{(R)}\lmk \Theta_l\lmk \zeij{\alpha\beta}\rmk\otimes\eij{-a,b}\rmk.
\end{align*}
Here, $\Theta_l: \mnz\to \Mat_{2n_0-1}\oplus \Mat_{n_0+1}$ is a linear map given by
\begin{align}\label{eq:tl}
\Theta_l\lmk\zeij{\alpha\beta}\rmk
:=\eijz{-(\alpha-1),\beta-1}\oplus \lmk -\delta_{\alpha\beta}\lmk 1-\delta_{\alpha,1}\rmk \bar r_\alpha^l e_{1,n_0+1}^{(n_0+1)}\rmk,\quad
\alpha,\beta=1,\ldots, n_0.
\end{align}
Hence, we obtain
\begin{align}\label{eq:matd}
\caD_l=\Gamma_{l,{\bbA'}}^{(R)}\lmk
\lmk \Theta_l\otimes\id\rmk\lmk \mnz\otimes \pu\mkk\pd\rmk\rmk,
\end{align}
and
\begin{align}\label{eq:apb}
\lim_{t\downarrow 0}\Gamma_{l,\bbB(t)}^{(R)}\circ \Theta_{l,t}\lmk X\rmk
=\Gamma_{l,\bbA'}^{(R)}\circ \lmk \Theta_{l}\otimes\id\rmk\lmk X\rmk,\quad
X\in\mnz\otimes \mkk,\quad l\in\nan.
\end{align}

\begin{lem}\label{lem:intg}
There exists an $l_4\in\nan$ such that
\[
\caD_{l+1}=\lmk \caD_l\otimes\bbC^n\rmk\cap \lmk \bbC^n\otimes\caD_l\rmk,\quad l_4\le l.
\]
\end{lem}
\begin{proof}
The inclusion $\subset$ is easy. For each $l\ge \max\{2n_0^6(k_R+1)(k_L+1),l_3\}$ and $t\in(0,1]$, by Proposition 3.1 of \cite{Ogata1} and Lemma \ref{lem:klmain}, Lemma \ref{lem:tspec},
we have $l\ge  2n_0^6(k_R+1)(k_L+1)\ge2 l_{\bb(t)}\ge m_{\bb(t)}$. Therefore,  we have
\[
\caG_{l+1,\bb(t)}=\lmk \caG_{l,\bb(t)}\otimes\bbC^n\rmk\cap \lmk \bbC^n\otimes\caG_{l,\bb(t)}\rmk.
\]
For any $\eta\in\caD_{l+1}$, by Lemma \ref{lem:gtz},
we have
\begin{align*}
\eta=G_{l+1}\eta=\lim_{t\downarrow 0} G_{l+1,\bb(t)}\eta=
\begin{cases}
\lim_{t\downarrow 0} \lmk G_{l,\bb(t)}\otimes\unit\rmk G_{l+1,\bb(t)}\eta=\lmk G_l\otimes\unit\rmk\eta\\
\lim_{t\downarrow 0} \lmk\unit\otimes G_{l,\bb(t)}\rmk G_{l+1,\bb(t)}\eta=\lmk\unit \otimes G_l\rmk\eta.
\end{cases}
\end{align*}
This proves $\caD_{l+1}\subset \lmk \caD_l\otimes\bbC^n\rmk\cap \lmk \bbC^n\otimes\caD_l\rmk$.

The proof of opposite inclusion is similar to the proof of Lemma 2.14 of \cite{Ogata1}.

By the routine argument using Lemma C.7 of \cite{Ogata1}, we see that there exists an $l_4'\in\nan$
such that
\begin{align*}
\caK_l\lmk {\bbA'}\rmk =\caV_l,
\end{align*}
for all $l\ge l_4'$.
The subspace $\caV_l$ is defined by
\begin{align}
\caV_l
=\spa&\lmk \left\{
\begin{gathered}
\lmk \unit_{2n_0-1}\oplus\unit_{n_0+1}\rmk, 
\lmk E_{0,\beta-1}^{(n_0-1,n_0-1)}\oplus e_{1\beta}^{(n_0+1)}\rmk,\\
\lmk E_{-(\alpha-1),0}^{(n_0-1,n_0-1)}\oplus e_{\alpha,n_0+1}^{(n_0+1)}\rmk,\\
\lmk 0_{2n_0-1}\oplus e_{1,n_0+1}^{(n_0+1)}\rmk,\\
\lmk  E_{-(\alpha-1),\beta-1}^{(n_0-1,n_0-1)}\oplus 0_{n_0+1}\rmk 
\end{gathered}\right\}_{\alpha,\beta=2,\ldots, n_0}
\otimes\left\{
\begin{gathered}
\unit, E_{-a,0}^{(k_R,k_L)}, \\
E_{0,b}^{(k_R,k_L)}, E_{-a,b}^{(k_R,k_L)}
\end{gathered}
\right\}_{a=1,\ldots, k_R, b=1,\ldots, k_L}\rmk\hat\Lambda^l.
\end{align}
Here, we set
$
\hat \Lambda:= \lmk \Lambda_{\lal_L}\oplus \Lambda_{\lal_F}\rmk \otimes\Lambda_{\lal_R}
$.

Now we are ready to prove $\caD_{l+1}=\lmk \caD_l\otimes\bbC^n\rmk\cap \lmk \bbC^n\otimes\caD_l\rmk$
for $l\ge l_4'+1$.
Let $l_4'+1\le l\in \nan$ and $\Phi\in\lmk \caD_l\otimes\bbC^n\rmk\cap \lmk \bbC^n\otimes\caD_l\rmk$.
We show that $\Phi\in \caD_{l+1}$.
By (\ref{eq:matd}), there exist $C_\mu,D_\nu\in \mnz\otimes\pu\mkk\pd$, $\mu,\nu=1,\ldots, n$ such that
\begin{align}\label{eq:pcd}
\Phi=\sum_{\mu=1}^n \Gamma_{l,{\bbA'}}^{(R)}\lmk\lmk \Theta_l\otimes\id\rmk\lmk C_\mu\rmk\rmk\otimes\psi_\mu
=\sum_{\nu=1}^n \psi_\nu\otimes\Gamma_{l,{\bbA'}}^{(R)}\lmk\lmk \Theta_l\otimes\id\rmk\lmk D_\nu\rmk\rmk.
\end{align}
Therefore, for all $\varpi^{(l-1)}\in \{1,\ldots,n\}^{\times l-1}$ and $\mu,\nu\in\{1,\ldots,n\}$, we have
\begin{align*}
\Tr\lmk
\lmk \Theta_l\otimes\id\rmk\lmk C_\mu\rmk\lmk {A'}_\nu\widehat{{A'}_{\varpi^{(l-1)}}}
\rmk^*\rmk
=\Tr\lmk
\lmk \Theta_l\otimes\id\rmk\lmk D_\nu\rmk\lmk \widehat{{A'}_{\varpi^{(l-1)}}}{A'}_\mu
\rmk^*\rmk.
\end{align*}
From this, we see that for each $\mu,\nu\in\{1\ldots,n\}$,
\begin{align*}
 {A'}_\nu^* \lmk \Theta_l\otimes\id\rmk\lmk C_\mu\rmk-\lmk \Theta_l\otimes\id\rmk\lmk D_\nu\rmk{A'}_\mu^*
\end{align*}
is orthogonal to $\caK_{l-1}(\bbA')$ with respect to the inner product 
$\braket{}{}_{\Tr}$ given by $\braket{X}{Y}_{\Tr}=\Tr X^*Y$.
As $l-1\ge l_4'$, we have
 $\caK_{l-1}(\bbA')=\caV_{l-1}$.

Note that
$
{A'_1}^{-1}\caV_l\subset \caV_{l-1}
$
from the definition of $\caV_l$.
From this, if $X\in \lmk \Mat_{2n_0-1}\oplus\Mat_{n_0+1}\rmk \otimes \mkk$ belongs to the orthogonal complement of $\caV_{l-1}$ with respect to
$\braket{}{}_{\Tr}$, then $\lmk {A'_1}^*\rmk^{-1} X$ belongs to the orthogonal complement of $\caV_{l}$.
Applying this to $ X={A'}_1^* \lmk \Theta_l\otimes\id\rmk\lmk C_\mu\rmk-\lmk \Theta_l\otimes\id\rmk\lmk D_1\rmk{A'}_\mu^*$,
we see for each $\mu=1,\ldots,n$, that 
\[
Z_\mu:=\lmk \Theta_l\otimes\id\rmk\lmk C_\mu\rmk-\lmk {A'_1}^*\rmk^{-1}\lmk \Theta_l\otimes\id\rmk\lmk D_1\rmk{A'}_\mu^*
\]
is orthogonal to $\caV_l=\caK_l(\bbA')$.
This means that
$Z_\mu$ is in the kernel of $\Gamma_{l,\bbA'}^{(R)}$ for each $\mu=1,\ldots,n$.

By the definition of $\Theta_l$, we have
\begin{align*}
\lmk {A'_1}^*\rmk^{-1}\lmk \lmk \Theta_l\otimes \id\rmk\lmk \zeij{\alpha\beta}\otimes\eij{-a,b}\rmk\rmk
=\lmk \Theta_{l+1}\otimes\id\rmk\lmk \lmk\overline{\lambda_{L, -(\alpha-1)}}\rmk^{-1}\lmk \overline{\lambda_{R,-a}}\rmk^{-1}\zeij{\alpha\beta}\otimes\eij{-a,b}\rmk.
\end{align*}
Therefore, we have
\begin{align*}
\lmk {A'_1}^*\rmk^{-1}\lmk \Theta_l\otimes\id\rmk
\lmk \mnz\otimes\pu\mkk\pd\rmk
\subset \lmk \Theta_{l+1}\otimes\id\rmk\lmk \mnz\otimes \pu\mkk\pd\rmk.
\end{align*}
In particular, we have
\[
\lmk {A'_1}^*\rmk^{-1}\lmk \Theta_l\otimes\id\rmk\lmk D_1\rmk
=\lmk \Theta_{l+1}\otimes\id\rmk\lmk W\rmk,
\]
with some $W\in \mnz\otimes \pu\mkk\pd$.
Hence for each $\mu=1,\ldots,n$, we obtain
\begin{align*}
\lmk \Theta_l\otimes\id\rmk\lmk C_\mu\rmk
=\lmk \Theta_{l+1}\otimes\id\rmk\lmk W\rmk{A'}_\mu^*+Z_\mu,
\end{align*}
where $W\in \mnz\otimes \pu\mkk\pd$ and $Z_\mu\in\ker\Gamma_{l,\bbA'}^{(R)}$.
Substituting this to (\ref{eq:pcd}), we obtain
\begin{align*}
\Phi=\sum_{\mu=1}^n \Gamma_{l,{\bbA'}}^{(R)}\lmk\lmk \Theta_l\otimes\id\rmk\lmk C_\mu\rmk\rmk\otimes\psi_\mu
= \Gamma_{l+1,{\bbA'}}^{(R)}\lmk  \lmk \Theta_{l+1}\otimes\id\rmk\lmk W\rmk \rmk\in\caD_{l+1}.
\end{align*}
This completes the proof.
\end{proof}
The Hamiltonian $H_{\Phi_{1-G_m}}$ given by the interaction $1-G_m$ via the formula (\ref{hamdef}), (\ref{GenHamiltonian}) is gapped with respect to the open boundary conditions:
\begin{lem}\label{lem:ggap}
For any $m,N\in\nan$ with $\max\{l_3,l_4\} \le m\le N$, the kernel of 
$\lmk H_{\Phi_{1-G_m}}\rmk_{[0,N-1]}$ is equal to $\caD_N$,
and its dimension is  $n_0^2(k_L+1)(k_R+1)$.
For each $m\ge \max\{l_3,l_4\}$, there exist a $\gamma_m>0$ and an $N_m\in\nan$ such that
\[
\gamma_m\lmk 1-G_{N}\rmk
\le \lmk H_{\Phi_{1-G_m}}\rmk_{[0,N-1]},\;\;\text{ for all } N\ge N_m.
\] 
(The numbers $l_3,l_4$ are given in Lemma \ref{lem:dfsn}, Lemma \ref{lem:intg}.)
\end{lem}
\begin{proof}
From Lemma \ref{lem:intg}, $\{\caD_l\}_l$ satisfies the intersection property. 
From Lemma \ref{lem:aprop},
 $\{\caG_{l,\bbA}\}_l$ satisfies {\it Condition 1} of \cite{Ogata1} Definition 2.1.
By the second statement of Lemma \ref{lem:dfsn}, the differemce between $G_l$ and $G_{l,\bbA}$ decays exponentially fast, as $l\to\infty$.
Therefore, $\{\caD_l\}_l$ satisfies the second property of {\it Condition 1}.
Hence we conclude that $\{\caD_l\}_l$ satisfies {\it Condition 1} .
Therefore, by the Theorem of Nachtergaele, \cite{Nachtergaele:1996vc}(see Theorem 2.2 of \cite{Ogata1} for the version we use)
we obtain the claim of the current Lemma.
\end{proof}
Note that for $m\ge \max\{l_3,l_4\}$, we have $n^m>n_0^2(k_R+1)(k_L+1)$ by Remark \ref{rem:nt}.
The Hamiltonian $H_{\Phi_{1-G_{m_1}}}$ and $H_{\Phi_{m_2,\bbA}}$ are equivalent with respect to the typeII-$C^1$-classification.
\begin{lem}\label{lem:gaeq}
Let $m_1,m_2\in\nan
$ with $\max\{l_3,l_4\} \le m_1$ and 
$\max\{l_\bbA,m_\bbA\} \le m_2$. (Recall the definition of $l_\bbA$ from Lemma \ref{lem:asp}.)
For each $t\in[0,1]$ we define $\Phi(t)\in\caJ$ by
\[
\Phi(X; t):=(1-t)\Phi_{1-G_{m_1}}(X)+t\Phi_{m_{2},\bbA}(X),\quad X\in {\mathfrak S}_{\bbZ}.
\]
Then
 we have
  $H_{\Phi_{1-G_{m_1}}}\simeq_{II}H_{\Phi_{m_2,\bbA}}$.
\end{lem}
\begin{proof}
It is proven by the same argument as in the proof of Corollary 1.4 of \cite{Ogata2}.
\end{proof}
At the end of this subsection, we prove Lemmas which we will use in the next section to discuss about the bulk classification. 
We introduce the following notation.
\begin{notation}\label{nota:ts}
Let $H_\Phi$ be a frustration free Hamiltonian given by a positive interaction 
$\Phi\in \caJ$.
For each $\Gamma\subset\bbZ$, we set
\[
\widetilde{\caS_{\bbZ,\Gamma}}\lmk H_\Phi\rmk:=
\left\{
\varphi\mid \varphi \;\text{is a state on} \;\caA_\bbZ \;\text{with}\;
\varphi\lmk\Phi(X)\rmk=0\;\text{for all} \;X\in \mathfrak S_{\Gamma}
\right\}.
\]
\end{notation}
\begin{rem}\label{rem:ic}
By definition, we clearly have
$\tilde\caS_{\bbZ,\Gamma_1}\lmk H_\Phi\rmk\subset \tilde\caS_{\bbZ,\Gamma_2}\lmk H_\Phi\rmk$
if $\Gamma_2\subset \Gamma_1$.
\end{rem}
From {\it 2.} of Lemma \ref{lem:dfsn}, we obtain the following.
\begin{lem}\label{lem:sag}
For any $m_1,m_2\in\nan$ with $m_1\ge l_4$, $m_2\ge  m_{\bbA}$, and $ M\in\nan$, we have
\begin{align}
\caS_{\bbZ}\lmk H_{\Phi_{1-G_{m_1}}}\rmk =\caS_{\bbZ}\lmk H_{\Phi_{m_2,\bbA}}\rmk,\quad
\widetilde{\caS_{\bbZ,[-M,M]^{c}}}\lmk H_{\Phi_{1-G_{m_1}}}\rmk
=\widetilde{\caS_{\bbZ,[-M,M]^{c}}}\lmk H_{\Phi_{m_2,\bbA}}\rmk.
\end{align}
\end{lem}
\begin{defn}
We say that an interaction $\Phi\in \caJ$ satisfies {\it Condition 6} if the followings hold.
\begin{enumerate}
\item There exists a state $\omega_\infty$ on $\caA_{\bbZ}$ such that
$\caS_{\bbZ}\lmk H_{\Phi}\rmk= \left\{\omega_\infty\right\}$.
\item For any $M\in\nan$ and $\varphi\in\widetilde{\caS_{\bbZ,[-M,M]^{c}}}\lmk H_{\Phi}\rmk$,
$\varphi$ is quasi-equivalent to $\omega_\infty$.
\end{enumerate}
\end{defn}
\begin{lem}\label{lem:24}
Let $h\in \caA_{\bbZ}$ be a positive operator given by either of the followings:
\begin{align*}
h:=\left\{
\begin{gathered}
1-G_{m,\bbB},\quad \bbB\in\Class A,\quad m\ge m_{\bbB},\\
1-G_{m,\bbA}, \quad \bbA : \text{given in }\quad (\ref{eq:defa}),\quad m\ge m_{\bbA} \\
1-G_m, \quad G_m : \text{given in }\quad (\ref{eq:gldef}),\quad m\ge {l_4} 
\end{gathered}
\right..
\end{align*}
Then
$\Phi_h$  given by the formula (\ref{hamdef}) satisfies
the {\it Condition 6}.
(Recall $l_4$ given in Lemma \ref{lem:intg}.)
\end{lem} 
\begin{proof}
Let us first consider  $1-G_{m,\bbB}$, $1-G_{m,\bbA}$ .
The first condition of {\it Condition 6} for $1-G_{m,\bbB}$, $1-G_{m,\bbA}$ is already checked.
From Lemma \ref{lem:mme} and Lemma \ref{lem:qe}, the second property of {\it Condition 6} holds
for $1-G_{m,\bbA}$.
As the corresponding property holds for  $1-G_{m,\bbB}$
from Lemma 3.25 of \cite{Ogata1}, the latter property holds also for $1-G_{m,\bbB}$.

Let us consider $1-G_m$.
From Lemma \ref{lem:sag} and {\it Condition 6} for $1-G_{m,\bbA}$ ,
we obtain  {\it Condition 6} for $1-G_m$.
\end{proof}
\begin{defn}
Let $\Phi_0,\Phi_1\in \caJ$ be positive interactions.
For each $\Lambda \in {\mathfrak S}_{\bbZ}$ and $i=0,1$, we denote by $G_{\Lambda,i}$
the orthogonal projection onto $\lmk \ker H_{\Phi_i}\rmk_{\Lambda}$.
We say the pair $(\Phi_0,\Phi_1)$ satisfies {\it Condition 7} if the followings hold.
\begin{enumerate}
\item The Hamiltonians $H_{\Phi_0},H_{\Phi_1}$ are frustration free.
\item The Hamiltonian $H_{\Phi_1}$ is gapped with respect to the open boundary conditions.
\item There exist constants $C>0$ and $0<r<1$ such that
\begin{align}\label{eq:g1g0d}
\lV G_{[0,N-1],0}G_{[0,N-1],1}-G_{[0,N-1],1}\rV\le C r^N,\quad \text{for all}\;N\in\nan.
\end{align}
\item There exists a state $\omega_\infty$ on $\caA_{\bbZ}$ such that
$\caS_{\bbZ}\lmk H_{\Phi_0}\rmk= \caS_{\bbZ}\lmk H_{\Phi_1}\rmk=\left\{\omega_\infty\right\}$.
\item For any $M\in\nan$ and $\varphi\in\widetilde{\caS_{\bbZ,[-M,M]^{c}}}\lmk H_{\Phi_1}\rmk$,
$\varphi$ is quasi-equivalent to $\omega_\infty$.
\end{enumerate}
\end{defn}
\begin{defn}
Let $\Phi_0,\Phi_1\in \caJ$ be positive interactions. We say the pair $(\Phi_0,\Phi_1)$ satisfies {\it Condition 8} if 
it satisfies  {\it Condition 7} and either (i) or (ii) of the following conditions holds.
\begin{enumerate}
\item[(i)] 
\begin{enumerate}
\item[(a)]There exists a unique $\alpha_{\Phi_0}$-ground state, and
\item[(b)]$\Phi_0\in\caJ_B$.
\end{enumerate}\item[(ii)]
\begin{enumerate}
\item[(a)] The Hamiltonian $H_{\Phi_0}$ is gapped with respect to the open boundary conditions, and
\item[(b)]for any $M\in\nan$ and $\varphi\in\widetilde{\caS_{\bbZ,[-M,M]^{c}}}\lmk H_{\Phi_0}\rmk$,
$\varphi$ is quasi-equivalent to $\omega_\infty$.
(Here, $\omega_\infty$ is the state given by {\it 4.} of {\it Condition 7}.)
\end{enumerate}
\end{enumerate}
\end{defn}
From Lemma \ref{lem:a1a5a}, Lemma \ref{lem:dfsn}, Lemma \ref{lem:intg}, Lemma \ref{lem:ggap},  Lemma  \ref{lem:sag} and Lemma \ref{lem:24}, we have the following.
\begin{lem}\label{lem:ag7}
Let $m_1,m_2$ with $\max\{l_3,l_4\} \le m_1$ and 
$2l_\bbA \le m_2$. 
Then $ (\Phi_{1-G_{m_1}},\Phi_{m_2,\bbA})$ satisfies {\it Condition 8}.
\end{lem}
By Corollary 1.4 of \cite{Ogata2}, Theorem 1.18 of \cite{Ogata1}, Lemma \ref{lem:a1a5a}, Lemma \ref{lem:dfsn}, and Lemma \ref{lem:24},
we obtain the following.
\begin{lem}\label{lem:av7}
Let $m_1,m_2$ with $2l_\bbA  \le m_1$ and 
$2l_\bbV\le m_2$. 
Then $ (\Phi_{{m_1,\bbA}},\Phi_{m_2,\bbV})$ satisfies {\it Condition 8}.
\end{lem}

\subsection{The overlaps of $G_{N,\bbB(t)}$}
Throughout this subsection let $2\le n_0\in\nan$, $k_R,k_L\in\nan\cup\{0\}$, and $\bbA$, $\bb(t)$ be given by (\ref{eq:defa}), (\ref{eq:pbp}).
As $\bbB(t)$ belongs to $\Class A$ if $t\in (0,1]$ from Lemma \ref{lem:tspec}, 
we obtain $e_{\bbB(t)}$ and $\varphi_{\bbB(t)}$ given in Proposition 3.1 of \cite{Ogata1}.
In this subsection, we show the following Lemma.
\begin{lem}\label{lem:overl}
There exist $C_9$, $0<s_4<1$, and $l_5\in\nan$ such that
\begin{align}
\sup_{t\in[0,1]}\lV \lmk \unit_{[0,N-l]}\otimes \hat G_{l}(t)\rmk\lmk \hat G_{N}(t)\otimes\unit_{\{N\}}-\hat G_{N+1}(t)\rmk \rV
\le  n_0^2\lmk k_R+1\rmk\lmk k_L+1\rmk C_9\lmk s_4^l+s_{4}^{N-l}\rmk,
\end{align} 
for all $N,l\in\nan$ such that $N\ge l\ge \max\{l_3,l_5,  n_0^6\lmk k_R+1\rmk\lmk k_L+1\rmk\}$.
(Recall $l_3$ given in Lemma \ref{lem:dfsn}.)
\end{lem}

First we estimate various sesqui-linear forms.
\begin{lem}\label{lem:26}
There exist constants $C_3,C_4>0$, $0<s_2<1$ such that
\begin{align}
&\lv
\braket{ \Gamma_{l,\bbB(t)}^{(R)}\circ \Theta_{l,t}\lmk X_1\rmk}{ \Gamma_{l,\bbB(t)}^{(R)}\circ\Theta_{l,t}\lmk X_2\rmk}-\varphi_{\bbB(t)}\lmk X_1^*a_t e_{\bbB(t)}a_tX_2\rmk
\rv
\le C_3 s_2^l\lV X_1\rV\lV X_2\rV,\label{eq:26}\\
&\lv
\braket{ \Gamma_{l,\bbB(t)}^{(R)}\lmk \Theta_{l+1,t}\lmk X_1\rmk B_{\mu}(t)^*\rmk }
{\Gamma_{l,\bbB(t)}^{(R)}\circ\Theta_{l,t}\lmk X_2\rmk}-\varphi_{\bbB(t)}\lmk B_{\mu}(t) X_1^*a_t e_{\bbB(t)}a_tX_2\rmk
\rv
\le C_3 s_2^l\lV X_1\rV\lV X_2\rV,\label{eq:272}\\
&\lv \varphi_{\bbB(t)}\lmk X_1^*a_t e_{\bbB(t)}a_tX_2\rmk\rv\le C_4 \lV X_1\rV\lV X_2\rV,\label{eq:peb1}\\
&\lv \varphi_{\bbB(t)}\lmk B_{\mu}(t) X_1^*a_t e_{\bbB(t)}a_tX_2\rmk\rv\le C_4 \lV X_1\rV\lV X_2\rV,\label{eq:psb}\\
&\text{for all}\;t\in(0,1], \; l\in\nan,\; X_1,X_2\in \mnz\otimes \pu\mkk\pd,\; \mu\in\{1,\ldots,n\}\nonumber.
\end{align}
\end{lem}
\begin{proof}
We prove (\ref{eq:272}) and (\ref{eq:psb}). The proofs of (\ref{eq:26}) and (\ref{eq:peb1}) are the same.
First, from the definition, we have
\begin{align*}
\Gamma_{l,\bbB(t)}^{(R)}\lmk \Theta_{l+1,t}\lmk X_1\rmk B_{\mu}(t)^*\rmk
=\sum_{\mu^{(l)}\in \{1,\ldots,n\}^{\times l}}\lmk \Tr\Theta_{l+1,t}\lmk X_1\rmk \lmk B_{\mu^{(l)}}(t)B_\mu(t)\rmk^*\rmk \ws{l}.
\end{align*} 
In the $\Tr$ on the right hand side, we have $l+1$ number of $B_\mu(t)$s and $\Theta_{l+1,t}\lmk X_1\rmk$.
Therefore, from the proof of Lemma \ref{lem:flim},
we see that it is of the form
$r_{T_{\tilde \oo(t)}}^{-\frac {l+1}2}\times\text{a polynimial of }t$.
Similarly, $\Gamma_{l,\bbB(t)}^{(R)}\circ\Theta_{l,t}\lmk X_2\rmk$ is 
of the form
$r_{T_{\tilde \oo(t)}}^{-\frac l2}\times\text{a polynimial of }t$.
Therefore, we have
\begin{align}\label{eq:pt}
\lim_{t\downarrow 0}t\braket{ \Gamma_{l,\bbB(t)}^{(R)}\lmk \Theta_{l+1,t}\lmk X_1\rmk B_{\mu}(t)^*\rmk }
{\Gamma_{l,\bbB(t)}^{(R)}\circ\Theta_{l,t}\lmk X_2\rmk}
=0.
\end{align}
Next, let $0<s_1<1$ be the constant given in Lemma \ref{lem:s1}. and set $S_1:=\{z\in\bbC\mid |z|=\frac{1+s_1}{2}\}\cup
\{z\in\bbC\mid |z-1|=\frac{1-s_1}{2}\}$. 
Then by Spectral Property II of $T_{\bb(t)}$, $t\in[0,1]$, and the definition of $s_1$ from Lemma \ref{lem:s1},
we have
\begin{align}\label{eq:tdecom}
T_{\bbB(t)}^l=P_{\{1\}}^{T_{\bb(t)}}+T_{\bbB(t)}^l\lmk \id -P_{\{1\}}^{T_{\bb(t)}}\rmk
=\lmk \oint_{|z-1|=\frac{1-s_1}2}dz+\oint_{ |z|=\frac{1+s_1}{2}}dz z^l \rmk
\lmk z-T_{\bbB(t)}\rmk^{-1}
\end{align}
From this and a routine calculation from \cite{Ogata1}, we obtain
\begin{align}\label{eq:rt}
&\braket{ \Gamma_{l,\bbB(t)}^{(R)}\lmk \Theta_{l+1,t}\lmk X_1\rmk B_{\mu}(t)^*\rmk }
{\Gamma_{l,\bbB(t)}^{(R)}\circ\Theta_{l,t}\lmk X_2\rmk}\nonumber\\
&=\sum_{\alpha,\beta=1}^{n_0}\sum_{i,j=-k_R}^0
\braket{\cnz{\alpha}\otimes\fii{i}}
{T_{\bbB(t)}^l\lmk B_\mu(t)\Theta_{l+1,t}\lmk X_1\rmk^*\ket{\cnz{\alpha}\otimes\fii{i}}
\bra{\cnz{\beta}\otimes\fii{j}}\Theta_{l,t}\lmk X_2\rmk\rmk\cnz{\beta}\otimes\fii{j}}\nonumber\\
&=\lmk \oint_{|z-1|=\frac{1-s_1}2}dz+\oint_{ |z|=\frac{1+s_1}{2}}dz z^l \rmk\sum_{\alpha,\beta=1}^{n_0}\sum_{i,j=-k_R}^0\nonumber\\
&\braket{\cnz{\alpha}\otimes\fii{i}}
{\lmk z-T_{\bbB(t)}\rmk^{-1}\lmk B_\mu(t)\Theta_{l+1,t}\lmk X_1\rmk^*\ket{\cnz{\alpha}\otimes\fii{i}}
\bra{\cnz{\beta}\otimes\fii{j}}\Theta_{l,t}\lmk X_2\rmk\rmk\cnz{\beta}\otimes\fii{j}}
\end{align} 
for each $t\in(0,1]$.
Now, we expand $\lmk z-T_{\bbB(t)}\rmk^{-1}$, $ B_\mu(t)$, as
$\lmk z-T_{\bbB(t)}\rmk^{-1}=\sum_{j=0,1} t^{j}A_j\lmk z\rmk+t^2A_2\lmk z,t\rmk$,
$B_\mu(t)=\sum_{j=0,1}t^jW_{j}+t^2W_2(t)$, $t\in[0,1]$, $z\in S_1$.
Here $A_j\lmk z\rmk$, $A_2\lmk z,t\rmk$, $W_j$, $W_2(t)$ are operators on $\mnz\otimes\mkk$
such that 
$\sup_{z\in S_1 }\lV A_j\lmk z\rmk\rV$, $\sup_{t\in[0,1],z\in S_1 }\lV A_2\lmk z,t\rmk\rV$,
$\sup_{t\in[0,1]}\lV W_2(t)\rV$ are finite.
Substituting these in  (\ref{eq:rt}), the right hand side of (\ref{eq:rt})
can be expanded as $t^{-2}b_2\lmk l,  X_1,X_2\rmk +t^{-1}b_1\lmk l, X_1,X_2\rmk+b_0\lmk t, l, X_1,X_2\rmk$.
Here, for the map $b_0:(0,1]\times \nan\times \lmk \Mat_{n_0}\otimes \pu\mkk\pd\rmk^{\times 2} \ni (t,l,X_1,X_2)\mapsto b_{0}\lmk t,l, X_1, X_2\rmk\in \bbC$, there exists
a constant $C_4>0$ such that 
\begin{align}\label{eq:bif}
\sup_{t\in(0,1]} \sup_{l\in\nan}\lv b_{0}\lmk t,l, X_1, X_2\rmk\rv\le C_4 \lV X_1\rV\lV X_2\rV.
\end{align}	
Furthermore, there exists a limit
$b_\infty\lmk t,X_1,X_2\rmk=\lim_{l\to\infty} b_{0}\lmk t,l, X_1, X_2\rmk$.
From (\ref{eq:bif}), we have
\begin{align}\label{eq:biff}
\sup_{t\in(0,1]} \lv b_{\infty}\lmk t,X_1, X_2\rmk\rv\le C_4 \lV X_1\rV\lV X_2\rV.
\end{align}	
There exist constants $C_3>0$ and  $0<s_2<1$ such that 
\begin{align}\label{eq:xlb}
\sup_{t\in(0,1]}\lv b_\infty\lmk t,X_1,X_2\rmk-b_{0}\lmk t,l, X_1, X_2\rmk\rv\le C_3 s_2^l\lV X_1\rV\lV X_2\rV,\quad
X_1,X_2\in \mnz\otimes \pu\mkk\pd\quad l\in\nan.
\end{align}
On the other hand, from the property (\ref{eq:pt}),
we know that $\lim_{t\downarrow 0} t^{-1}b_2\lmk l,  X_1,X_2\rmk +b_1\lmk l, X_1,X_2\rmk+tb_0\lmk t, l, X_1,X_2\rmk=0$.
This means $b_2\lmk l,  X_1,X_2\rmk =b_1\lmk l, X_1,X_2\rmk=0$.
Substituting these, we obtain
\begin{align}\label{eq:lcn}
\braket{ \Gamma_{l,\bbB(t)}^{(R)}\lmk \Theta_{l+1,t}\lmk X_1\rmk B_{\mu}(t)^*\rmk }
{\Gamma_{l,\bbB(t)}^{(R)}\circ\Theta_{l,t}\lmk X_2\rmk}=b_0\lmk t, l, X_1,X_2\rmk.
\end{align} 
Take $l\to\infty$ limit of this equation.
From (\ref{eq:xlb}) and the first equality of (\ref{eq:rt}) combined with the spectral property of $T_{\bbB(t)}$,
we obtain
\begin{align}\label{eq:li}
\varphi_{\bbB(t)}\lmk B_{\mu}(t) X_1^*a_t e_{\bbB(t)}a_tX_2\rmk
=b_\infty\lmk t,  X_1,X_2\rmk.
\end{align}
From (\ref{eq:xlb}), (\ref{eq:lcn}), and (\ref{eq:li}),
we obtain (\ref{eq:272}).
Furthremore, from (\ref{eq:biff}) and (\ref{eq:li}), we obtain (\ref{eq:psb}).
\end{proof}
\begin{lem}\label{lem:pud}
There exist constants $C_5, C_6>0$ such that
\begin{align}
C_5\Tr X^*X\le \varphi_{\bbB(t)}\lmk X^*a_t e_{\bbB(t)}a_tX\rmk
\le C_6 \Tr X^*X,\quad X\in \mnz\otimes \pu\mkk\pd.
\end{align}
\end{lem}
\begin{proof}
We define a linear map $\iota:\mnz\to \Mat_{2n_0-1}$ by $\iota\lmk\zeij{\alpha\beta}\rmk=
\eijz{-(\alpha-1),\beta-1}$, $\alpha,\beta=1,\ldots,n_0$. By Lemma \ref{lem:ups},
we have
\begin{align}\label{eq:aap}
\lV
\Gamma_{l,\bbA}^{(R)}\circ \lmk \iota\otimes\id\rmk\lmk X\rmk
-\Gamma_{l,\bbA'}^{(R)}\circ \lmk \Theta_{l}\otimes\id\rmk\lmk X\rmk
\rV
\le n_0C_7\kappa^{\frac l2}\lV X\rV,\quad
l\in\nan,\quad X\in \mnz\otimes \pu\mkk\pd.
\end{align}
By (33) of \cite{Ogata1} with Lemma \ref{lem:aprop}, and  (\ref{eq:26}), (\ref{eq:aap}),
we have a ($t$-independent) sequence of  positive numbers
$\delta_l$, $l\in\nan$ with  $\delta_l\to 0$, $l\to\infty$, such that
\begin{align*}
&\lv \varphi_{\bbB(t)}\lmk X_0^*a_t e_{\bbB(t)}a_tX_1\rmk-\varphi_{\bbA}\lmk
\lmk
\lmk \iota\otimes \id\rmk\lmk X_0\rmk
\rmk^*e_{\bbA}
\lmk \iota\otimes \id\rmk\lmk X_1\rmk
\rmk\rv\\
&\le
\lv \braket{ \Gamma_{l,\bbB(t)}^{(R)}\circ \Theta_{l,t}\lmk X_0\rmk}{ \Gamma_{l,\bbB(t)}^{(R)}\circ\Theta_{l,t}\lmk X_1\rmk}
-\braket{ \Gamma_{l,\bbA'}^{(R)}\circ \lmk \Theta_{l}\otimes\id\rmk\lmk X_0\rmk}{ \Gamma_{l,\bbA'}^{(R)}\circ\lmk \Theta_{l}\otimes \id\rmk
\lmk X_1\rmk}\rv
+\delta_l.
\end{align*}
We first take  $\limsup_{t\downarrow 0}$, applying (\ref{eq:apb}) for this inequality, and then take $l\to\infty$ limit. Hence we obtain
\begin{align}\label{eq:tlim}
\lim_{t\downarrow 0}\varphi_{\bbB(t)}\lmk X_0^*a_t e_{\bbB(t)}a_tX_1\rmk
=\varphi_{\bbA}\lmk
\lmk
\lmk \iota\otimes \id\rmk\lmk X_0\rmk
\rmk^*e_{\bbA}
\lmk \iota\otimes \id\rmk\lmk X_1\rmk
\rmk,\quad
X_0,X_1\in \mnz\otimes \pu\mkk\pd.
\end{align}
As $\iota\otimes \id$ is a bijection from $\mnz\otimes\pu\mkk\pd$ to $\puuz\lmk \Mat_{2n_0-1}\otimes\mkk\rmk\pddz$
and $s\lmk e_\bbA\rmk=\puuz$ and $s\lmk \varphi_\bbA\rmk=\pddz$,
$(X_0, X_1)\mapsto \varphi_{\bbA}\lmk
\lmk
\lmk \iota\otimes \id\rmk\lmk X_0\rmk
\rmk^*e_{\bbA}
\lmk \iota\otimes \id\rmk\lmk X_1\rmk
\rmk$, $X_0,X_1\in\mnz\otimes\pu\mkk\pd$ defines an inner product on $\mnz\otimes\pu\mkk\pd$.
From this and (\ref{eq:tlim}), and the finite dimensionality of  $\mnz\otimes\pu\mkk\pd$,
we prove the claim of the Lemma.
\end{proof}
From the upper bound of the last Lemma, we obtain the following.
\begin{cor}\label{lem:aea}
We have
$\sup_{t\in (0,1]}\lV a_t e_{\bbB(t)}a_t\rV<\infty$.
\end{cor}
Furthermore, from Lemma \ref{lem:26} and Lemma \ref{lem:pud} , we obtain the following.
\begin{cor}\label{cor:lb}
Let $C_5$ be the constant given
in Lemma \ref{lem:pud}.
Then there exists an $l_5\in\nan$ such that
\begin{align}
\frac{C_5}{2}\Tr X^*X\le \lV \Gamma_{l,\bbB(t)}^{(R)}\circ \Theta_{l,t}\lmk X\rmk\rV^2,\quad
X\in\mnz\otimes \pu\mkk\pd,\quad l_5\le l,\quad t\in(0,1].
\end{align}
\end{cor}
For any $N,l\in\nan$ with $N\ge l\ge 2$ and $t\in(0,1]$, $Z,X_\mu\in \mnz\otimes\pu \mkk\pd$, $\mu=1,\ldots,n$,
we define 
\begin{align*}
&\caQ\lmk t, N,l,Z,\{X_\mu\}_{\mu}\rmk:=
\sum_{\mu^{(N-l+1)}\in\{1,\ldots,n\}^{\times N-l+1}}
\lv
\sum_{\mu=1}^n
\braket{\Gamma_{l-1,\bbB(t)}^{(R)}\lmk \Theta_{l,t}\lmk Z\rmk B_{\mu}\lmk t\rmk^*\rmk}
{\Gamma_{l-1,\bbB(t)}^{(R)}\lmk
\lmk \widehat{B_{\mu^{(N-l+1)}}(t)}\rmk^*\Theta_{N,t}\lmk X_\mu\rmk
\rmk}
\rv^2,\\
&\caR\lmk  t,Z,\{X_\mu\}_{\mu}\rmk:=
\sum_{\mu=1}^na_{t}X_{\mu}\rho_{\bbB(t)}B_{\mu}(t)Z^* a_t e_{\bbB(t)}\in \mnz\otimes \pu\mkk\pu.
\end{align*}
\begin{lem}\label{lem:qest}
There exist constants $C_{10}>0$, $0<s_3<1$ satisfying the followings.:
For any $N,l\in\nan$ with $N\ge l\ge 2$, $t\in(0,1]$,and $Z,X_\mu\in \mnz\otimes\pu\mkk\pd$, $\mu=1,\ldots,n$,
we have
\begin{align}\label{eq:qqnl}
\lv
\caQ\lmk t, N,l,Z,\{X_\mu\}_{\mu}\rmk-\varphi_{\bb(t)}
\lmk
\lmk \caR\lmk  t,Z,\{X_\mu\}_{\mu}\rmk\rmk^*
e_{\bbB(t)}
\caR\lmk  t,Z,\{X_\mu\}_{\mu}\rmk
\rmk
\rv
\le C_{10}\lmk s_3^{l}+s_3^{N-l}\rmk
\lmk\sum_{\mu=1}^n\lV X_\mu\rV^2\rmk
\lV Z\rV^2.
\end{align}
\end{lem}
\begin{proof}
The proof is analogous to that of Lemma \ref{lem:26}.
Let $0<s_1<1$ be the constant given in Lemma \ref{lem:s1}.
In this proof, we use a notation
\begin{small}
\begin{align*}
&\hat{\int_{N,l}}\\
&:= \frac{1}{\lmk 2\pi i\rmk^3}
\lmk\oint_{|z|=\frac{1+s_1}{2}} dzz^{N-l+1}+\oint_{|z-1|=\frac{1-s_1}{2}}dz\rmk
\lmk\oint_{|\zeta|=\frac{1+s_1}{2}} d\zeta \zeta^{l-1}+\oint_{|\zeta-1|=\frac{1-s_1}{2}}d\zeta\rmk
\lmk\oint_{|w|=\frac{1+s_1}{2}} dw w^{l-1}+\oint_{|w-1|=\frac{1-s_1}{2}}dw\rmk,
\end{align*}\end{small}
Note that 
$\Gamma_{l-1,\bbB(t)}^{(R)}\lmk \Theta_{l,t}\lmk Z\rmk B_{\mu}\lmk t\rmk^*\rmk$ and 
$\Gamma_{l-1,\bbB(t)}^{(R)}\lmk
\lmk \widehat{B_{\mu^{(N-l+1)}}(t)}\rmk^*\Theta_{N,t}\lmk X_\mu\rmk
\rmk$ 
are of the form
$r_{T_{\tilde \oo(t)}}^{-\frac l2}\times\text{a polynimial of }t$ and 
$r_{T_{\tilde \oo(t)}}^{-\frac N2}\times\text{a polynimial of }t$,
from the proof of Lemma \ref{lem:flim}.
Therefore, we have
\begin{align}\label{eq:pt27}
\lim_{t\downarrow 0}t\cdot \caQ\lmk t, N,l,Z,\{X_\mu\}_{\mu}\rmk
=0.
\end{align}
From (\ref{eq:tdecom}) and a routine calculation from \cite{Ogata1}, we obtain
\begin{small}
\begin{align}
&\caQ\lmk t, N,l,Z,\{X_\mu\}_{\mu}\rmk\nonumber\\
&=
\sum_{\mu,\mu'=1}^n
\sum_{\alpha,\beta=1}^{n_0}\sum_{i,j=-k_R}^0
\sum_{\alpha',\beta'=1}^{n_0}\sum_{i',j'=-k_R}^0\nonumber\\
&\braket{\cnz{\beta'}\otimes\fii{j'}}
{T_{\bbB(t)}^{N-l+1}\lmk
\begin{gathered}
T_{\bbB(t)}^{l-1}\lmk
\Theta_{N,t}\lmk X_{\mu'}\rmk^*\ket{\cnz{\beta'}\otimes\fii{j'}}\bra{\cnz{\alpha'}\otimes\fii{i'}} 
\Theta_{lt}\lmk Z\rmk {B_{\mu'}(t)}^*\rmk\\
\cdot \ket{\cnz{\alpha'}\otimes\fii{i'}}\bra{\cnz{\alpha}\otimes\fii{i}} \\
\cdot
T_{\bbB(t)}^{l-1}\lmk B_{\mu}(t)
{\Theta_{lt}\lmk Z\rmk}^*
\ket{\cnz{\alpha}\otimes\fii{i}}\bra{\cnz{\beta}\otimes\fii{j}}\Theta_{N,t}\lmk X_{\mu}\rmk
\rmk
\end{gathered}
\rmk\lmk
\cnz{\beta}\otimes\fii{j}
\rmk
}
\label{eq:te}\\
&=\widehat{\int_{N,l}}\sum_{\mu,\mu'=1}^n
\sum_{\alpha,\beta=1}^{n_0}\sum_{i,j=-k_R}^0
\sum_{\alpha',\beta'=1}^{n_0}\sum_{i',j'=-k_R}^0\nonumber\\
&\left\langle{\cnz{\beta'}\otimes\fii{j'}},\right.\nonumber\\
&\left.{
\lmk z-T_{\bb(t)}\rmk^{-1}\lmk
\begin{gathered}
\lmk \zeta-T_{\bb(t)}\rmk^{-1}\lmk
\Theta_{N,t}\lmk X_{\mu'}\rmk^*\ket{\cnz{\beta'}\otimes\fii{j'}}\bra{\cnz{\alpha'}\otimes\fii{i'}} 
\Theta_{lt}\lmk Z\rmk {B_{\mu'}(t)}^*\rmk\\
\cdot \ket{\cnz{\alpha'}\otimes\fii{i'}}\bra{\cnz{\alpha}\otimes\fii{i}} \\
\cdot
\lmk w-T_{\bb(t)}\rmk^{-1}\lmk B_{\mu}(t)
{\Theta_{lt}\lmk Z\rmk}^*
\ket{\cnz{\alpha}\otimes\fii{i}}\bra{\cnz{\beta}\otimes\fii{j}}\Theta_{N,t}\lmk X_{\mu}\rmk
\rmk
\end{gathered}
\rmk\lmk
\cnz{\beta}\otimes\fii{j}
\rmk}\right\rangle\label{eq:inte}
\end{align} 
\end{small}
for each $t\in(0,1]$.
Now, as in the proof of Lemma \ref{lem:26},
expanding  $\lmk z-T_{\bbB(t)}\rmk^{-1}$, $\lmk w-T_{\bbB(t)}\rmk^{-1}$, $\lmk \zeta-T_{\bbB(t)}\rmk^{-1}$,
$ B_\mu(t)$ we obtain an expansion
\begin{align}\label{eq:qesp}
\caQ\lmk t, N,l,Z,\{X_\mu\}_{\mu}\rmk=\sum_{j=1}^4t^{-j}b_j\lmk N, l,  Z,\{X_\mu\}_\mu\rmk +b_0\lmk t, N, l, Z,\{X_\mu\}_\mu\rmk.
\end{align}
By taking the limit $N\to\infty$ and then $l\to\infty$ after that for $b_0\lmk t, N, l, Z,\{X_\mu\}\rmk$
we obtain a function
$b_\infty\lmk t,Z,\{X_\mu\}\rmk
=\lim_{l\to\infty}\lmk \lim_{N\to\infty} b_0\lmk t, N, l, Z,\{X_\mu\}\rmk\rmk$.
There exist
 constants $C_{10}>0$ and $0<s_3<1$ such that 
\begin{align}\label{eq:bif7}
&\sup_{t\in(0,1]} \sup_{N,l\in\nan;N\ge l}\lv b_0\lmk t, N, l, Z,\{X_\mu\}\rmk
\rv,\;
\max_{j=1,\ldots, 4}\sup_{N,l\in\nan;N\ge l}\lv b_j\lmk N, l,  Z,\{X_\mu\}\rmk\rv,\;
\sup_{t\in(0,1]} \lv b_\infty\lmk t,Z,\{X_\mu\}\rmk
\rv\le C_{10} \sum_{\mu=1}^n\lV X_\mu\rV^2\cdot \lV Z\rV^2,\\
&\sup_{t\in(0,1]} \sup_{N,l\in\nan;N\ge l}
\lv
b_\infty\lmk t,Z,\{X_\mu\}\rmk
-b_0\lmk t, N, l, Z,\{X_\mu\}\rmk
\rv
\le C_{10}\lmk s_3^l+s_{3}^{N-l}\rmk \sum_{\mu=1}^n\lV X_\mu\rV^2\cdot \lV Z\rV^2,\label{eq:xlb27}
\end{align}	
for any $Z, X_\mu\in\mnz\otimes \pu\mkk\pd$.

The expansion (\ref{eq:qesp}) and the property (\ref{eq:pt27}) implies 
$ b_j\lmk N, l,  Z,\{X_\mu\}\rmk=0$ for any $j=1,\ldots,4$,
$N,l\in\nan$ with $N\ge l$, $t\in (0,1]$ and  $Z, X_\mu\in\mnz\otimes \pu\mkk\pd$.
Hence we obtain
\begin{align}\label{eq:lcn27}
\caQ\lmk t, N,l,Z,\{X_\mu\}_{\mu}\rmk=b_0\lmk t, N, l, Z,\{X_\mu\}\rmk,
\end{align} 
for any $N,l\in\nan$ with $N\ge l$, $t\in(0,1]$, $Z,X_\mu\in \mnz\otimes\pu\mkk\pd$, $\mu=1,\ldots,n$.
Take $N\to\infty$ limit and $l\to \infty$ then after in this equation.
From (\ref{eq:te}) combined with the spectral property of $T_{\bbB(t)}$,
we obtain
\begin{align}\label{eq:li27}
\varphi_{\bb(t)}
\lmk
\lmk \caR\lmk  t,Z,\{X_\mu\}_{\mu}\rmk\rmk^*
e_{\bbB(t)}
\caR\lmk  t,Z,\{X_\mu\}_{\mu}\rmk
\rmk
=b_\infty\lmk t,Z,\{X_\mu\}\rmk.
\end{align}
From (\ref{eq:xlb27}), (\ref{eq:lcn27}), and (\ref{eq:li27}),
we obtain (\ref{eq:qqnl}).
\end{proof}
\begin{lem}\label{lem:zp}
There exist constants $C_9>0$, $0<s_4<1$ satisfying the followings.:
For any $N,l\in\nan$ with $N\ge l\ge \max\{l_5,  n_0^6\lmk k_R+1\rmk\lmk k_L+1\rmk\}$, $t\in(0,1]$, $Z\in \mnz\otimes\pu\mkk\pd$,
$\Phi\in \lmk \caG_{N,\bbB(t)}\otimes \bbC^n\rmk \cap \caG_{N+1,\bbB(t)}^{\perp}$
we have
\begin{align}\label{eq:zieq}
\lV \lmk 
\unit_{N+1-l}\otimes \bra{\Gamma_{l,\bbB(t)}^{(R)}\circ\Theta_{l,t}\lmk Z\rmk}\rmk\Phi
\rV^2
\le 
C_9\lmk s_4^l+s_4^{N-l}\rmk
\lV \Phi\rV^2\lV  \Gamma_{l,\bbB(t)}^{(R)}\circ\Theta_{l,t}\lmk Z\rmk\rV^2.
\end{align}
(Recall $l_5$ given in Corollary \ref{cor:lb}.)
\end{lem}
\begin{proof} 
Let
$N,l\in\nan$ with $N\ge l\ge \max\{l_5,  n_0^6\lmk k_R+1\rmk\lmk k_L+1\rmk\}$, $t\in(0,1]$, $Z\in \mnz\otimes\pu\mkk\pd$, and
$\Phi\in \lmk \caG_{N,\bbB(t)}\otimes \bbC^n\rmk \cap \caG_{N+1,\bbB(t)}^{\perp}$.
As $\bbB(t)\in\Class (n,n_0,k_R,k_L)$, and $N\ge n_0^6\lmk k_R+1\rmk\lmk k_L+1\rmk\ge l_{\bbB(t)}$,
from  Proposition 3.1 of \cite{Ogata1}, for $\Phi\in  \caG_{N,\bbB(t)}\otimes \bbC^n$
there exist $\{\Phi_\mu\}_{\mu=1}^{n}\subset \mnz\otimes \pu\mkk\pd$
such that
$\Phi=\sum_{\mu=1}^n\Gamma_{N,\bbB(t)}^{(R)}\circ\Theta_{Nt}\lmk\Phi_\mu\rmk\otimes\psi_\mu$.
By Corollary \ref{cor:lb}, we have
\begin{align}\label{eq:phu}
\sum_{\mu=1}^{n}\lV\Phi_\mu\rV^2
\le
\frac{2}{C_5}\sum_{\mu=1}^{n} \lV \Gamma_{l,\bbB(t)}^{(R)}\circ \Theta_{l,t}\lmk \Phi_\mu\rmk\rV^2
=\frac{2}{C_5}\lV\Phi\rV^2.
\end{align}
We claim 
\begin{align}\label{eq:pbxa}
\lv
\sum_{\mu=1}^n \varphi_{\bbB(t)}\lmk B_{\mu}(t) X^* a_t e_{\bbB(t)}a_t\Phi_\mu\rmk
\rv
\le C_3\cdot\sqrt{\frac {2n}{C_5}}\cdot s_1^N\cdot\lV\Phi\rV\lV X\rV,\quad
X\in \mnz\otimes \pu\mkk\pd. 
\end{align}
To see this, let $X\in \mnz\otimes \pu\mkk\pd$.
From $\Phi\in \caG_{N+1,\bbB(t)}^{\perp}$, we have
\begin{align*}
0=\braket{\Gamma_{N+1,\bbB(t)}^{(R)}\circ \Theta_{N+1,t}
\lmk
X
\rmk}{\Phi}
=\sum_{\mu=1}^n\braket{\Gamma_{N,\bbB(t)}^{(R)}\lmk \Theta_{N+1,t}\lmk X\rmk B_{\mu}(t)^*  \rmk}
{ \Gamma_{N,\bbB(t)}^{(R)}\circ\Theta_{Nt}\lmk\Phi_\mu\rmk }.
\end{align*}
From this, 
we obtain
\begin{align*}
&\lv
\sum_{\mu=1}^n \varphi_{\bbB(t)}\lmk B_{\mu}(t) X^* a_t e_{\bbB(t)}a_t\Phi_\mu\rmk
\rv\\
&=
\lv
\sum_{\mu=1}^n \lmk
\varphi_{\bbB(t)}\lmk B_{\mu}(t) X^* a_t e_{\bbB(t)}a_t\Phi_\mu\rmk-\braket{\Gamma_{N,\bbB(t)}^{(R)}\lmk \Theta_{N+1,t}\lmk X\rmk B_{\mu}(t)^*  \rmk}
{ \Gamma_{N,\bbB(t)}^{(R)}\circ\Theta_{Nt}\lmk\Phi_\mu\rmk }
\rmk
\rv
\le
 C_3 s_2^N\lV X\rV \sum_{\mu=1}^n\lV \Phi_\mu\rV\\
&\le
C_3 s_2^N\lV X\rV\sqrt{\frac {2n}{C_5}}\lV\Phi\rV.
\end{align*}
In the first inequality, we used Lemma \ref{lem:26}.
The second inequality is due to (\ref{eq:phu}) and  Cauchy-Schwartz inequality.
This proves the claim (\ref{eq:pbxa}).

From Corollary \ref{lem:aea}, 
\[
C^{'}:=\sup_{t\in(0,1]}\lmk
\lV e_{\bbB(t)}\rV\lV a_te_{\bbB(t)}a_t\rV \lmk \sum_{\mu=1}^n \lV B_\mu(t)\rV^2\rmk^{\frac 12}
\rmk,\quad
C^{''}:=\frac {2\sqrt{n}}{C_5}C_3C'
\]
are finite positive constants.
We apply (\ref{eq:pbxa}) to 
\[
S_t:=\sum_{\mu'=1}^n\Phi_{\mu'}\rho_{\bbB(t)}B_{\mu'}(t)Z^*a_te_{\bbB(t)}\rho_{\bbB(t)}e_{\bbB(t)}a_tZ\in \mnz\otimes \pu\mkk\pd.
\]
Note that, by $\Tr\rho_{\bbB(t)}=1$ and Cauchy-Schwartz inequality and (\ref{eq:phu}), we have
\[
\lV S_t\rV\le C'\cdot \lV Z\rV^2\lmk \sum_{\mu=1}^n \lV\Phi_\mu\rV^2\rmk^{\frac 12}
\le
 C'\cdot \lV Z\rV^2\sqrt{\frac2{C_5}}\lV \Phi\rV.
\]
Applying  (\ref{eq:pbxa}) to $S_t$, we obtain 
\begin{align}\label{eq:asa}
&\lv
\varphi_{\bb(t)}
\lmk
\lmk \caR\lmk  t,Z,\{\Phi_\mu\}_{\mu}\rmk\rmk^*
e_{\bbB(t)}
\caR\lmk  t,Z,\{\Phi_\mu\}_{\mu}\rmk
\rmk
\rv
=\lv\sum_{\mu=1}^n
\varphi_{\bb(t)}
\lmk
B_{\mu}(t) S_t^*a_t e_{\bbB(t)}a_t\Phi_\mu
\rmk
\rv\nonumber\\
&\le  C_3\cdot\sqrt{\frac {2n}{C_5}}\cdot s_1^N\cdot\lV\Phi\rV\lV S_t\rV
\le
C^{''}s_1^N\lV Z\rV^2\lV \Phi\rV^2.
\end{align}
Now we are ready to show (\ref{eq:zieq}).
It is straightforward to check that
$\lV \lmk 
\unit_{N+1-l}\otimes \bra{\Gamma_{l,\bbB(t)}^{(R)}\circ\Theta_{l,t}\lmk Z\rmk}\rmk\Phi
\rV^2$ is equal to $\caQ\lmk t, N,l,Z,\{\Phi_\mu\}_{\mu}\rmk$.
From this, applying Lemma \ref{lem:qest} and (\ref{eq:asa}),
\begin{align*}
&\lV \lmk 
\unit_{N+1-l}\otimes \bra{\Gamma_{l,\bbB(t)}^{(R)}\circ\Theta_{l,t}\lmk Z\rmk}\rmk\Phi
\rV^2=\caQ\lmk t, N,l,Z,\{\Phi_\mu\}_{\mu}\rmk\\
&\le
\lv
\caQ\lmk t, N,l,Z,\{\Phi_\mu\}_{\mu}\rmk-
\varphi_{\bb(t)}
\lmk
\lmk \caR\lmk  t,Z,\{\Phi_\mu\}_{\mu}\rmk\rmk^*
e_{\bbB(t)}
\caR\lmk  t,Z,\{\Phi_\mu\}_{\mu}\rmk
\rmk
\rv
+
\varphi_{\bb(t)}
\lmk
\lmk \caR\lmk  t,Z,\{\Phi_\mu\}_{\mu}\rmk\rmk^*
e_{\bbB(t)}
\caR\lmk  t,Z,\{\Phi_\mu\}_{\mu}\rmk
\rmk\\
&\le C_{10}\lmk s_3^{l}+s_3^{N-l}\rmk
\lmk\sum_{\mu=1}^n\lV \Phi_\mu\rV^2\rmk
\lV Z\rV^2
+C^{''}s_1^N\lV Z\rV^2\lV \Phi\rV^2
\le C_{10}\lmk s_3^{l}+s_3^{N-l}\rmk
\frac{2}{C_5}\lV\Phi\rV^2
\lV Z\rV^2
+C^{''}s_2^N\lV Z\rV^2\lV \Phi\rV^2.
\end{align*}
In the last inequality, we used (\ref{eq:phu}).
This and Corollary \ref{cor:lb} proves the claim of Lemma.
\end{proof}
\begin{proofof}[Lemma \ref{lem:overl}]
Let
$C_9$, $0<s_4<1$ be the constants in Lemma \ref{lem:zp}, and $l_5\in\nan$ in Corollary \ref{cor:lb}.
Let $N,l\in\nan$ such that $N\ge l\ge \max\{l_5,  n_0^6\lmk k_R+1\rmk\lmk k_L+1\rmk\}$.

For each $t\in(0,1]$, we have $\bbB(t)\in\Class(n,n_0,k_R,k_L)$ and $l\ge n_0^6\lmk k_R+1\rmk\lmk k_L+1\rmk\ge l_{\bbB(t)}$.
Therefore, by Proposition 3.1 of \cite{Ogata1}, there exist $Z_{t,i}\in\mnz\otimes\pu\mkk\pd$, $i=1,\ldots,n_0^2\lmk k_R+1\rmk\lmk k_L+1\rmk$ such that
$\{\Gamma_{l,\bbB(t)}^{(R)}\circ\Theta_{l,t}\lmk Z_{t,i}\rmk\}_{i=1}^{n_0^2\lmk k_R+1\rmk\lmk k_L+1\rmk}$
is a CONS of $\caG_{l,\bbB(t)}$.

Let $\xi\in\bigotimes_{i=0}^N\bbC^n$ and $t\in(0,1]$.
Then $\Phi_t:=\lmk G_{N,\bbB(t)}\otimes\unit_{\{N\}}-G_{N+1,\bbB(t)}\rmk \xi\in \lmk \caG_{N,\bbB(t)}\otimes \bbC^n\rmk \cap \caG_{N+1,\bbB(t)}^{\perp}$.
From Lemma \ref{lem:zp}, we have
\begin{align*}
&\lV\lmk \unit_{[0,N-l]}\otimes G_{l,\bbB(t)}\rmk\lmk G_{N,\bbB(t)}\otimes\unit_{\{N\}}-G_{N+1,\bbB(t)}\rmk\xi\rV^2
= \sum_{i=1}^{n_0^2\lmk k_R+1\rmk\lmk k_L+1\rmk}
\lV
\lmk \unit_{[0,N-l]}\otimes
\bra{\Gamma_{l,\bbB(t)}^{(R)}\circ\Theta_{l,t}\lmk Z_{t,i}\rmk}\rmk
\Phi_t
\rV^2\\
&\le C_9\lmk s_4^l+s_4^{N-l}\rmk\sum_{i=1}^{n_0^2\lmk k_R+1\rmk\lmk k_L+1\rmk}
\lV \Phi_t\rV^2\lV  \Gamma_{l,\bbB(t)}^{(R)}\circ\Theta_{l,t}\lmk Z_{t,i}\rmk\rV^2
\le  n_0^2\lmk k_R+1\rmk\lmk k_L+1\rmk C_9\lmk s_4^l+s_4^{N-l}\rmk\lV\xi\rV^2.
\end{align*}
Hence we obtain
\begin{align*}
&\sup_{t\in[0,1]}\lV \lmk \unit_{[0,N-l]}\otimes \hat G_{l}(t)\rmk\lmk \hat G_{N}(t)\otimes\unit_{\{N\}}-\hat G_{N+1}(t)\rmk \rV
=\sup_{t\in(0,1]}\lV \lmk \unit_{[0,N-l]}\otimes G_{l,\bbB(t)}\rmk\lmk G_{N,\bbB(t)}\otimes\unit_{\{N\}}-G_{N+1,\bbB(t)}\rmk \rV\\
&\le n_0^2\lmk k_R+1\rmk\lmk k_L+1\rmk C_9\lmk s_4^l+s_4^{N-l}\rmk.
\end{align*}
For the first equality, we used the continuity of $\hat G_l(t)$.
\end{proofof}
\begin{lem}\label{lem:c5c}
We define $l_6\in\nan$ by
 \[
l_6:=\inf\{ l'\in\nan\mid l'\ge l_5,\; n_0^2\lmk k_R+1\rmk\lmk k_L+1\rmk C_9s_4^l<\frac{1}{8\sqrt l},\quad
\text{for all}\quad l\ge l'\}.
\]
The paths $\{\hat \xi_{\alpha,\beta,a,b}^{(l)}\}$ given in  Lemma \ref{lem:gtz}
satisfies the {\it Condition 5} with respect to $(\max\{ l_3, n_0^6(k_R+1)(k_L+1)\},\max\{ l_4,  2n_0^6(k_R+1)(k_L+1)\}, l_6)$.
(Recall that $C_9$, $s_4$ are given in Lemma \ref{lem:overl}, and 
$l_3,l_4,l_5$ by Lemma \ref{lem:dfsn}, \ref{lem:intg}, \ref{lem:overl}, respectively.)

\end{lem}
\begin{proof}
(i) of {\it Condition 5} follows from Lemma \ref{lem:gtz}.
(ii) of  {\it Condition 5}  for $l'=\max\{ l_3, n_0^2(k_R+1)(k_L+1)\}$ follows from
Lemma \ref{lem:dfsn} and Proposition 3.1 of \cite{Ogata1} with $l_{\bb(t)}\le n_0^6(k_R+1)(k_L+1)$.
(iii) of  {\it Condition 5}  for $m'=\max\{ l_4,2 n_0^6(k_R+1)(k_L+1)\}$ follows from
Lemma \ref{lem:intg} and Proposition 3.1 of \cite{Ogata1} with $l_{\bb(t)}\le n_0^6(k_R+1)(k_L+1)$.
(iv) of  {\it Condition 5}  with $\varepsilon_l:=4n_0^2\lmk k_R+1\rmk\lmk k_L+1\rmk C_9s_4^l$ and $l^{''}=l_6$
is from Lemma \ref{lem:overl}.
\end{proof}

\subsection{Proof of Theorem \ref{singthm}}
In this subsection, we complete the proof of Theorem \ref{singthm}.
We use the following notations.
\begin{nota}
If $2\le n_0\in\nan$ and $k_R,k_L\in\nan\cup\{0\}$,
we denote the path $\bbB(t)\in\Class(n,n_0,k_R,k_L)$
given in (\ref{eq:pbp})
by $\bbB_{s, n_0,k_R,k_L}(t)$.
We denote $\bbA$ given in (\ref{eq:defa}) by $\bbA_{s,n_0,k_R,k_L}$,
and $\bbV\in \Class(n,1,n_0(k_R+1)-1, n_0(k_L+1)-1)$ given in Notation \ref{nota:bbk} associated to $\bbA_{s,n_0,k_R,k_L}$ by $\bbV_{s,n_0,k_R,k_L}$.
Furthermore, we denote the path of the orthogonal projections $\hat G_l(t)$ given in Lemma \ref{lem:gtz}
by $\hat G_{n_0,k_R,k_L,l}(t)$.
We also set
\begin{align}\label{eq:mti}
\tilde m(n_0,k_R,k_L)
:=\max\left\{
2n_0^6\lmk k_R+1\rmk\lmk k_L+1\rmk,
2l_{\bbA_{s,n_0,k_R,k_L}}, 2l_{\bbV_{s,n_0,k_R,k_L}},
l_{4, n_0,k_R,k_L }, l_{3, n_0,k_R,k_L }
\right\},
\end{align}
where $ l_{\bbA_{s,n_0,k_R,k_L}}$,
$l_{4, n_0,k_R,k_L }$, $l_{3, n_0,k_R,k_L}$
are the numbers 
$l_\bbA$, $l_4$, $l_3$ given in Lemma \ref{lem:asp},
Lemma \ref{lem:intg}, Lemma \ref{lem:dfsn}, respectively.
The number 
$ l_{\bbV_{s,n_0,k_R,k_L}}$ is the one defined in Definition 1.13 \cite{Ogata1}
for $\bbV_{s,n_0,k_R,k_L}\in \Class A$.
For $n_0=1$ and each $k_R,k_L\in\nan\cup\{0\}$, we fix some $\bbV_{s,1,k_R,k_L}\in\Class(n,1,k_R, k_L) $
and set
$\bbB_{s, 1,k_R,k_L}(t):=\bbV_{s,1,k_R,k_L}$, $t\in[0,1]$.
We also set $\tilde m(1,k_R,k_L):=2(k_R+1)(k_L+1)$.
\end{nota}
\begin{proofof}[Theorem \ref{singthm}]
Let $n_0, n_0'\in\nan$, $k_R,k_L,k_R',k_L'\in\nan\cup\{0\}$,
 $\bb\in\Class(n,n_0,k_R,k_L)$, and
$\bb'\in\Class(n,n_0',k_R',k_L')$.
Let $m\ge 2l_\bbB$, and
$m'\ge 2l_{\bbB'}$.
From \cite{Bachmann:2011kw} and Theorem 1.18 (vi) of \cite{Ogata1}, 
 if
$H_{\Phi_{m,\bb}}\simeq_{II} H_{\Phi_{m',\bb'}}$,
 we have
$
n_0(k_R+1)=n_0'(k_R'+1),\;\text{and}\;
n_0(k_L+1)=n_0'(k_L'+1)
$.

Now we show the converse. Set
$m_0:=\max\{m,m',\tilde m(n_0,k_R,k_L),\tilde m(n_0',k_R',k_L')\}$.We consider the following sequence of paths.
\begin{enumerate}
\item By the same argument as Lemma 4.2 of \cite{bo}, 
by considering a path $\Phi(X;t):=(1-t)\Phi_{m,\bbB}(X)+t\Phi_{m_0,\bbB}(X)$, we obtain
$H_{\Phi_{m,\bbB}}\simeq_{I}H_{\Phi_{m_0,\bbB}}$.
For $N$ large enough, the kernel of 
$\lmk H_{\Phi_{m,\bbB}}\rmk_{[0,N-1]}$ and $\lmk H_{\Phi_{m_0,\bbB}}\rmk_{[0,N-1]}$ coincide.
From this, Lemma \ref{lem:24}, and Theorem 1.18 of \cite{Ogata1}
we see that $(\Phi_{m,\bbB}, {\Phi_{m_0,\bbB}})$ satisfy {\it Condition 8}.
\item
As $\bbB, \bbB_{s, n_0,k_R,k_L}(1)\in \Class(n,n_0,k_R,k_L)$, 
we have $H_{\Phi_{m_0,\bbB}}\simeq_{I}H_{\Phi_{m_0,\bbB_{s, n_0,k_R,k_L}(1)}}$ from Proposition \ref{classthm}. 
By the proof of  Proposition \ref{classthm}, the path connecting $\Phi_{m_0,\bbB}$ and 
$\Phi_{m_0,\bbB_{s, n_0,k_R,k_L}(1)}$ is given by $\Phi_{m_0,\bbB(t)}$ with $\bbB(t)\in  \Class(n,n_0,k_R,k_L)$, a path 
of positive frustration free translation invariant finite range interactions satisfying {\it Condition 6}.
\item (If $n_0\ge 2$.)
By Lemma \ref{lem:c5c}, 
the paths $\{\hat \xi_{\alpha,\beta,a,b}^{(l)}\}$ given by  Lemma \ref{lem:gtz}
satisfies the {\it Condition 5} with $m'=\max\{l_{4,n_0,k_R,k_L}, 2n_0^6(k_R+1)(k_L+1)\}\le m_0$. From Lemma \ref{lem:vpath}, this implies
$H_{\Phi_{m_0,\bbB_{s, n_0,k_R,k_L}(1)}}=H_{\Phi_{\overline{\hat G_{n_0,k_R,k_L,m_0}(1)}}}\simeq_{I}H_{\Phi_{\overline{\hat G_{n_0,k_R,k_L,m_0}(0)}}}$.
The path connecting them is given by $\Phi(t)=\Phi_{\overline{\hat G_{n_0,k_R,k_L,m_0}(t)}}$, 
a path of positive frustration free translation invariant finite range interactions satisfying {\it Condition 6}.
\item  (If $n_0\ge 2$.) By Lemma \ref{lem:gaeq}, 
we have $H_{\Phi_{\overline{\hat G_{n_0,k_R,k_L,m_0}(0)}}}\simeq_{II}H_{\Phi_{m_0,\bbA_{s,n_0,k_R,k_L}}}$
via the path $\Phi(X;t):=(1-t)\Phi_{\overline{\hat G_{n_0,k_R,k_L,m_0}(0)}}(X)+t{\Phi_{m_0,\bbA_{s,n_0,k_R,k_L}}}(X)$. From Lemma \ref{lem:ag7}, the pair of the interactions $(\Phi_{\overline{\hat G_{n_0,k_R,k_L,m_0}(0)}}, {\Phi_{m_0,\bbA_{s,n_0,k_R,k_L}}})$ satisfies {\it Condition 8}.
\item (If $n_0\ge 2$.)
By Notation \ref{nota:bbk},  we have $H_{\Phi_{m_0,\bbA_{s,n_0,k_R,k_L}}}\simeq_{II}H_{\Phi_{m_0,\bbV_{s,n_0,k_R,k_L}}}$,
via the path $\Phi(X;t):=(1-t){\Phi_{m_0,\bbA_{s,n_0,k_R,k_L}}}(X)+ t{\Phi_{m_0,\bbV_{s,n_0,k_R,k_L}}}(X)$.
From Lemma \ref{lem:av7}, the pair of the interactions $({\Phi_{m_0,\bbA_{s,n_0,k_R,k_L}}}, {\Phi_{m_0,\bbV_{s,n_0,k_R,k_L}}})$ satisfies {\it Condition 8}.
\end{enumerate}
Hence we obtain
$H_{\Phi_{m,\bbB}}\simeq_{II}H_{\Phi_{m_0,\bbV_{s,n_0,k_R,k_L}}}$. Recall that we set
$\bbB_{s, 1,k_R,k_L}(1):=\bbV_{s,1,k_R,k_L}$, for $n_0=1$. Parallel argument implies
$H_{\Phi_{m',\bbB'}}\simeq_{II}H_{\Phi_{m_0,\bbV_{s,n_0',k_R',k_L'}}}$.
If
$
n_0(k_R+1)=n_0'(k_R'+1),\;\text{and}\;
n_0(k_L+1)=n_0'(k_L'+1)
$,
then
we have $\bbV_{s,n_0,k_R,k_L}, \bbV_{s,n_0',k_R',k_L'}\in \Class(n,1,n_0(k_R+1)-1,n_0(k_L+1)-1)=
\Class(n,1,n_0'(k_R'+1)-1,n_0'(k_L'+1)-1)$.
Therefore we have $H_{\Phi_{m_0,\bbV_{s,n_0,k_R,k_L}}}\simeq_{I}H_{\Phi_{m_0,\bbV_{s,n_0',k_R',k_L'}}}$ from Proposition \ref{classthm}. 
This completes the proof.
\end{proofof}

\section{Bulk classification}
In this section we prove Theorem \ref{thm:bulk}.
First, recall that gap of local Hamiltonians imply that of bulk Hamiltonians.
\begin{lem}\label{lem:obb}
Let $\Phi$ be a translation invariant finite range interaction on $\caA_{\bbZ}$.
Assume that the Hamiltonian $H_\Phi$ is gapped with respect to the open boundary condition.
Let $\gamma>0$ be 
a lower bound of the gap.
Then for any $\omega\in\caS_{\bbZ}(H_{\Phi})$, we have
\[
\sigma\lmk H_{\omega,\alpha_\Phi}\rmk\setminus \{0\}\subset [\gamma,\infty).
\]
\end{lem}
\begin{proof}
The proof is the same as that of \cite{KN97}, and \cite{ss03}.
\end{proof}
As an immediate consequence of Lemma \ref{lem:obb}, we obtain the following.
\begin{lem}\label{lem:tl}
Let $H_0,H_1$ be Hamiltonians associated with interactions
$\Phi_{0},\Phi_{1}\in{\caJ}$ that are gapped with respect to the open boundary conditions.
Assume that $H_0$ and $H_1$ are type I-$C^1$-equivalent. Let
 $m\in\nan$,  $\Phi:[0,1]\to {\caJ}_m$, and $\gamma>0$ be the corresponding path of
interactions and the uniform lower bound of the gap in Definition \ref{def:phafst}.
Then, for any $t\in[0,1]$ and $\omega_t\in\caS_{\bbZ}\lmk H_{\Phi(t)}\rmk$,
we have
\[
\sigma\lmk H_{\omega_t,\alpha_{\Phi(t)}}\rmk \setminus \{0\}\subset [\gamma,\infty).
\]
\end{lem}

\subsection{Frustration free Hamiltonians and their bulk ground states}
To check that some interaction $\Phi$ belongs to $\caJ_B$, we have to consider not only elements in
$\caS_{\bbZ}(H_{\Phi})$ but all elements in 
${\mathfrak B}_{\Phi}$. 
For the interactions that is considered in this paper, these set coincides.
Recall Notation \ref{nota:ts}.
The proof of the following Lemma is exactly the same as that of  Lemma 3.10 \cite{Ogata1}.
\begin{lem}\label{lem:tss}
Let $H_\Phi$ be a frustration free Hamiltonian given by a positive interaction 
$\Phi\in \caJ$.
Then we have $\tilde\caS_{\bbZ,\bbZ}\lmk H_\Phi\rmk=\caS_\bbZ\lmk H_\Phi\rmk$.
If $\omega$ is a unique element in $\caS_\bbZ\lmk H_\Phi\rmk$, 
it is pure.
\end{lem}
\begin{lem}\label{lem:kh}
Let $H_\Phi$ be a frustration free Hamiltonian 
given by a positive interaction $\Phi\in \caJ$.
Let $\omega$ be an element in $\caS_{\bbZ}(H_\Phi)$ and $(\caH,\pi,\Omega)$
its GNS triple.
Then for any unit vector $\xi\in\ker H_{\omega,\alpha_\Phi}$,
$\omega_\xi:=\braket{\xi}{\pi\lmk\cdot\rmk\xi}$ defines a state in 
 $\caS_{\bbZ}(H_\Phi)$.
\end{lem}
\begin{proof}
We first claim that
\begin{align}\label{eq:ihfh}
\lV
\lmk \pi\lmk \lmk H_{\Phi}\rmk_{[a,b]} \rmk\rmk^\frac 12 \eta
\rV
\le
\lV H_{\omega,\alpha_\Phi}^\frac12 \eta\rV,
\end{align}
for all $\eta$ in the domain of  $(H_{\omega,\alpha_\Phi}^\frac12)$ and
$a\le b$.
To see this, we consider arbitrary $a\le b$, finite interval $I=[c,d]$ with 
$c\le a$, $b\le d$, and $A\in\caA_I$.
Let $R\in\nan$ be the interaction length of $\Phi$.
Then we have 
$\delta_{\Phi}(A)=i\left[ \lmk H_{\Phi}\rmk_{ [c-R, d+R]}, A\right]$, by the definition of $\alpha_\Phi$.
From this , we have
\begin{align*}
&\braket{\pi(A)\Omega}{H_{\omega,\alpha_\Phi}\pi(A)\Omega}
=-i\omega\lmk A^*\delta_{\Phi}(A)\rmk=
\omega\lmk A^* \lmk H_{\Phi}\rmk_{[c-R, d+R]} A\rmk
- \omega\lmk A^*A\lmk \lmk H_{\Phi}\rmk_{[c-R, d+R]} \rmk\rmk\\
&=\omega\lmk A^* \lmk H_{\Phi}\rmk_{[c-R, d+R]} A\rmk
\ge \omega\lmk A^* \lmk H_{\Phi}\rmk_{[a,b]} A\rmk.
\end{align*}
The third equality is due to the frustration freeness of $\Phi$ and that $\omega\in\caS_{\bbZ}(H_\Phi)=\widetilde{\caS_{\bbZ,\bbZ}}(H_\Phi)$,
due to Lemma \ref{lem:tss}.
The inequality is by the positivity of $\Phi$.
As $\pi \lmk \caA_{\bbZ}^{\rm loc}\rmk\Omega$ is a core of $H_{\omega,\alpha_\Phi}$, this prove the claim.

By (\ref{eq:ihfh}), for any
unit vector $\xi\in\ker H_{\omega,\alpha_\Phi}$, and fintie interval $[a,b]$, 
we have $\pi\lmk \lmk H_{\Phi}\rmk_{[a,b]} \rmk \xi=0$. This means
$\omega_\xi\in \widetilde{\caS_{\bbZ,\bbZ}}(H_\Phi)=\caS_{\bbZ}(H_\Phi)$.
\end{proof}
\begin{lem}\label{lem:zee}Let $\Phi_0,\Phi_1\in \caJ$ be positive interactions.
Assume that the pair $(\Phi_0,\Phi_1)$ satisfies the {\it Condition 7}.
Let $\omega_\infty$be the state in {\it 4} of {\it Condition 7}.
For each $s\in[0,1]$ we define $\Phi(s)\in\caJ$ by
\begin{align}\label{eq:int10}
\Phi(X; s):=(1-s)\Phi_{0}(X)+s\Phi_{1}(X),\quad X\in {\mathfrak S}_{\bbZ}.
\end{align}
Then, for each $s\in(0,1)$, $\omega_\infty$ is the unique $\alpha_{\Phi(s)}$-ground state.
\end{lem}
\begin{proof}
{\it Step 1.} 
Note that for any $s\in(0,1)$, we have $\omega_{\infty}\in 
\widetilde{\caS_{\bbZ,\bbZ}}\lmk H_{\Phi(s)}\rmk$. This is because 
\[
\omega_\infty\lmk \Phi(X;s)\rmk=(1-s)\omega_\infty\lmk \Phi_0(X)\rmk+ s\omega_\infty\lmk \Phi_1(X)\rmk
=0,\quad X\in {\mathfrak S}_{\bbZ}
\]
due to the frustration freeness of $H_{\Phi_0}$, $H_{\Phi_1}$ and Lemma \ref{lem:tss}.
Therefore, $H_{\Phi(s)}$ is frustration free and $\omega_\infty\in \caS_{\bbZ}(H_{\Phi(s)})$ for all $s\in(0,1)$. In particular,
$\omega_\infty$ is an $\alpha_{\Phi(s)}$-ground state on $\caA_{\bbZ}$ for all $s\in(0,1)$.
Furthermore, $\omega_\infty$ is pure, because it is the unique element in $\caS_{\bbZ}\lmk H_{\Phi_0}\rmk$, 
(Lemma \ref{lem:tss}).
\\\\
{\it Step 2.} We claim that for any 
$s\in(0,1)$, any pure $\alpha_{\Phi(s)}$-ground state $\omega$ on $\caA_{\bbZ}$
is quasi-equivalent to $\omega_\infty$.
Let $s\in(0,1)$ and $\omega$ be an arbitrary pure $\alpha_{\Phi(s)}$-ground state on $\caA_{\bbZ}$ and 
$(\caH,\pi,\Omega)$ the GNS triple of $\omega$. 

As $\Phi_0,\Phi_1$ are of finite range, there exists an $R>0$ such that $\Phi_0(X)=\Phi_1(X)=0$ 
if the diameter of $X$ is larger than $R$. Let $\gamma>0$ and $N_0\in\nan$ be numbers satisfying
$\gamma\le \inf_{N\ge N_0}\inf \lmk \sigma\lmk \lmk H_{\Phi_1}\rmk_{[0,N-1]}\rmk\setminus\{0\}\rmk$, 
from {\it 2.} of {\it Condition 7}.

First we claim
\begin{align}\label{eq:acp}
\sum_{X\in{\mathfrak S}_{\bbZ}}\omega\lmk \Phi(X;s)\rmk<\infty.
\end{align}
To see this, recall the characterization of the ground state in Thorem 6.2.52 of \cite{BR2}.
For each $M\in\nan$, we choose a state $\sigma_M$ on $\caA_{[-M,M]}$ whose support is in the kernel
of $\lmk H_{\Phi_1}\rmk_{[-M,M]}$. We consider the state 
$\omega_M':=\sigma_M\otimes\omega\vert_{\caA_{\bbZ\setminus [-M,M]}}$ on $\caA_{\bbZ}$.
Note that $\omega_M'$ and $\omega$ coincide on $\caA_{\bbZ \setminus [-M,M]}$.
By Thorem 6.2.52 of \cite{BR2}, as $\omega$ is an $\alpha_{\Phi(s)}$-ground state, we obtain
\begin{align}\label{eq:m01b}
&0\le \sum_{X:X\cap[-M,M]\neq\emptyset }\omega\lmk \Phi(X;s)\rmk\le
\sum_{X:X\cap[-M,M]\neq\emptyset }\omega_M'\lmk \Phi(X;s)\rmk=
\sum_{X:X\cap[-M,M]\neq\emptyset }\lmk(1-s)\omega_M'\lmk \Phi_0(X)\rmk+s\omega_M'\lmk \Phi_1(X)\rmk\rmk
\nonumber\\
&=
(1-s)\sigma_M\lmk\lmk H_{\Phi_0}\rmk_{[-M,M]}\rmk
+s\sigma_M\lmk\lmk H_{\Phi_1}\rmk_{[-M,M]}\rmk+\sum_{\stackrel{X\cap[-M,M]\neq\emptyset}{X\cap[-M,M]^c\neq\emptyset} }\omega_M'\lmk \Phi\lmk X;s\rmk\rmk.
\end{align}
As the support of $\sigma_M$ is in the kernel
of $\lmk H_{\Phi_1}\rmk_{[-M,M]}$, we have $\sigma_M\lmk\lmk H_{\Phi_1}\rmk_{[-M,M]}\rmk=0$.
In order to estimate $\sigma_M\lmk\lmk H_{\Phi_0}\rmk_{[-M,M]}\rmk$,
set $C_1:=\sum_{X: 0 \in X
}\lV \Phi_1(X)\rV$, and note that
$\lV \lmk H_{\Phi_0}\rmk_{[-M,M]}\rV\le C_1(2M+1)$.
We then obtain
\begin{align*}
&\sigma_M\lmk\lmk H_{\Phi_0}\rmk_{[-M,M]}\rmk
=\sigma_M\lmk\lmk H_{\Phi_0}\rmk_{[-M,M]} \lmk G_{[-M,M],1}-G_{[-M,M],0}\rmk G_{[-M,M],1}\rmk
+\sigma_M\lmk\lmk H_{\Phi_0}\rmk_{[-M,M]}G_{[-M,M],0}G_{[-M,M],1}\rmk\\
&=\sigma_M\lmk\lmk H_{\Phi_0}\rmk_{[-M,M]} \lmk G_{[-M,M],1}-G_{[-M,M],0}\rmk G_{[-M,M],1}\rmk\\
&\le\lV\lmk G_{[-M,M],1}-G_{[-M,M],0}\rmk G_{[-M,M],1}\rV
\lV  \lmk H_{\Phi_0}\rmk_{[-M,M]}\rV
\le C_1\cdot C r^{2M+1}(2M+1).
\end{align*}
Here, $C$, $r$ are constants given in {\it 3.} of {\it Condition 7}.
Hence $\sigma_M\lmk\lmk H_{\Phi_0}\rmk_{[-M,M]}\rmk$ converges to $0$ as $M\to\infty$.

By the translation invariance of $\Phi_0,\Phi_1$ and the definition of $R$,
for any $M\ge R$, we have
\[
\sum_{\stackrel{X\cap[-M,M]\neq\emptyset}{X\cap[-M,M]^c\neq\emptyset} }
(1-s)\lV\Phi_0(X)\rV+s\lV\Phi_1(X)\rV
=\sum_{\stackrel{X\cap[-R,R]\neq\emptyset}{X\cap[-R,R]^c\neq\emptyset} }
(1-s)\lV\Phi_0(X)\rV+s\lV\Phi_1(X)\rV=:C_2.
\]
Substitute this and $\lim_M\sigma_M\lmk\lmk H_{\Phi_0}\rmk_{[-M,M]}\rmk=0$, $\sigma_M\lmk\lmk H_{\Phi_1}\rmk_{[-M,M]}\rmk=0$ to (\ref{eq:m01b}).
Taking $M\to\infty$ limit, we obtain 
$0\le \sum_{X\in{\mathfrak S}_{\bbZ} }\omega\lmk \Phi(X;s)\rmk\le
C_2<\infty$, proving the claim (\ref{eq:acp}).

From this, there exists an $M_0\in\nan$ with $R\le M_0$ such that 
\begin{align}
0\le s\sum_{X\subset[-M_0,M_0]^c }\omega\lmk \Phi_1(X)\rmk\le \sum_{X\subset[-M_0,M_0]^c }\omega\lmk \Phi(X;s)\rmk<\frac {s\gamma}{4}.
\end{align}
By the spectral gap condition for $\Phi_1$, for any $N\in\nan$ with $M_0+N_0\le N$, we have
\begin{align}\label{eq:og}
&\lV \lmk 1-G_{[-N,N]\setminus [-M_0,M_0],1}\rmk\Omega\rV^2
\le \gamma^{-1}\sum_{X\subset [-M_0,M_0]^c }\omega\lmk \Phi_1(X)\rmk<\frac 14.
\end{align}

Now we set for each $N_0+M_0\le N$, 
a vector $\xi_N:=G_{[-N,N]\setminus [-M_0,M_0],1}\Omega$.
The norm of these vectors are decreasing by the definition. From (\ref{eq:og}), we have 
$\lV \xi_N\rV\ge \frac 12$ for each $N_0+M_0\le N$. Therefore $\lV \xi_N\rV$ converges to a strictly positive constant.
Furthermore, for $N'\le N$, by the definition, we have
\begin{align*}
\lV \xi_N-\xi_{N'}\rV^2
=\lV \xi_N\rV^2+\lV \xi_{N'}\rV^2-\braket{\xi_N}{\xi_{N'}
}-\braket{\xi_{N'}}{\xi_{N}}
=-\lV \xi_N\rV^2+\lV \xi_{N'}\rV^2.
\end{align*}
This and the fact that  $\lV \xi_N\rV$ converges means that $\{\xi_N\}_N$ is a Cauchy sequence.
Therefore, it has a limit $\tilde \xi\in\caH$. From the bound $\lV \xi_N\rV\ge \frac 12$,
$\tilde \xi$ is non-zero, and we define a vector 
$\xi:=\lV \tilde \xi\rV^{-1}\tilde \xi$ in $\caH$ and an $\omega$-normal state
$\omega_{\xi}:=\braket{\xi}{\pi\lmk \cdot \rmk\xi}$.
As $\omega$ is pure,  $\omega_{\xi}$ and $\omega$ are quasi-equivalent.
From the definition of this $\tilde \xi$, for any $X\in{\mathfrak S}_{\bbZ}$, with
$X\cap[-M_0,M_0]=\emptyset$, we have $\omega_\xi\lmk \Phi_1(X)\rmk=0$.
Hence $\omega_{\xi}$ belongs to $\widetilde{\caS_{\bbZ,[-M_0,M_0]^c}}\lmk H_{\Phi_1}\rmk$.
By {\it 5.} of the assumption, this implies that $\omega_\xi$ and $\omega_\infty$ are quasi-equivalent.
Combined with the quasi-equivalence of $\omega$ and $\omega_\xi$, we conclude that
$\omega$ and $\omega_\infty$ are quasi-equivalent, proving the claim of {\it Step 2}.
\\\\
{\it Step 3.} We claim that for any 
$s\in(0,1)$, if $\omega$ is a pure $\alpha_{\Phi(s)}$-ground state, then
we have $\omega=\omega_\infty$.
Let $(\caH_\infty,\pi_\infty,\Omega_\infty)$ be the GNS triple of $\omega_\infty$.
As $\omega_\infty$ is pure, we have
$\pi_\infty(\caA_{\bbZ})^{''}=\caB(\caH_\infty)$.
From {\it Step 2}, 
$\omega$ is $\omega_\infty$-normal, hence represented by a density matrix on $\caH$.
As furthermore $\omega$ is pure, the density matrix is a one rank projection onto $\bbC\xi$ for some unit vector
$\xi\in\caH_\infty$ i.e., $\omega=\omega_\xi=\braket{\xi}{\pi_\infty\lmk \cdot\rmk\xi}$.
By the $\alpha_{\Phi(s)}$-invariance of $\omega$ and $\pi_\infty(\caA_{\bbZ})^{''}=\caB(\caH_\infty)$, 
we have $e^{itH_{\omega_\infty,\alpha_{\Phi(s)}}}\lmk \ket{\xi}\bra{\xi}\rmk e^{-it H_{\omega_\infty,\alpha_{\Phi(s)}}}= \ket{\xi}\bra{\xi}$
for all $t\in\bbR$.
From this and the strong continuity of the one parameter group of unitatries $e^{itH_{\omega_\infty,\alpha_{\Phi(s)}}}$,
$\xi$ is an eigenvector of $H_{\omega_\infty,\alpha_\Phi(s)}$ with some eigenvalue $\lambda$.
From the positivity of  $H_{\omega_\infty,\alpha_{\Phi(s)}}$, we have $\lambda\ge 0$.
The triple $(\caH_\infty, \pi_\infty, \xi)$ is the GNS triple of  $\omega$ because of $\pi_\infty(\caA_{\bbZ})^{''}=\caB(\caH_\infty)$,
 and
$H_{\omega_\infty,\alpha_{\Phi(s)}}-\lambda$ satisfies
\[
e^{it\lmk H_{\omega_\infty,\alpha_{\Phi(s)}}-\lambda\rmk}\pi_\infty\lmk A\rmk\xi
=\pi_\infty\lmk \alpha_{\Phi(s),t}(A)\rmk\xi,\quad A\in\caA_{\bbZ},\quad t\in[0,1].
\]
Hence we have $H_{\omega,\alpha_{\Phi(s)}}= H_{\omega_\infty,\alpha_{\Phi(s)}}-\lambda$.
As this operator have to be positive, we have
\[
0\le \braket{\Omega_\infty}{H_{\omega,\alpha_{\Phi(s)}}\Omega_\infty}
=\braket{\Omega_\infty}{H_{\omega_\infty,\alpha_{\Phi(s)}}\Omega_\infty}-\lambda
=-\lambda.
\]
Hence we conclude $\lambda=0$, and $\xi$ is in the kernel of $H_{\omega_\infty,\alpha_{\Phi(s)}}$.
From Lemma \ref{lem:kh} we obtain that
$\omega=\omega_\xi\in \caS_{\bbZ}(H_{\Phi(s)})=\widetilde{\caS_{\bbZ,\bbZ}}(H_{\Phi(s)})\subset \widetilde{\caS_{\bbZ,\bbZ}}(H_{\Phi_1})=\caS_{\bbZ}(H_{\Phi_1})=\{\omega_\infty\}$
, i.e., $\omega=\omega_\infty$.
\\
{\it Step 4.} Let $s\in(0,1)$.
As the set of all $\alpha_{\Phi(s)}$-ground states 
${\mathfrak B}_{\Phi(s)}$ is a non-empty $wk*$-compact convex set, 
by Krein-Milman theorem, ${\mathfrak B}_{\Phi(s)}$ is the convex closed hull of its extremal points.
Any extremal point of  is pure ${\mathfrak B}_{\Phi(s)}$ (Theorem 5.3.37 of \cite{BR2}).
Therefore, from {\it Step 3}, $\omega_\infty$ is the unique extremal point of ${\mathfrak B}_{\Phi(s)}$.
This means $\omega_\infty$ is the unique $\alpha_{\Phi(s)}$-ground state.
\end{proof}
We obtain the following corollary as a special case of the previous Lemma.
\begin{cor}\label{lem:frfh}
Let $H_{\Phi}$ be a frustration free Hamiltonian given by a positive interaction $\Phi\in \caJ$, and assume that
$H_{\Phi}$ is gapped with respect to the open boundary conditions.
Assume $\Phi$ satisfies {\it Condition 6}.
Then, the $\alpha_{\Phi}$-ground state is unique.
\end{cor}
\begin{proof}
Set $\Phi_0=\Phi_1=\Phi$ in Lemma \ref{lem:zee}.
\end{proof}
\subsection{Bulk equivalence}
In this subsection, we derive sufficient conditions for the bulk equivalence.
\begin{lem}\label{lem:19}
Let $H_0,H_1$ be gapped  Hamiltonians with respect to the open boundary conditions, associated with interactions
$\Phi_{0},\Phi_{1}\in{\caJ}$.
Assume that $H_{\Phi_0}$ and $H_{\Phi_1}$ are type I-$C^1$-equivalent, and let $\Phi$ be the corresponding path of interactions.
Assume that for each $s\in[0,1]$, $H_{\Phi(s)}$ is frustration-free and $\Phi(s)$
 satisfies {\it Condition 6}.
Then we have $H_{\Phi_0}\simeq_BH_{\Phi_1}$,
\end{lem}
\begin{proof}
By Corollary \ref{lem:frfh},
there exists a unique $\alpha_{\Phi(s)}$-ground state $\omega_s$ for each $s\in[0,1]$.
As $\caS_{\bbZ}\lmk H_{\Phi(s)}\rmk$ is non-empty, this means $\omega_s\in \caS_{\bbZ}\lmk H_{\Phi(s)}\rmk$.
Let $\gamma>0$ be the uniform lower bound of the gap in Definition \ref{def:phafst}.
By Lemma \ref{lem:tl}, 
we then have 
\[
\sigma\lmk H_{\omega_s,\alpha_{\Phi(s)}}\rmk \setminus \{0\}\subset [\gamma,\infty).
\]
By the uniqueness of the ground state, $0$ is a non-degenerate eigenvalue of $ H_{\omega_s,\alpha_{\Phi(s)}}$, for each $s\in[0,1]$.
This proves $\Phi(s)\in\caJ_B$, and $H_{\Phi_0}\simeq_BH_{\Phi_1}$.
\end{proof}
\begin{lem}\label{lem:ugs}Let $\Phi_0,\Phi_1\in \caJ$ be positive interactions.
Assume that the pair $(\Phi_0,\Phi_1)$ satisfies the {\it Condition 8}.
Then we have $\Phi_0,\Phi_1\in\caJ_B$.
Furthermore, for each $s\in[0,1]$ we define $\Phi(s)\in\caJ$ by (\ref{eq:int10}).
Then we have $H_{\Phi_0}\simeq_B H_{\Phi_1}$ with respect to this path.
\end{lem}
\begin{proof}
First we note that for any $s\in[0,1]$, $\omega_\infty$ in {\it 4.} of {\it Condition 7} is the unique $\alpha_{\Phi(s)}$-ground state.:
By  Lemma \ref{lem:zee}, $\omega_\infty$ is the unique $\alpha_{\Phi(s)}$-ground state for each $s\in(0,1)$.
For $s=0,1$, we obtain it by Corollary \ref{lem:frfh}.

Let $(\caH_\infty,\pi_\infty,\Omega_\infty)$ be the GNS triple of $\omega_\infty$.
We have $\pi_{\infty}\lmk \caA_{\bbZ}\rmk^{''}=\caB\lmk \caH_{\infty}\rmk$
because $\omega_\infty$ is pure.
Therefore, for any $s\in[0,1]$,
any unit vector $\eta$ in the kernel of $H_{\omega_\infty,\alpha_{\Phi(s)}}$
orthogonal to $\Omega_\infty$ gives an $\alpha_{\Phi(s)}$-ground state $\braket{\eta}{\pi_\infty(\cdot)\eta}$ which is different from $\omega_\infty$.
Recall that our $\omega_\infty$ is the unique $\alpha_{\Phi(s)}$-ground state.
Therefore, we have $\ker H_{\omega_\infty,\alpha_{\Phi(s)}}=\bbC\Omega_\infty$ for any $s\in[0,1]$.

From the assumption and Lemma \ref{lem:obb}, there exists $\gamma>0$ such that 
\begin{align}\label{eq:b01g}
\sigma\lmk H_{\omega_\infty,\alpha_{\Phi_i}}\rmk\setminus \{0\}\subset [\gamma,\infty),\quad i=0,1.
\end{align}
In particular, we have $\Phi_0,\Phi_1\in \caJ_B$.
In order to complete the proof, it suffices to show 
\begin{align}\label{eq:intpg}
\sigma\lmk H_{\omega_\infty,\alpha_{\Phi(s)}}\rmk \setminus \{0\}\subset [\gamma,\infty)
\end{align}
for any $s\in[0,1]$.
To prove this, it suffices to show that
the spectral projection of $H_{\omega_\infty,\alpha_{\Phi(s)}}$ corresponding to $[0,\lambda]$
is equal to the orthogonal projection
onto $\bbC\Omega_{\infty}$, for all $0<\lambda<\gamma$.
We assume this is not the case, and show a contradiction.
Suppose that there exists a $0<\lambda<\gamma$ and a unit vector $\xi\in\caH_{\infty}$ such that orthogonal to
$\Omega_{\infty}$ and
$\Proj[H_{\omega_\infty,\alpha_{\Phi(s)}}\in [0,\lambda]]\xi=\xi$. (Recall the notation $\Proj$ in Appendix A \cite{Ogata1}.) 
As $\xi$ is in the domain of $H_{\omega_\infty,\alpha_{\Phi(s)}}$ and $\pi_\infty\lmk \caA^{\rm loc}_{\bbZ}\rmk\Omega_{\infty}$ is a core of
$H_{\omega_\infty,\alpha_{\Phi(s)}}$ \cite{BR2}, there exists a net $\{ A_\beta\}_\beta$ in $\caA^{\rm loc}_{\bbZ}$ such that
\begin{align}
\lim_{\beta}\lV \pi_\infty\lmk  A_\beta\rmk \Omega_{\infty}- \xi\rV=0,\quad
\lim_{\beta}\lV H_{\omega_\infty,\alpha_{\Phi(s)}}\pi_\infty\lmk  A_\beta\rmk \Omega_{\infty}- H_{\omega_\infty,\alpha_{\Phi(s)}}\xi\rV=0.
\end{align}
Note that $\omega_\infty( A_\beta)=\braket{\Omega_{\infty}}{\pi_{\infty}\lmk  A_\beta\rmk\Omega_{\infty}}\to \braket{\Omega_{\infty}}{\xi}=0$ from this equation.
Therefore, setting
 $\tilde  A_\beta:= A_\beta-\omega_\infty( A_\beta)$, we obtain
\begin{align}
\lim_{\beta}\lV \pi_\infty\lmk \tilde  A_\beta\rmk \Omega_{\infty}- \xi\rV=0,\quad
\lim_{\beta}\lV H_{\omega_\infty,\alpha_{\Phi(s)}}\pi_\infty\lmk \tilde  A_\beta\rmk \Omega_{\infty}- H_{\omega_\infty,\alpha_{\Phi(s)}}\xi\rV=0.
\end{align}
By the definition, ${\pi_\infty\lmk \tilde  A_\beta\rmk \Omega_{\infty}}$ and
${\Omega_{\infty}}$ are orthogonal. Recall that we have $\ker H_{\omega_\infty,\alpha_{\Phi_i}}=\bbC\Omega_\infty$ $i=0,1$,
and (\ref{eq:b01g}).
From these facts, ${\pi_\infty\lmk \tilde  A_\beta\rmk \Omega_{\infty}}$ belongs to $\Proj[H_{\omega_\infty,\alpha_{\Phi_i}}\in [\gamma,\infty)]\caH$,
$i=0,1$.
Now using the definition of the bulk Hamiltonian \cite{BR2}, we have
\begin{align*}
&\braket{\pi_\infty\lmk \tilde  A_\beta\rmk \Omega_{\infty}}{ H_{\omega_\infty,\alpha_{\Phi(s)}}\pi_\infty\lmk \tilde  A_\beta\rmk \Omega_{\infty}}
=-i\braket{\pi_\infty\lmk \tilde  A_\beta\rmk \Omega_{\infty}}{\pi_\infty\lmk\delta_{\alpha_{\Phi(s)}} \lmk \tilde  A_\beta\rmk \rmk \Omega_{\infty}}\\
&=-i\braket{\pi_\infty\lmk \tilde  A_\beta\rmk \Omega_{\infty}}{\lmk (1-s)\pi_\infty\lmk\delta_{\alpha_{\Phi_0}} \lmk \tilde  A_\beta\rmk \rmk +
s\pi_\infty\lmk\delta_{\alpha_{\Phi_1}} \lmk \tilde  A_\beta\rmk \rmk\rmk\Omega_{\infty}}\\
&=(1-s)\braket{\pi_\infty\lmk \tilde  A_\beta\rmk \Omega_{\infty}}{ H_{\omega_\infty,\alpha_{\Phi_0}}\pi_\infty\lmk \tilde  A_\beta\rmk \Omega_{\infty}}
+s\braket{\pi_\infty\lmk \tilde  A_\beta\rmk \Omega_{\infty}}{ H_{\omega_\infty,\alpha_{\Phi_1}}\pi_\infty\lmk \tilde  A_\beta\rmk \Omega_{\infty}}\\
&\ge\lmk (1-s)\gamma+s\gamma\rmk \lV \pi_\infty\lmk \tilde  A_\beta\rmk \Omega_{\infty}\rV^2
=\gamma\lV\pi_\infty\lmk \tilde  A_\beta\rmk \Omega_{\infty}\rV^2.
\end{align*}
Taking the $\beta\to \infty$ limit, 
from the choice of $\xi$,
we obtain
\begin{align*}
\lambda\ge \braket{\xi}{H_{\omega_\infty,\alpha_{\Phi(s)}}\xi}\ge \gamma,
\end{align*}
which contradict $\lambda<\gamma$.
This proves (\ref{eq:intpg}) for all $s\in[0,1]$.
\end{proof}
\begin{lem}\label{lem:bob}
Let $\bb\in\ClassA$
with respect to $(n_0,k_R,k_L,\lal, \bbD,\bbG,Y)$. 
Let $\oo_{\bbB}\in\Prim(n,n_0)$ be the $n$-tuple given by Lemma 3.2 of
\cite{Ogata1}.
Then we have
$H_{\Phi_{m,\bbB}}\simeq_{B}H_{\Phi_{m,\omega_\bbB}}$,
for $m\ge 2 l_\bbB$.
\end{lem}
\begin{proof}
For any $l\ge l_\bbB$, we have
\[
\caK_l\lmk \oo_\bbB\rmk \otimes \eij{00}
=\lmk \unit_{n_0}\otimes \eij{00}\rmk\caK_l(\bbB)\lmk \unit_{n_0}\otimes \eij{00}\rmk
=\mnz\otimes \eij{00}.
\]
From this, we have $l_{\oo_\bbB}\le l_{\bbB}$, and
$G_{N,\oo_\bbB}\le G_{N,\bbB}$ for all $l_\bbB\le N$.

We apply Lemma \ref{lem:ugs} to $(\Phi_0,\Phi_1)=({\Phi_{m,\bbB}}, {\Phi_{m,\omega_\bbB}})$
with $m\ge 2 l_\bbB$.
We check {\it Condition 7} for $({\Phi_{m,\bbB}}, {\Phi_{m,\omega_\bbB}})$.
The first and second condition of {\it Condition 7} follows from Theorem 1.18 of \cite{Ogata1}.
The third condition follows from $G_{N,\oo_\bbB}\le G_{N,\bbB}$, $l_\bbB\le N$.
From this $G_{N,\oo_\bbB}\le G_{N,\bbB}$ and the frustration freeness, we have 
$\{\omega_{\oo_{\bbB,\infty}}\}=\caS_\bbZ\lmk H_{\Phi_{m, \oo_\bbB}}\rmk\subset \caS_\bbZ\lmk H_{{\Phi}_{m,\bbB}}\rmk =\{\omega_{\bbB,\infty}\}$.
This proves the fourth condition of  {\it Condition 7} .
The fifth condition of  {\it Condition 7} follows from Lemma \ref{lem:24}.
Furthermore, $({\Phi_{m,\bbB}}, {\Phi_{m,\omega_\bbB}})$ satisfies {\it Condition 8} with respect to (ii) by Theorem 1.18 of\cite{Ogata1} and Lemma \ref{lem:24}. 
Hence we may apply Lemma \ref{lem:ugs} and obtain $H_{\Phi_{m,\bbB}}\simeq_{B}H_{\Phi_{m,\omega_\bbB}}$.
\end{proof}

\subsection{Parents Hamiltonian}
In this subsection, we connect $\caJ_{FB}$ to MPS Hamiltonians.
The key ingredient is \cite{Matsui1}.

\begin{lem}\label{lem:vrep}
Let $\Phi\in\caJ_{FB}$.
Then there exist $n_0\in\nan$ and  $\vv\in\Prim(n,n_0)$ such that
\[
H_\Phi\simeq_{B} H_{\Phi_{m, \vv}},\quad m\ge 2l_\vv.
\]
Furthermore, we have $G_{N,\vv}\le G_{N,\Phi}$ for all $N\in\nan$, where 
$G_{N,\Phi}$ denotes the orthogonal projection onto $\ker H_{\Phi,[0,N-1]}$.
\end{lem}
\begin{proof}
Let $\omega$ be the unique $\alpha_\Phi$-ground state and 
$(\caH,\pi,\Omega)$ its GNS triple.
As $\omega$ is the unique $\alpha_\Phi$-ground state and $H_\Phi$ is frustration free, we have $\omega\in\widetilde\caS_{\bbZ,\bbZ}(H_\Phi)=\caS_{\bbZ}(H_\Phi)$.
The uniqueness also implies $\ker H_{\omega,\alpha_\Phi}
=\bbC\Omega$.

By Theorem 4.1 of \cite{NS}, the second condition of Definition \ref{def:jb}
and $\ker H_{\omega,\alpha_\Phi}
=\bbC\Omega$ implies 
$\lim_{N\to\infty}\omega\lmk A\tau_{\pm N}\lmk B\rmk\rmk=\omega\lmk A\rmk\omega\lmk B\rmk$, for all
$A,B\in\caA_{\bbZ}^{\rm loc}$. Therefore, we can apply Theorem 1.2 of \cite{Matsui1} and
$\omega$ is a pure finitely correlated state.
Applying Theorem 1.5 of \cite{fnwpure},
we see that $\omega$ is right-generated by a minimal standard triple (see \cite{fnwpure}
or \cite{Ogata2} Appendix C) 
$(\Mat_{n_0},\bbE,\rho)$ where $\bbE$ is given by an isometry $V:\bbC^{n_0}\to\bbC^n\otimes \bbC^{n_0}$, as
 $\bbE(X)=V^* XV$, $X\in \Mat_n\otimes \mnz$.
Furthermore, from Proposition 3.7 and Proposition 2.4  of \cite{fnwpure},
$\sigma\lmk \bbE_{\unit}\rmk\cap \bbT=\{1\}$, and $1$ is a non-degenerate eigenvalue of
$\bbE_{\unit}$.
The isometry $V$ can be decomposed as
$V\chi=\sum_{\mu=1}^n\psi_\mu\otimes v_{\mu}^*\chi$, $\chi\in\bbC^{n_0}$ with $v_\mu\in\mnz$.
With this notation, we have $T_\vv=\bbE_{\unit}$ and it is a unital CP map.
The state $\rho$ is faithful and $T_\vv$-invariant.
Therefore, $T_\vv$ is primitive and $r_{T_\vv}=1$.
Hence we have $\vv\in \Prim(n,n_0)$. 

We check that  $G_{l,\vv}\le G_{l,\Phi}$ for all $l\in\nan$.
We claim that for any $l\in\nan$, we have $s\lmk \omega\vert_{\caA_{[0,l-1]}}\rmk=G_{l,\vv}$.
For $\xi\in\bigotimes_{i=0}^{l-1}\bbC^n$, we have
\begin{align*}
\omega\lmk \ket{\xi}\bra{\xi}\rmk
=\rho\lmk\lmk\sum_{\mu^{(l)}\in \{1,\ldots,n\}^{\times l}}\braket{\widehat\psi_{\mu^{(l)}}}{\xi} \widehat v_{\mu^{(l)}}\rmk \lmk\sum_{\mu^{(l)}\in \{1,\ldots,n\}^{\times l}}\braket{\widehat\psi_{\mu^{(l)}}}{\xi} \widehat v_{\mu^{(l)}}\rmk ^*\rmk,
\end{align*}
by a straightforward calculation.
From this and the faithfulness of $\rho$, we see that
$\omega\lmk \ket{\xi}\bra{\xi}\rmk=0$ if and only if 
$\sum_{\mu^{(l)}\in \{1,\ldots,n\}^{\times l}}\braket{\widehat\psi_{\mu^{(l)}}}{\xi} \widehat v_{\mu^{(l)}}=0$.
By the definition of $\Gamma_{l,\vv}^{(R)}$, $\sum_{\mu^{(l)}\in \{1,\ldots,n\}^{\times l}}\braket{\widehat\psi_{\mu^{(l)}}}{\xi} \widehat v_{\mu^{(l)}}=0$ if and only if $\xi\in \caG_{l,\vv}^{\perp}$. This proves the claim.
As we have $\omega\in  \widetilde \caS_{\bbZ,\bbZ}\lmk H_{\Phi}\rmk$,
 $\omega\lmk H_{\Phi,[0,l-1]}\rmk=0$ holds for each $l\in\nan$.
From the claim above, it means that $G_{l,\vv}=s\lmk \omega\vert_{\caA_{[0,l-1]}}\rmk\le G_{l,\Phi}$ for all $l\in\nan$.

Now we would like to show $H_\Phi\simeq_{B} H_{\Phi_{m, \vv}}$ if $m\ge 2l_\vv$.
We apply 
Lemma \ref{lem:ugs} to $(\Phi_0,\Phi_1)=(\Phi,\Phi_{m, \vv})$.
We have to check {\it Condition 8} for $(\Phi,\Phi_{m, \vv})$.
{\it 1.2.} of {\it Condition 7} is from $\Phi\in \caJ_{FB}$ and Theorem 1.18 of \cite{Ogata1}
with $\vv\in\Prim(n,n_0)$.
The third condition follows from $G_{l,\vv}\le G_{l,\Phi}$, $l\in\nan$.
As $\omega$ is generated by $\vv$, we have $\caS_{\bbZ}(H_{\Phi_{m,\vv}})=\{\omega_{\vv,\infty}\}=\{\omega\}=\caS_{\bbZ}(H_\Phi)$,
and {\it 4} of {\it Condition 7} follows.
{\it 5.} of {\it Condition 7} follows from Lemma \ref{lem:24}. As our $\Phi$ belongs to $\caJ_{FB}$, (i) of {\it Condition 8} holds.
Hence we may apply Lemma \ref{lem:ugs} , and obtain $H_\Phi\simeq_{B} H_{\Phi_{m, \vv}}$.
\end{proof}

\subsection{Proof of Theorem \ref{thm:bulk}}
Let $\Phi_0,\Phi_1\in \caJ_{FB}$.
By Lemma \ref{lem:vrep}, we obtain $m_i$, $n_{0,i}\in\nan$ and  $\vv_i\in\Prim(n,n_{0,i})$
with $m_i\ge 2l_{\vv_i}$, for each $i=0,1$, such that
$
H_{\Phi_i}\simeq_{B} H_{\Phi_{m_i, \vv_i}}.
$Set
$k_0:=\max\{m_i,\tilde m(n_{0,i},0,0)\}_{i=0,1}$.
(Recall (\ref{eq:mti}).)
We consider the sequence of paths {\it 1.-5.} given in the proof of Theorem \ref{singthm}, for $m=m_i$ $\bbB=\vv_i$, $i=0,1$.
For the paths {\it 1.,4.,5}, {\it Condition 8} holds that we may apply Lemma \ref{lem:ugs} to get $\simeq_{B}$.
For {\it 2,3}, frustration-freeness of the path,  {\it Condition 6} and $\simeq_I$ implies
$\simeq_B$ by Lemma \ref{lem:19}.
Hence we obtain $H_{\Phi_i}\simeq_{B} H_{\Phi_{m_i, \vv_i}}\simeq_{B}H_{\Phi_{k_0,\bbV_{s,n_{0,i},0,0}}}$, $i=0,1$.
By Lemma \ref{lem:bob}, we obtain $H_{\Phi_{k_0,\bbV_{s,n_{0,i},0,0}}}\simeq_{B}H_{\Phi_{k_0,\oo_{\bbV_{s,n_{0,i},0,0}}}}$, $i=0,1$.
Note that $\oo_{\bbV_{s,n_{0,i},0,0}}\in \Prim(n, 1)$, $i=0,1$.
Therefore, as in path {\it 2.} of Theorem \ref{singthm} we obtain 
$H_{\Phi_{k_0,\oo_{\bbV_{s,n_{0,0},0,0}}}}\simeq_{B}H_{\Phi_{k_0,\oo_{\bbV_{s,n_{0,1},0,0}}}}$.
Hence we obtain 
\[
H_{\Phi_0}\simeq_{B} H_{\Phi_{m_0, \vv_0}}\simeq_{B}H_{\Phi_{k_0,\bbV_{s,n_{0,0},0,0}}}
\simeq_{B}H_{\Phi_{k_0,\oo_{\bbV_{s,n_{0,0},0,0}}}}\simeq_{B}H_{\Phi_{k_0,\oo_{\bbV_{s,n_{0,1},0,0}}}}
\simeq_{B}
H_{\Phi_{k_0,\bbV_{s,n_{0,1},0,0}}}\simeq_{B} H_{\Phi_{m_1, \vv_1}}\simeq_{B} H_{\Phi_1},
\]
proving the Theorem.
\section*{Acknowledgements}
{
This work was supported by JSPS KAKENHI Grant Number 25800057 and 16K05171. 
}
\appendix
\section{$C^\infty$-path of linear independent vectors}
\begin{lem}\label{li}
Let $k,m\in\nan$ with $m\le k$.
Let $\zeta_i:[0,1]\to\cc^k$, $i=1,\ldots, m$ be $C^\infty$-maps.
For each $t\in[0,1]$, let $P(t)$ be the orthogonal projection onto $\spa\{\zeta_i(t)\}_{i=1}^m$.
Suppose that for each $t\in[0,1]$,
the vectors $\{\zeta_i(t)\}_{i=1}^m$ are linearly independent.
Then there exist $C^\infty$-maps $\eta_i:[0,1]\to \cc^k$, $i=1,\ldots,m$
such that 
$\{\eta_i(t)\}_{i=1}^m$ is a CONS of $P(t)\cc^k$, for each $t\in[0,1]$.
In particular, 
the map $[0,1]\ni t\mapsto P(t)\in \mk$ is $C^\infty$.
\end{lem}
\begin{proof}
This is immediate by Gram-Schmidt orthogonalization.
\end{proof}
\begin{lem}\label{ap}
Let $k,m\in\nan$ with $m<k$ and $\caK$ an $m$-dimensional subspace of $\cc^k$.
Let $P_{\caK}$ be the orthogonal projection onto $\caK$.
Let $\eta:[0,1]\to \cc^k$ be a $C^\infty$-map with 
$\eta(0),\eta(1)\notin \caK^{\perp}$.
Then for any $\varepsilon>0$,
there exists  a $C^\infty$-map $\xi:[0,1]\to \cc^k$ 
such that $\sup_{t\in[0,1]}\lV \xi(t)-\eta(t)\rV
<\varepsilon$, $\xi(0)=\eta(0)$, $\xi(1)=\eta(1)$ and $P_{\caK}\xi(t)\neq 0$ for all $t\in[0,1]$.
\end{lem}
\begin{proof}
This is immediate from Lemma A.1 of \cite{bo} for a sub-manifold $\caN=\caK^{\perp}$.
\end{proof}
\begin{lem}\label{un}
Let $k\in\nan$, $-\infty<a<b<\infty$ and $t_0\in[a,b]$.
Let
$P:[a,b]\to \mk$ be a $C^{\infty}$-map
such that $P(t)\in\caP(\mk)$ for each $t\in[a,b]$
and $\sup_{t\in[a,b]}\lV P(t)-P(t_0)\rV<\frac 14$.
Then there exists a $C^\infty$-map $U:[a,b]\to \mk$
with $U(t_0)=1$ such that $U(t)\in \caU(\mk)$ and 
$P(t)=U(t)P(t_0)U(t)^*$ for each $t\in[a,b]$.
\end{lem}
\begin{proof}
Set $X(t):=P(t)P(t_0)+\lmk 1-P(t)\rmk \lmk 1-P(t_0)\rmk$ for $t\in[a,b]$.
Then we have $\lV X(t)-1\rV<\frac 14$ and can define
$U(t):=X(t)\lv X(t)\rv^{-1}$ as $X(t)$ is invertible.
 By definition, $U(t)$ is unitary for all $t\in[a,b]$
 and satisfies $U(t_0)=1$.
 Furthermore, $U:[a,b]\to \mk$ is $C^\infty$.
 We claim $U(t) P(t_0)U(t)^*=P(t)$
 for all $t\in[a,b]$.
 To see this, note that $X(t)P(t_0)=P(t)X(t)$. From this and its adjoint $P(t_0)X(t)^*=X(t)^*P(t)$, we also have 
 $\lv X(t)\rv^{-1}P(t_0)=P(t_0)\lv X(t)\rv^{-1}$.
 Hence we have $U(t)P(t_0)U(t)^*=X(t)\lv X(t)\rv^{-1}P(t_0)\lv X(t)\rv^{-1}X(t)^*
=P(t)$.
\end{proof}
\begin{lem}\label{up}
For $k\in\nan$, let $P:[0,1]\to \mk$ be a $C^{\infty}$-map
such that $P(t)\in\caP(\mk)$ for each $t\in[0,1]$.
Then there exists a continuous and piecewise $C^\infty$-map $U:[0,1]\to \mk$
with $U(0)=1$ such that $U(t)\in \caU(\mk)$ and 
$P(t)=U(t)P(0)U(t)^*$ for each $t\in[0,1]$.
\end{lem}
\begin{proof}
By the continuity of $P$ and the compactness of $[0,1]$, there exist $l\in \nan$ and $0=t_0<t_1<\cdots<t_l=1$
such that $\sup_{t\in[t_i,t_{i+1}]}\lV P(t)-P(t_i)\rV<\frac 14$, $i=0,\ldots,l-1$.
By Lemma \ref{un}, there exist  $C^\infty$-maps $U_i:[t_i,t_{i+1}]\to \mk$
$i=0,\ldots,l-1$
with $U_i(t_i)=1$, such that $U_i(t)\in\caU(\mk)$, $P(t)=U_i(t)P(t_i)U_i(t)^*$
for $t\in[t_i,t_{i+1}]$.
Set $U(t):=U_{i-1}(t)U_{i-2}(t_{i-1})\cdots U_1(t_2)U_0(t_1)$
for $t\in [t_{i-1},t_i]$. Then $U(t)$ satisfies the required conditions.
\end{proof}
\begin{lem}\label{lem:cio}
Let $k,m\in\nan$ with $m<k$.
Let $\zeta_i:[0,1]\to\cc^k$, $i=1,\ldots, m$ be $C^\infty$-maps.
For each $t\in[0,1]$ let $P(t)$ be the orthogonal projection onto $\spa\{\zeta_i(t)\}_{i=1}^m$.
Suppose that for each $t\in[0,1]$,
the vectors $\{\zeta_i(t)\}_{i=1}^m$ are linearly independent.
Let $\xi_0,\xi_1\in \cc^k$ with $(1-P(0))\xi_0,(1-P(1))\xi_1\neq 0$.
Then there exists a continuous and piecewise $C^\infty$-map $\xi:[0,1]\to\cc^k$
with $\xi(0)=\xi_0$, $\xi(1)=\xi_1$ such that $(1-P(t))\xi(t)\neq 0$
for all $t\in[0,1]$.
\end{lem}
\begin{proof}
By Lemma \ref{li}, $[0,1]\ni t\mapsto P(t)\in \mk$ is $C^\infty$.
Therefore, applying Lemma \ref{up} to $P(t)$
we obtain a continuous and piecewise $C^\infty$-map 
$U:[0,1]\to\caU(\mk)$ with $U(0)=1$
such that $P(t)=U(t)P(0)U(t)^*$ for each $t\in[0,1]$.
Set $\eta_0=\xi_0$, and $\eta_1:=U(1)^*\xi_1$.
By the assumptions, $(1-P(0))\eta_0\neq  0$ and $(1-P(0))\eta_1\neq 0$.
 Let $\tilde\eta(t):=(1-t)\eta_0+t\eta_1$.
Applying Lemma \ref{ap} to this $\tilde \eta$
and $(1-P(0))\cc^k$
we obtain 
  a $C^\infty$-map $\tilde \xi:[0,1]\to \cc^k$ 
such that  $\tilde \xi(0)=\eta_0$, $\tilde \xi(1)=\eta_1$ and $\caK:=(1-P(0))\tilde \xi(t)\neq 0$ for all $t\in[0,1]$.
Set $\xi(t):=U(t)\tilde\xi(t)$ for $t\in[0,1]$.
Then $\xi:[0,1]\ni t\mapsto \xi(t)\in\cc^k$ is continuous and piecewise $C^{\infty}$.
Furthermore, we have 
\[
\lmk 1-P(t)\rmk \xi(t)=U(t)\lmk 1-P(0)\rmk U(t)^*U(t)\tilde\xi(t)
=U(t)\lmk 1-P(0)\rmk \tilde\xi(t)\neq 0,
\]
and $\xi(0)=\xi_0$, $\xi(1)=\xi_1$.
\end{proof}
The proof of the following Lemma is standard.
\begin{lem}\label{poa}
Let $k,m\in\nan$ with $m<k$.
Let $
X:[0,1]\to (\Mat_k)_+
$
be continuous and piecewise $C^\infty$-path of positive matrices such that
the rank of $X(t)$ is $m$ for all $t\in[0,1]$.
Let $S(t)$ be the support projection of $X(t)$,
and set $\gamma(t):=\drr (\sigma(X(t))\setminus\{0\},\{0\})$.
 Then, the path of projections
 \[
 S:[0,1]\ni t\mapsto S(t)\in \Mat_k
 \]
 is continuous and piecewise $C^1$
 and 
 \[
 \inf_{t\in[0,1]}\gamma(t)>0.
 \]
\end{lem}

\section{CP maps of matrix algebras}
In this section we collect known results about positive maps on matrix
algebras. 
\begin{lem}\label{lem:qw}
Let $n,k\in\nan$
and $\vv\in\mk^{\times n}$.
Suppose that there exist $l_0,d\in\nan$,
such that
\begin{description}
\item[(i)]
$\kl{l_0}(\vv)$ has an invertible element, and
\item[(ii)]
there exists $l_1\in\nan$
such that $\dim\kl{l}(\vv)=d$
for all@$l\ge l_1$.
\end{description}
Then we have
\[
\min\left
\{l\in\nan\mid\dim\kl{l'}(\vv)=d,
\text{for all} \; \; l'\ge l
\right\}
\le
dl_0.
\]\end{lem}

\begin{proof}
The proof is basically in \cite{Sanz:2010aa}.
\end{proof}
\begin{lem}{\cite{Sanz:2010aa}}\label{lem:oqw}
Let $n,n_0\in\nan$ and $\oo\in\Primz(n,n_0)$.
Define
\[
l_{\oo}:=
\inf\{l\in\nan\mid \kl{l'}(\omega)=\mnz,\quad\text{for all } l'\ge l\}.
\]
Then 
we have
$l_{\omega}\le n_0^4$.
\end{lem}

The following statement is standard.
\begin{lem}\label{lem:cinf}
Let $k\in\nan$ and $T:[0,1]\to B(\mk)$
be a $C^{\infty}$-
map satisfying the following conditions for each $t\in[0,1]$.:
\begin{description}
\item [(1)]The spectral radius $r_{T(t)} $ is strictly positive, and
\item[(2)]
 the spectral radius $r_{T(t)} $ is a non degenerate eigenvalue of $T(t)$, and
 \item[(3)]
 there exists $0<s_t<r_{T(t)} $ such that $\sigma(T(t))\setminus \{r_{T(t)}\}\subset\caB_{s_t}(0)$.
\end{description}
Then 
\begin{enumerate}
\item the map
$[0,1]\ni t\mapsto P_{\{r_{T(t)}\}}^{T(t)}\in  B(\mk)$ is $C^\infty$, and
\item
$[0,1]\ni t\mapsto {r_{{T(t)}}}\in\cc$ is $C^\infty$, and
\item
there exists $0<s<1$ such that 
$\sigma(r_{T(t)}^{-1}T(t))\setminus \{1\}\subset\caB_{s} (0)$, for all $t\in[0,1]$.
\end{enumerate}
\end{lem}

\section{Linear spaces spanned by given set of vectors}

\begin{lem}\label{lem:lb}
Let $m,k\in\nan$, with $m\le k$, and
$\{\xi_i\}_{i=1}^m$, a set of 
vectors of $\cc^k$.
Let $A$ be an $m\times m$ matrix
given by
$A=\lmk\braket{\xi_i}{\xi_j}\rmk_{i,j=1}^m$.
Let $X:=\sum_{i=1}^m\ket{\xi_i}\bra{\xi_i}\in\Mat_k$
and $P$ be the support projection of $X$.
Suppose that there  exists a positive constant
$c$ such that
$
c\unit\le A$.
Then we have
$
cP\le X$.In particular, we have
$\sigma(X)\setminus \{0\}\subset [c,\lV X\rV]$.
\end{lem}


\begin{lem}\label{lem:pq}
Let $0<a_1<a_2$.
There is a constant $C_{a_1,a_2}, C_{a_1,a_2}'>0$ which depend only on
$a_1,a_2$ satisfying the following.: 
Let $k,m\in\nan$ with $m\le k$,
$\{\xi_i\}_{i=1}^m$, $\{\eta_i\}_{i=1}^m$
sets of linearly independent vectors in $\cc^k$.
Let $P,Q$ be orthogonal projections onto the subspace
spanned by $\{\xi_i\}$, $\{\eta_i\}$, respectively.
Assume that the spectrum of $X=\sum_{i=1}^m\ket{\xi_i}\bra{\xi_i},
Y=\sum_{i=1}^m\ket{\eta_i}\bra{\eta_i}
\in\Mat_k$ satisfy
$\sigma(X)\setminus \{0\},\sigma(Y)\setminus\{0\}\subset (a_1,a_2)$, and
$\lV X-Y\rV<C_{a_1,a_2}$.
Then we have
\[
\lV P-Q\rV\le C_{a_1,a_2}'\lV X-Y\rV.
\]
\end{lem}

\section{Quasi-equivalence of states}
\begin{lem}\label{lem:qe}
Let $d\in\nan$, and $\mathfrak A$ be a unital $C^*$-algebra.
Suppose that for any $N\in\nan$, there exists a unital $C^*$-algebra  ${\mathfrak B}_N$
such that $\mathfrak A=\lmk \bigotimes_{i=1}^N\Mat_d\rmk\otimes 
{\mathfrak B}_N$.
Let $\varphi_1,\varphi_2$ be states on $\mathfrak A$.
Suppose that there exist constants $C>0$, $0<s<1$ such that
\begin{align*}
\lv
\varphi_1(A)-\varphi_2(A)
\rv
\le C s^N\lV A\rV,\quad N\in\nan,\quad A\in \unit_{\bigotimes_{i=1}^N\Mat_d}\otimes {\mathfrak B}_N.
\end{align*}
Then $\varphi_1$ and $\varphi_2$ are quasi-equivalent.
\end{lem}
\begin{proof}
The proof is analogous to that of Lemma 6.2.55 of \cite{BR2}.
We omit the details.
\end{proof}

\end{document}